\documentclass[reqno,keywordsasfootnote,11pt]{article}

\usepackage[english]{babel}
\usepackage[latin1]{inputenc}
\usepackage{amsmath,amsthm,amsfonts,amssymb,times}
\numberwithin{equation}{section}
\usepackage[final]{epsfig}
\usepackage{color}

\usepackage{mathrsfs}

\usepackage{enumerate}

\def\titlerunning#1{\gdef\titrun{#1}}
\makeatletter
\def\author#1{\gdef\autrun{\def\and{\unskip, }#1}\gdef\@author{#1}}
\def\address#1{{\def\and{\\\hspace*{18pt}}\renewcommand{\thefootnote}{}%
\footnote {#1}}%
\markboth{\autrun}{\titrun}}
\makeatother
\def\email#1{e-mail: #1}
\def\subjclass#1{{\renewcommand{\thefootnote}{}%
\footnote{\emph{Mathematics Subject Classification (2010):} #1}}}
\def\keywords#1{\par\medskip
\noindent\textbf{Keywords.} #1}

\newtheorem{theorem}{Theorem}[section]
\newtheorem{corollary}[theorem]{Corollary}
\newtheorem{lemma}[theorem]{Lemma}
\newtheorem{proposition}[theorem]{Proposition}

\theoremstyle{definition}
\newtheorem{definition}[theorem]{Definition}

\numberwithin{equation}{section}

\usepackage{graphicx,pgf}

\makeatletter
\let\Im\undefined
\let\Re\undefined
\DeclareMathOperator{\Im}{Im \,}
\DeclareMathOperator{\Re}{Re \,}

\DeclareMathOperator{\tr}{tr}
\DeclareMathOperator{\E}{\mathbb{E}}
\DeclareMathOperator{\dist}{dist}
\DeclareMathOperator{\supp}{supp}

\newcommand{\N} {{\mathbb N}}
\newcommand{\C}{\mathbb{C}}
 
\newcommand{\G}{\mathcal{G}}
\newcommand{\R}{\mathbb{R}}

\def\Ev#1{{\mathbb E}\left(#1 \right)}
\def\Pr#1{{\mathbb P}\left(#1 \right)}

\def\Prs#1{{\mathbb P}_s\left(#1 \right)}

\def\Avus#1{ Av_u^{(s)}\left( #1 \right)  }
 
\def\vP{\varPsi} 
\def\T{{\mathcal T}} 

\DeclareMathOperator{\indfct}{\rm 1}

\def\be{\begin{equation}}
\def\ee{\end{equation}} 
\def\Pr#1{ {\mathbb P}{\left(   #1 \right)}}

\usepackage{hyperref}

\makeatother

\begin{document}




\titlerunning{Resonant delocalization}

\title{\Huge Resonant Delocalization \\ for 
Random Schr\"odinger Operators  \\ on Tree Graphs \\[2ex] 
 } 

\author{Michael Aizenman
\and 
Simone Warzel}

\date{\small December 14, 2011 }

\maketitle

\address{M. Aizenman: Depts.\ of Physics and Mathematics,  Princeton University,  Princeton NJ 08544, USA  
\and S. Warzel: Zentrum Mathematik, TU M\"unchen, 
 Boltzmannstr. 3, 85747 Garching, Germany; \email{warzel@ma.tum.de} (corresponding author)
 }



\begin{abstract} 
We analyse the spectral phase diagram of Schr\"odinger operators $ T +\lambda V$ on regular tree graphs, with $T$ the graph adjacency operator and $V$ a  random potential given by \emph{iid} random variables.   
The main result is a criterion for the 
emergence of absolutely continuous (\emph{ac}) spectrum due to  fluctuation-enabled resonances between distant sites.   Using it we prove that for  unbounded random potentials \emph{ac} 
spectrum appears at arbitrarily weak disorder $(\lambda \ll 1)$ in an energy regime which extends beyond the spectrum of~$T$.   Incorporating considerations of the Green function's  large deviations we obtain an extension of the  criterion  which indicates that, under  a yet unproven regularity condition of the large deviations' 'free energy function',  the regime of pure  \emph{ac} spectrum is complementary to that of previously proven localization.  For bounded potentials we disprove the existence at weak disorder  of a mobility edge beyond which the spectrum is localized.


\vspace*{-8.5cm}
\begin{center}
\emph{Dedicated to Hajo Leschke on the occasion of his 66th birthday. }
\end{center}
\end{abstract}

\newpage

\setcounter{tocdepth}{2} 
{\small
\tableofcontents 
}
\vfill


\setcounter{footnote}{0}
\section{Introduction}   

\subsection{The article's topic}

The subject of this work are the spectral properties of random  self-adjoint operators in the Hilbert space~$ \ell^2(\T) $ associated with the vertex set $\T$ of a regular rooted tree graph of a fixed branching number $K > 1 $. The operators take the form
\be  \label{eq:O}
H_\lambda(\omega)  \ = \ T + \lambda \, V(\omega)    \, , 
\ee 
with $T$ the adjacency matrix and $V(\omega)$ a  random potential, i.e., a multiplication operator which is specified by a collection of random variables indexed by~$ \T $.  For simplicity we focus on the case of  independent identically distributed (\emph{iid}) random variables of absolutely continuous distribution, $\varrho(v) \, dv$.  The strength of the disorder is expressed through the parameter~$ \lambda \geq 0 $.  Some of the results presented below will be formulated for unbounded random potentials, in which case the support of the distribution of $V(x)$ is assumed to be the full line.   For other results we assume that the range of values of  $V(x)$ is the interval~$ [-1,1]$.  

It is well known that random Schr\"odinger operators, of which the above tree version is a relatively more approachable example,  exhibit regimes of spectral and dynamical localization where the operator's spectrum consists of a dense collection of eigenvalues with localized eigenfunctions (cf.~\cite{CL,PF,stoll1,Ki}).   However, it still remains an outstanding mathematical challenge to elucidate the conditions for the occurrence of continuous spectrum, and in particular {\it absolutely continuous}  (henceforth called `$ac$') spectrum, in the presence of homogeneous disorder.   The significance of the $ac$ spectrum from the scattering perspective, or a schematic conduction experiment, is  illustrated in  Figure~\ref{fig:wire}.   In the operator's $(E,\lambda)$ {\it phase diagram}, the boundary separating the regime of  localization from the regime of {\it continuous spectrum}, assuming such is found,  is referred to as the {\it mobility edge}~\cite{A}.

The results presented here focus on a  new resonance-driven mechanism by which $ac$ spectrum occurs for  operators such as  $H_\lambda(\omega)$ in the setup described above.   Following is a summary of the main points.
\begin{enumerate} 
\item[1.]  A new sufficiency criterion is derived for $ac$ spectrum on tree graphs in terms of a related Lyapunov exponent.   
\end{enumerate} 
The guiding observation for $1.$ is that  localized modes join into 
extended states when their energy differences  are smaller that the corresponding tunneling amplitudes.  The latter decay exponentially in the distance at the rate whose typical values is given by the Lyapunov exponent.   Hence the probability of a mixing resonance between  localized modes at specified location  is exponentially small.  However, when the volume of the relevant configuration space increases exponentially resonances will be found, and delocalization  prevails.   
This criterion is particularly applicable at weak and moderate disorder.  It is applied here for two results, which apply separately for bounded and for unbounded random potentials: 
\begin{enumerate} 
\item[2.]  For unbounded potentials we  show that $ac$ spectrum appears 'discontinuously' at arbitrarily weak disorder in regimes with very low density of states (of {\it Lifshits tail} asymptotic falloff).  This answers a  puzzle which has been open since the earlier works on the subject~\cite{AAT,AT} concerning the location of the mobility edge and the nature of the continuous spectrum below it.  
\item[3.]   For bounded random potentials it is shown that  at weak disorder there is no mobility edge beyond which the  states are localized.  This has the surprising implication that for this case the standard picture of the phase diagram needs to be corrected.  
\end{enumerate} 
  In essence, $2.$ and $3.$ show that while in one dimension arbitrary weak  level of disorder yields localization, on trees the $ac$ spectrum is quite robust.  
  \begin{enumerate} 
\item[4.]  Extending the analysis which yields the criterion $1.$ through considerations of the Green function's large-deviations, we obtain an improved sufficiency criterion for $ac$ spectrum which appears to be complimentary to the previously derived criterion for localization.  To reduce technicalities, the derivation of the extended criterion is limited to unbounded potentials with support in $\R$.  
\end{enumerate}   
The last point is an  indication that the  mechanism which is discussed here is in essence the relevant one, in the tree setup.  
   
    \begin{figure}
\begin{center}
\includegraphics[width=0.4\textwidth]{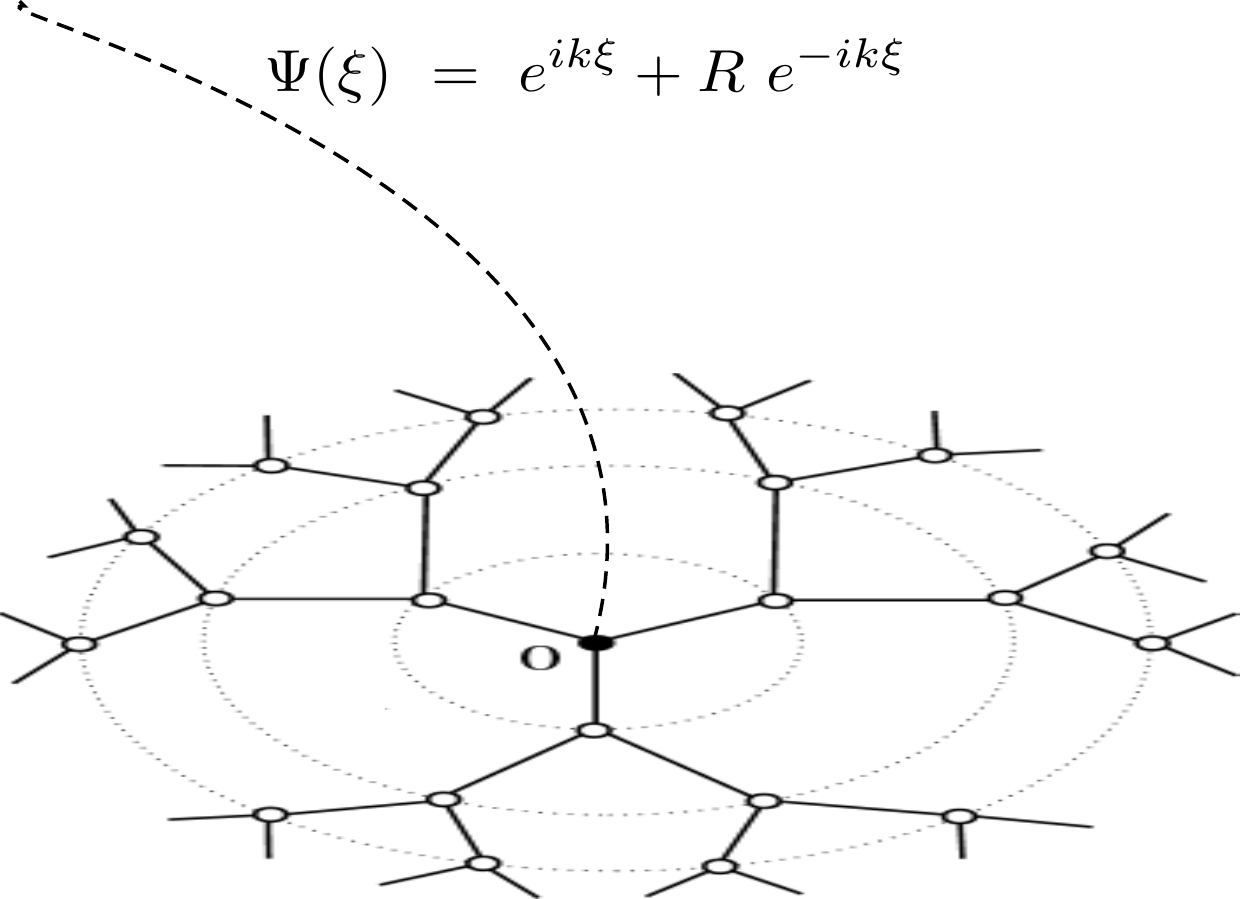}
\caption{A model setup for quantum conduction through the graph (after \cite{MD}): particles are sent  at energy $ E=k^2 +U_{\rm wire} $ down  a wire which is attached to the graph at $ x=0 $.   
In the stationary state  the particles' wave function  is described along the wire by the combination of plane waves 
$ e^{ik_E\xi} + R_E \, e^{-ik_E\xi} $, and along the graph it is given by a decaying solution of the Schr\"odinger equation. 
The natural matching conditions relate the reflection coefficient $ R_E $ to the Green function, and it is found that
$|R_E| < 1 $ exactly if $\Im \langle \delta_0, (H_\lambda-E-i 0)^{-1} \delta_0 \rangle \neq 0$, which is also the condition for $E$ to be in the support of the \emph{ac} spectrum of 
$H_\lambda$.     
    }
\label{fig:wire} 
\end{center}
\end{figure}

A physics-oriented summary of the results $2.$ and $3.$ was given in~\cite{AiWa_2010a} and, correspondingly, \cite{AiWa_math2010}.  Our purpose here is to provide the detailed derivation of the above statements.   In the proof we do not present the direct construction of extended states, but  instead focus on  properties of  the Green function which in essence convey the same information.


\subsection{Past results and the questions settled here}

\subsubsection{The deterministic spectrum}

By a simple calculation, cf.~\eqref{eq:quadrG},\footnote{Even though the graph $\T$ is of constant degree $(K+1)$, except at the root, the spectrum of $T$  does not extend to  $[-(K+1),(K+1)] $.   This  is related to the graph's exponential growth, more precisely to the positivity of its Cheeger constant.   Nevertheless, this  larger set does describe the operator's $ \ell^\infty $-spectrum.} 
\be \label{sigmaT}
\sigma(T) \ = \ [-2\sqrt{K},2\sqrt{K}] \, .   
\ee  
For ergodic random potentials, a class which includes the \emph{iid} case, the spectrum of $H_\lambda(\omega) = T + \lambda V(\omega)$ is almost surely given by a non-random set,  which under the present assumptions is~\cite{CL,PF,Ki}:
\begin{eqnarray}    \label{eq:spec}
\sigma(H_\lambda) &=& \sigma(T) + \lambda \, {\rm supp}~\rho \, .  
\end{eqnarray}
Thus, as the strength of the disorder is increased from $\lambda = 0$ upward:
\begin{enumerate} 
\item  In the unbounded case, of  potentials with $\supp \varrho = \R$, the spectrum of $ H_\lambda(\omega)$ changes discontinuously from an interval to the full line.
\item In the bounded case  the spectrum changes continuously, spreading at a  linear rate   which equals $1$ if $\supp~\varrho = [-1,1]$. 
\end{enumerate} 

 The determination of the nature of the spectral measures whose support spans  $\sigma(H_\lambda) $ requires however a more detailed consideration. 
 The spectral  analysis 
proceeds through the study of the  corresponding Green function 
\be \label{def:Green}
G_\lambda(x,y;\zeta,\omega) := \left\langle \delta_x , \left( H_\lambda(\omega) - \zeta \right)^{-1} \delta_y \right\rangle \, , 
\ee 
where  
$ \zeta \in \C^+ := \left\{ \zeta \in \C \, | \, \Im \zeta > 0 \right\} $ and $\delta_x\in \ell^2(\T)$ is the Kronecker function localized at~$x\in \T$.   
In particular,  the spectral measure $ \mu_{\lambda,\delta_x}(\cdot;\omega) $ associated with $H_\lambda(\omega)$ and $ \delta_x\in\ell^2(\T) $
is related to the Green function through the Stieltjes transform:
\begin{equation}\label{eq:specrep}
	G_\lambda(x,x;\zeta,\omega) \ = \ \int \frac{ \mu_{\lambda,\delta_x}(du;\omega) }{u -\zeta} \, . 
\end{equation}
Of particular interest is the limiting value 
$
 G_\lambda(x,x;E+i0,\omega) \ := \ \lim_{\eta \downarrow 0 } \,  G_\lambda(x,x;E+i\eta,\omega)  $, 
 which exists  for  almost every $E\in \R$ 
 (by the general theory of the Stieltjes transform~\cite{D,CL,PF}).  

The different spectra of $ H_\lambda(\omega) $ are associated with the Lebesgue decomposition of the measures   $ \mu_{\lambda,\delta_x}(\cdot;\omega) $ into their different components: pure point~($pp$), singular continuous~($sc$), and absolutely continuous~($ac$), not all of which need be present. 
Ergodicity, combined with the proof of equivalence of the local measures~\cite{JakLast,JL},  implies that the supports of the different components of $\mu_{\lambda,\delta_x}(du;\omega) $ are also almost surely non-random~\cite{CL,PF,Ki}, and coincide for all $x\in \T$.

The spectral characteristics are related to the {\it dynamical} properties of the unitary time evolution generated by $H_\lambda(\omega) $ (cf.\  the RAGE  theorem in~\cite{stoll1,Ki}) and to questions of conduction.

The absolutely continuous component of  $\mu_{\lambda,\delta_x}(\cdot;\omega) $ is given by 
\be 
 \mu^{(ac)}_{\lambda,\delta_x} (du;\omega) \ = \  \pi^{-1}\Im G_\lambda(x,x;u+i0,\omega) \,  du  \, , 
  \ee 
which is not zero provided the non-negative function satisfies
$ \Im G_\lambda(x,x;E+i0,\omega) \neq 0 $
 on  a positive measure set of  energies.    
As noted in \cite{MD,ASWb},  this condition is equivalent also to the statement  that current which is injected  coherently at energy $E$ down a  wire attached at a site $x $ will be conducted through the graph to infinity, see Figure~\ref{fig:wire}.

Another possible behavior is {\it localization}:
\begin{definition} \label{def:loc}
The  operator $H_\lambda(\omega)$ associated with a metric graph (not necessarily a tree) is said to exhibit: 
\begin{enumerate}[i.]
\item \emph{spectral localization} in an interval $ I \subset \R $ if the  spectral measures $ \mu_{\lambda,\delta_x}(\cdot;\omega) $ associated to $ \delta_{x} \in \ell^2(\mathcal{T}) $ are  almost surely all of only pure-point type in $ I $.
\item  \emph{exponential dynamical localization} in $ I$ if for all $ x \in \T $ and $ R > 0 $ sufficiently large:
	 	 \begin{equation}\label{def:exdynloc1}
	 	\sum_{\substack{ y\in \T:\\  \dist(x,y) = R}}  \E \left(\sup_{t\in \R} |\langle \delta_x\, , \, P_I(H_\lambda) \, e^{-itH_\lambda} \, \delta_y \rangle |^2\right) \ \leq \ C_{\lambda}\,  e^{-\mu_{\lambda}(I) \, R}  \, ,
	 \end{equation}
 at some $ \mu_{\lambda}(I) > 0 $, and  $ C_\lambda < \infty $, with $ \E[\cdot] $ denoting the average with respect to the underlying probability measure.
 \end{enumerate}
\end{definition}
	 
 For a particle which  is initially placed at $ x \in \T $ the left side of \eqref{def:exdynloc1} provides an upper bound on the probability to be found a time $t$ later at  distance $R$ from $x$, under the quantum mechanical time-evolution generated by $ H_\lambda $ restricted to states with energies in~$ I $.     Dynamical localization is the stronger of the two statements. By known arguments (i.e., the Wiener and RAGE theorem, cf.~\cite{Ki,stoll1})  it  implies also the spectral localization.

\subsubsection{Unbounded random potentials}

The  spectral `phase diagram'  
of the operators considered here was studied already in the early works of Abou-Chacra, Anderson and Thouless~\cite{AAT,AT}.  Arguments and numerical work presented in~\cite{AT}  led the authors to surmise that for (centered) unbounded random potentials, the mobility edge, which separates the localization regime from that of continuous spectrum, 
 exists at a location  which roughly corresponds to the outer curve in Figure~\ref{fig1}.   Curiously,  for $\lambda \downarrow 0$ that line approaches energies $|E|=K+1$ which is not the edge of the spectrum
 of the limiting operator $T$.

\begin{figure} 
\begin{center}
\includegraphics[width=0.7\textwidth]{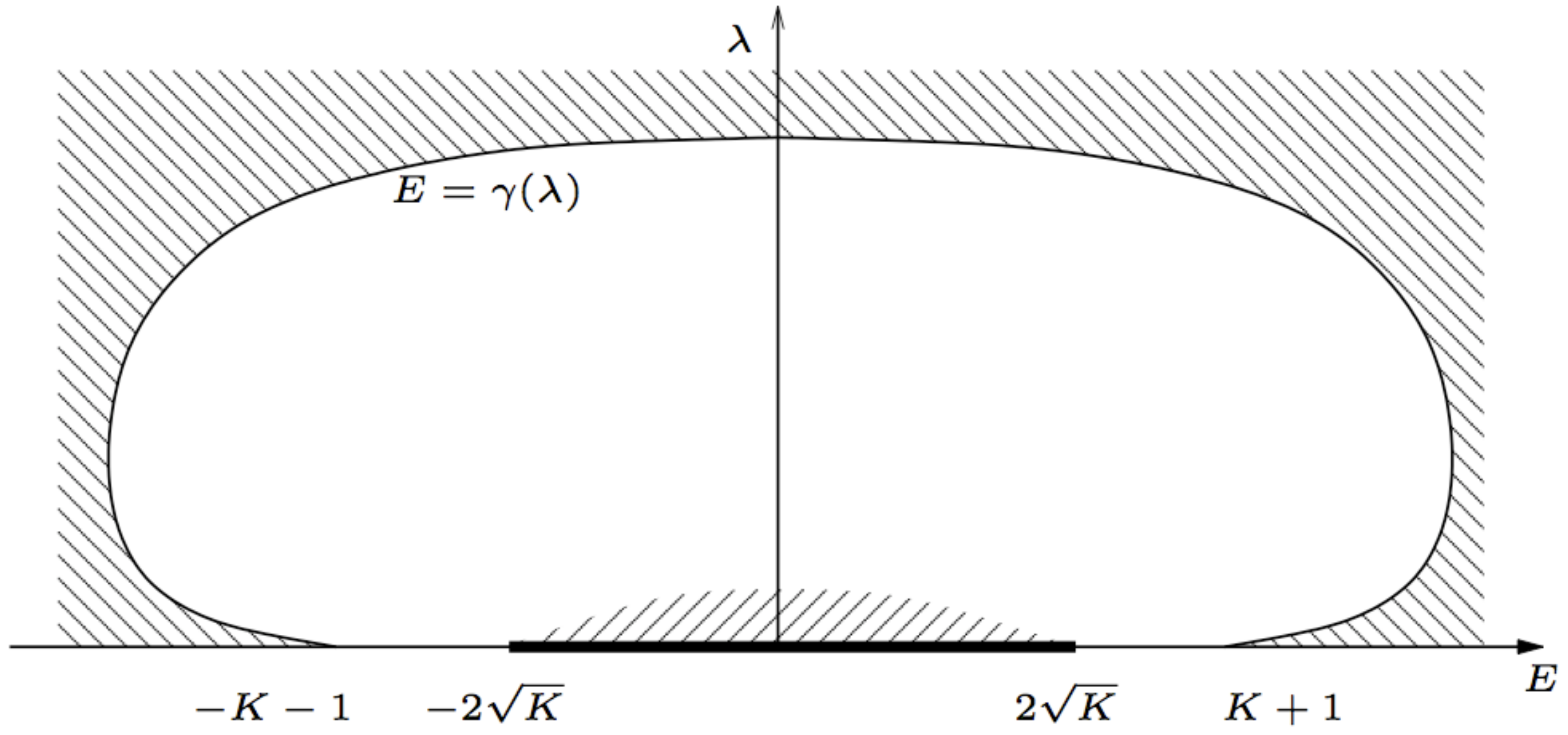}  
\caption{A sketch  of the previously known parts of the phase diagram for unbounded potentials.  The outer region is of proven localization, the smaller hatched region is of proven  delocalization.   The new result extends the latter up to the outer curve, assuming  $\varphi_\lambda(1;E) =- \log K  $ holds only along a line.   The intersection of the curve with the energy axis is stated exactly, while in other details the depiction is only schematic.  
}\label{fig1} 
\end{center}
\end{figure}

Rigorous results for the above class of operators have established  the existence of a  localization regime  and of regions of ac spectrum, leaving however a gap in which neither analysis applied.   More specifically,  the following was proven for the class of operators described above (under assumptions which are somewhat more general than the conditions A-D below):  

\begin{description}
 \item[Localization regime~\cite{AM,A_wd}:]  For any unbounded random potential  with $\supp \rho = \R$, whose probability distribution satisfies also a mild regularity condition, there is a regime of energies of the form: $|E|> \gamma(\lambda)$, with   
\be \lim_{\lambda \downarrow 0} \gamma(\lambda) = K+1 \, , 
\ee  
where with probability one,  $H_\lambda(\omega)$ has only pure point spectrum, and where it also exhibits dynamical localization. 

\item[Extended states / continuous spectrum~\cite{K,K2,ASW,FHS}:] 
For energies $|E|<2\sqrt{K}$ and at weak enough disorder, i.e. $|\lambda|< \widehat \lambda(E)$ (with $\widehat \lambda(E) \downarrow 0$ for $|E| \to  2\sqrt{K}$),  the operator's spectrum is almost surely (purely) absolutely continuous.   
\end{description}

Thus, the previous results have covered two regimes whose boundaries, sketched in Figure~\ref{fig1}, do not connect.   Particularly puzzling has been the region of weak disorder and
\be  \label{interE}
 2\sqrt{K} < |E| < K+1 \, .  
 \ee 
At those energies   
the mean density of states vanishes to all orders in $\lambda$, for $\lambda \downarrow 0 $~\cite{MD}. 
 Such rapid decay is characteristic of the so-called Lifshits tail spectral regime.  In finite dimensions it is known to lead to localization~\cite{PF,Ki}.  
On tree graphs however, this implication could not be established, and  
localization at weak disorder  was successfully proven \cite{A_wd} only for $ |E| > K+1$  (cf.~Figure~\ref{fig1} and Proposition~\ref{thm:main1} below). 
For energies $E$  in the range \eqref{interE} the nature of the spectrum at weak disorder has been a puzzle even at the level of heuristics~\cite{MD}.   The question is answered by the second of the results mentioned above.  

%
%
 
\subsubsection{Bounded random potentials} 

 
 \begin{figure} 
\begin{center}
\includegraphics[width=0.7\textwidth]{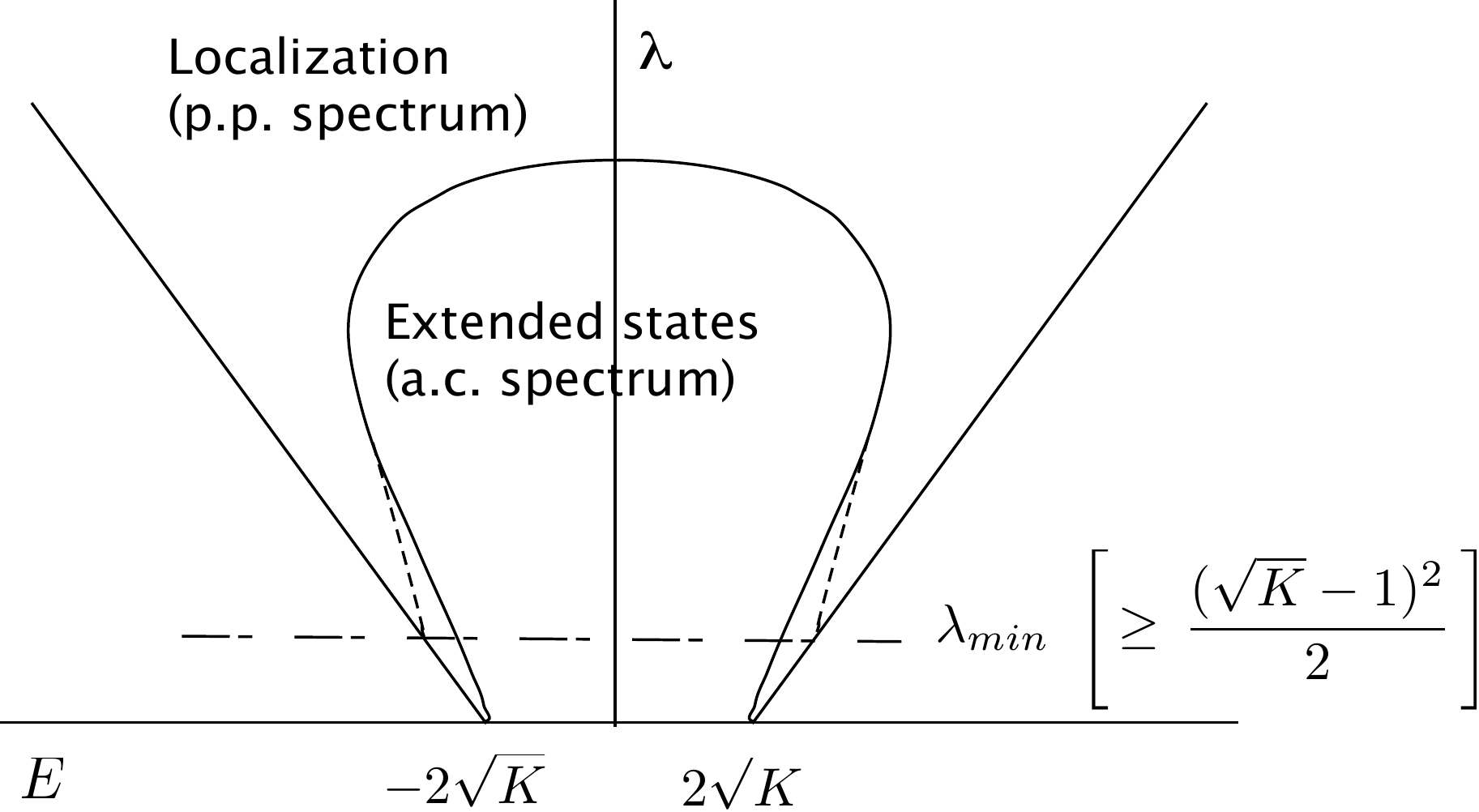}  
\caption{Sketch of the previously expected phase diagram for the Anderson model on the Bethe lattice (the solid line) and the correction presened here (dashed line). 
 Our analysis suggests that at weak disorder  there is no localization and the spectrum is purely \emph{ac}.  While the proof of that is incomplete, we prove that for $\lambda \le  (\sqrt{K}-1)^2/2 )$ near the spectral edges the spectrum is purely absolutely continuous. }
\label{fig:bounded} 
\end{center}
\end{figure}
 
It has been expected that for bounded random potentials the phase diagram  of the random operators \eqref{eq:O} looks qualitatively as depicted  in Fig.~\ref{fig:bounded} (c.f.~\cite{AT,BST}), the key points being: 
\begin{enumerate} 
\item At weak and moderate disorder a mobility edge has been expected to occur, within which the spectrum is absolutely continuous and beyond which it is pure point -  consisting there of a dense countable collection of  eigenvalues with proper eigenfunctions.   
\item The extended states disappear at strong enough disorder ($\lambda > \lambda_{\rm sd}(K)$), where  complete localization prevails.  
\end{enumerate} 
Significant parts of this picture 
have been supported by rigorous results, in particular complete localization at strong disorder~\cite{AM,A_wd}, and  the  persistence of \emph{ac} spectrum at weak disorder~\cite{K,ASW,FHS} 
(though some questions remain as to the precise asymptotics of $\lambda_{\rm sd}(K)$ for $K\to \infty$.   However, as stated in 3. above, at weak and moderate disorder, for regular trees this  picture needs to be modified.  

%

Let us now turn to a more precise formulation of the statements listed above.  

\section{Statement of the main results}  
\subsection{The setup}

Our discussion will focus on operators of the form \eqref{eq:O} in the Hilbert space $\ell^2(\mathcal T)$ of complex-valued, square-summable functions on $ \T $, under the following assumptions:
\begin{enumerate}[A:]
\item\label{assA} $ \T $ is the vertex set of a rooted tree graph with a fixed branching number $K > 1 $  (the root being denoted by $ 0 \in \T $).
\item $ T$ is the adjacency operator  of the graph, i.e., $ \left(T\psi\right)(x) := \sum_{\dist(x,y) = 1 } \psi(y) $ for all $ \psi \in  \ell^2(\T) $. 
\item\label{assC}    
 $\{ V(x;\omega) \, | \, x\in \T\} $ form independent identically distributed (\emph{iid}) random variables, with a probability distribution $\varrho(v) dv$ with $ \varrho \in L^\infty(\mathbb{R}) $,which has a finite moment, i.e., for some $ \varsigma \in (0,1) $:
  \be \label{eq:Veps}
 \int |v|^\varsigma \ \varrho(v) \, dv \ < \ \infty \, .
 \ee
\item  \label{assD} 
The probability density $ \varrho $  is bounded relative to its {\it minimal function}, which we define as 
$ M(v)  := \inf_{\nu \in (0,1] } (2\nu)^{-1}  \int \indfct_{ |x-v| \leq \nu }\,  \varrho(x) \, dx $.   I.e., for Lebesgue-almost all $ v \in \R $:
\begin{equation}\label{eq:regrho}
  \varrho(v) \ \leq \  c \, M(v) \, , 
\end{equation} 
with a finite constant $c$.  
\end{enumerate} 

In case of unbounded potentials, we will mostly restrict our attention to those which additionally satisfy the following assumption:
\begin{enumerate}[A:]
\setcounter{enumi}{4}
\item \label{assE}
For all $k<\infty$:  \quad $ \inf_{ |v|\le k } \varrho(v) > 0 $.   
\end{enumerate}

While condition \ref{assD} could be relaxed, let us note that it is satisfied by all probability distributions whose densities are bounded functions on $\R$ of finitely many humps (see Appendix~\ref{app:L1corr}).  This class includes finite linear combinations of Gaussian, Cauchy, and the piecewise constant functions.

\subsection{The Lyapunov exponent criterion for ac spectrum} 

For a criterion  which is particularly useful at weak disorder (and, separately, also for high values of $K$) let us introduce   the  Lyapunov exponent, which we define for the rooted tree (with the root at $x=0$) as:  
\be \label{eq:LyapE}
  L_\lambda(E) \ := \   -\E(\log |G_\lambda(0,0;E+i0) |)   \, .  
  \ee 
Since  Lyapunov exponents are usually associated  with dynamical systems, let us just comment that the relevance of such a perspective can be seen from the recursive structure of the rooted tree, and the factorization of the Green function  which are discussed in Proposition~\ref{prop:2relations} below.    

The first of the results listed in the introduction is:  

\begin{theorem}\label{thm:mainL}   
For the random operator $H_\lambda(\omega)$ as in~\eqref{eq:O}, with $\lambda>0$,  satisfying Assumptions~\ref{assA}--\ref{assD}:  at Lebesgue-almost every $ E \in \R $ at which 
\be  \label{lyapcond}
L_\lambda(E)   \  < \   \log K \, ,  
\ee 
the operator's Green function satisfies almost surely:
\be  \label{cond:ac}
\Im G_\lambda(0,0;E+i0) \ > \ 0 \, . 
\ee   
\end{theorem}  

The proof of  Theorem~\ref{thm:mainL}, which  is the topic of Section~\ref{Sec:pf_part1} below,  reveals a mechanism for the formation of extended states through rare fluctuation-enabled resonances between distant sites.   \\   
%
%
%
%

For the full spectral implication of the condition~\eqref{cond:ac}, if satisfied throughout an interval of energies, let us quote the following principle which Mira Shamis showed us to follow directly by the arguments presented in Simon and Wolff \cite{Sim_Wolff}.

\begin{proposition}\label{prop:shamis}  Assume that the distribution of $V(0;\cdot)$  conditioned on the  values of the potential at all other sites is almost surely absolutely continuous.  If for some interval $I\subset \R$,  the condition~\eqref{cond:ac} holds for almost every $E\in I$ then with probability one within $I$  the spectral measure $\mu_{\lambda,\delta_0}(du; \omega)$ is absolutely continuous.   If  the  analogous conditions  holds for all sites $x$, then the spectrum of $H_\lambda(\omega)$ is almost surely purely absolutely continuous in $I$. 
\end{proposition} 
 The  proof combines the characterization (due to Aronszajn \cite{Aron}) of the support the singular component of  $\mu_{\lambda,\delta_0}(du; \omega)$  as the set of energies where  condition \eqref{cond:ac} fails, with the {\it spectral averaging}  principle which implies that if this set is of zero Lebesgue measure than  also the spectral measure of this set is zero for almost all realizations of the potential.  This argument applies as well to all other choices for the graph and for the unperturbed operator $T$.



 \subsection{Implications for the phase diagram} 
 \label{sec:applications}

A simple exact calculation (cf. Subsection~\ref{sec:randf})  shows that for $\lambda = 0$ one has
\be 
L_0(E) \  < \  \log K  \qquad \mbox{if and only if}  \qquad  |E| \  < \  K+1  \, . 
\ee
Curiously, the energy range defined by the above condition is strictly larger that the $\ell^2$-spectrum of 
$T$ (cf. \eqref{sigmaT}).

It seems natural to expect  $L_\lambda(E) $ to be continuous in $(\lambda,E)$, a fact which is easily established for 
 the Cauchy random potential, i.e., for $ \varrho( v ) = \pi^{-1} \left( v^2 + 1\right)^{-1} $,  
in which case $ L_\lambda(E) = - \log |G_0(0,0;E+i\lambda) | $. 
In such a situation, Theorem~\ref{thm:mainL}  together with Proposition~\ref{prop:shamis} carry the implication  that for any closed energy interval $I$ in the  range 
$ |E| <   K+1$, at weak enough disorder the random operator $H_\lambda(\omega)$ has almost surely purely absolutely continuous  spectrum in~$I$. 
  
While we do not have a general proof of the continuity of $L_\lambda(E) $, one can show that its averages over intervals are continuous.  Using this weaker continuity  we arrive at the following  conclusion.  
   
\begin{corollary}\label{thm:charcac1}
For unbounded random potentials with $\supp \varrho = \R$,
under the assumption of~Theorem~\ref{thm:mainL} in  
 every closed interval $ I \subset (-K-1\, , K+1) $ there is absolutely continuous spectrum at  sufficiently  low disorder, i.e.   the condition~\eqref{phi_ac} holds at a set of positive measure of energies provided $ 0 <  \lambda < \widehat\lambda(I) $ at some $\widehat\lambda(I) > 0 $.
%
%
\end{corollary}
The proof of Corollary~\ref{thm:charcac1} which is given below in Section~\ref{sec:conLE1} yields also an explicit lower bound on the fraction of $ I $ occupied by ac spectrum. 

For bounded potentials we prove, through other  estimates of $L_\lambda(E)$ which are provided in Section~\ref{sec:conLE2}: 

\begin{corollary}\label{thm:charcac2}
For bounded random potentials with $\supp \varrho = [-1,1]$,
under the assumption of~Theorem~\ref{thm:mainL} 
for 
\be \label{eq:weak_disorder}
\lambda \ < \  [\sqrt{K}-1 ]^2 / 2
\ee 
with probability one $H_\lambda(\omega)$ has purely  absolutely continuous spectrum at the spectral edges, i.e. within a range of energies  of the form
\be 
 |E_\lambda| -  \delta(\lambda) \ \leq \ | E | \ \leq \   |E_\lambda|  .
\ee 
at some $\delta(\lambda) >0$, with  $ E_\lambda = \inf \sigma(H_\lambda)= - 2\sqrt{K} - \lambda \, .$ 
%
%
%
\end{corollary}

  \subsection{Large deviations and a complementary localization criterion}

The criterion provided by Theorem~\ref{thm:mainL} can be improved by taking into account large deviation effects.   The pertinent observation here is that while typically  
\be \label{eq:Lyaptyp}
 \log |G_\lambda(0,x;E+i0)| / |x| \approx - L_\lambda(E) \, , 
 \ee
with $ |x| := \dist(x,0) $, there  typically  also are exponentially many sites to which the Green function  (which can be viewed as expressing  the tunneling amplitude) exhibits a slower decay rate.    A notable feature of  the resulting improved criterion is that it appears to be complementary to the previously developed criterion for localization.  

Information about the large deviations can be recovered from a suitable free energy function, which we define for $s\in [-\varsigma,1)$ by 
\be  \label{eq:phi}
\varphi_\lambda(s;E)  \ := \  \lim_{|x|\to \infty} \frac{\log\, \mathbb{E}\left[\left|   G_\lambda(0,x;E+i0) \right|^s\right]}{|x|}  \, ,
\ee   
and for $s=1$ by $\varphi_\lambda(1;E) := \lim_{s\uparrow 1}\varphi_\lambda(s;E) $.

The existence of the limit (for Lebesgue-almost all $ E \in \mathbb{R} $) is proven below in Section~\ref{Sect:pressure}. We also show there that 
the  function $s \mapsto \varphi_\lambda(s;E) $, which is obviously convex, is monotone decreasing in $ s$  over $ [-\varsigma,1)$, and thus the limit  at $s=1$ is well-defined for almost all~$ E \in \R $. 

Following is the improved version of Theorem~\ref{thm:mainL}.   To avoid an additional complication in the derivation, we establish it here for potentials 
with $ \supp \varrho = \mathbb{R} $ only.   

\begin{theorem}\label{thm:main_phi}   
Under Assumptions~\ref{assA}--\ref{assE}, for any $\lambda >0 $ and Lebesgue-almost all $ E \in \R $ at which 
\be \label{phi_ac}
\varphi_\lambda(1;E) \ > \ - \log K\, ,   
\ee 
the operator's Green function satisfies almost surely
\be
\Im G_\lambda(0,0;E+i0) \ > \ 0 \, . 
\ee   
\end{theorem}  

By convexity arguments   
$
\varphi_\lambda(s;E) \ge -s  \,  L_\lambda(E)  
$ (cf. Section~\ref{Sect:pressure})
and hence the condition \eqref{lyapcond} of Theorem~\ref{thm:mainL} is satisfied whenever \eqref{phi_ac} holds.

For a better appreciation of the criterion provided by the condition~\eqref{phi_ac}, let us note that the opposite inequality implies localization.  This is implied  by the previously established localization results~\cite{AM,A_wd}
which can be recast as follows (cf.~Thm~1.2, and Eqs.~(2.10), (2.12)  in Ref.~\cite{A_wd}). 

\begin{proposition}\label{thm:main1}
Under Assumptions~\ref{assA}--\ref{assC}, if the following condition holds for an interval~$ I $ and a specified $\lambda>0$
\be \label{phi_pp}
{\rm ess}\sup_{E\in I} \; \varphi_\lambda(1;E) \ < \ - \log K   \, , 
\ee  
then the operator $H_\lambda(\omega)$ exhibits exponential dynamical localization in $I$, in the sense of~\eqref{def:exdynloc1} with some $ \mu_\lambda(I) > 0 $.  

Furthermore, the domain in which \eqref{phi_pp} holds includes for each energy $|E|>K+1$ an interval with a positive range of~$\lambda>0$.
\end{proposition}


The relation of the condition~\eqref{phi_pp}, which encodes information about the decay of the Green function, with the time evolution operator is explained  by the following bound:
\be  \label{eq:repFM}
\E \left(\sup_{t\in \R} |\langle \delta_x\, , \, P_I(H_\lambda) \, e^{-itH_\lambda} \, \delta_y \rangle |^2\right) \ 
  \le \  C_{s,\lambda}   \int_I   
\E\left( |G(x,y;E+i0)|^s \right)\, dE \, . 
\ee
which holds for any $s\in [0,1)$  and $ \lambda >0 $ at some constant $ C_{s,\lambda} < \infty $. This inequality
is a reformulation of a result of~\cite{A_wd} on the eigenfunction correlator which was extended in~\cite{Simon_aizThm} so as to apply directly to infinite systems.  (This relation holds in the broader context  of operators with random potential on arbitrary  graphs.)

One may add that if it is only known that for almost all $E\in I$
\be \label{phi_pp2}
\varphi_\lambda(1;E) \ < \ - \log K \   \,   
\ee  
then one may still conclude~\cite{AM}  that the operator has only pure point spectrum in $I$, though not necessarily of uniform localization length.  (The argument proceeds by establishing $ \liminf_{\eta \downarrow 0} \ \sum_{ y \in \mathcal{T} } $ $  \mathbb{E}\left[ \left| G_\lambda(x,y;E+i\eta) \right|^s \right] \ < \infty  $ for some $ s \in (0,1) $ and all $ x \in \T $, and then invoking the Simon-Wolff criterion~\cite{Sim_Wolff} instead of \eqref{eq:repFM}).


\subsection{Further comments} 

\begin{enumerate}
\item
The spectral  criteria provided by Theorems~\ref{thm:mainL} and Theorems~\ref{thm:main_phi} for for $ac$ spectrum, and Proposition~\ref{thm:main1} for localization extend to the corresponding operator on   the fully regular tree graph~$\mathcal{B} $, where every vertex has exactly $ K+1$ neighbors.  The Green function of the operator on $\mathcal{B} $ can be computed from the one on the rooted tree~$\T $ with the help of the recursion relation~\eqref{eq:recur_gen} below. In particular, this implies coincidence of the regimes of 
$ac$ spectra of the operator $ H_\lambda $ on $ \mathcal{T}  $ and~$ \mathcal{B} $.

\item At first sight the $\ell^1$-nature of the condition~\eqref{phi_ac} for $ac$ spectrum may be surprising since -- ignoring fluctuations  -- the loss of square summability seems to correspond to  an $\ell^2$-condition.   
The difference is due to the essential role played by extreme fluctuations, cf.\ Section~\ref{Sec:pf_part1}.  
The constructive effect of fluctuations here  stands in curious contrast to  the fluctuation-reduction arguments which were employed to prove  stability under weak disorder of the $ac$ spectrum for energies $E\in \sigma (T)$~\cite{K,ASW,FHS}. 

 \item The conditions~\eqref{phi_ac} for $ac$ spectrum and~\eqref{phi_pp2} for    localization are not   fully complementary since it was not yet proven that  the equality $ \varphi_\lambda(1;E)  =  - \log K $ holds in the phase diagram only along  a curve. Hence it will be good to see a proof that   
  $\varphi_\lambda(1;E)$ is  differentiable  in $(\lambda, E)$ with only isolated critical points, and that it is likewise regular in $E$ for each given $\lambda$.    This could allow to conclude that the phase diagram of   $H_\lambda $ includes only regimes of localization and regimes of purely $ac$ spectrum (i.e., no $sc$ spectrum), separated by a curve or curves, which are the mobility edge(s).

%

\item  The key observation that rare resonances, whose probabilities of occurrence decay exponentially in the distance, may actually be found to occur on all distance scales when the volume is also growing exponentially fast, is not applicable to graphs of finite dimension.   
However,  it may be of relevance for 
random operators  on other hyperbolic graphs which may include loops (examples of which were considered in~\cite{FHS2,FHH,KSa}), and also  for the analogous random operators on the  Poincar\'e disk.   
Another setup which it will be of interest to see analyzed are random operators on  hypercubes of increasing dimension, which form the configuration spaces of a many particle system.       
\end{enumerate}

\section{Basic properties of the Green function on tree graphs}

\subsection{Notation} \label{notation} 

Analysis on  trees, of this as well as of other problems,  is  aided by the observation that upon the removal of any site   $x$
the tree graph splits into a collection  of   disconnected components, which in case $x$ is the root are isomorphic to the original graph.   For different problems on trees this leads to recursion relations in terms of suitably selected quantities.  
The following notation will facilitate the formulation of such relations in the present context. 
\begin{enumerate}
\item 
For a collection of vertices  $v_1, ... v_n $  on a tree graph $\T$ we denote by $ \T_{v_1,...v_n} $ the disconnected subgraph obtained by  deleting this collection from $\T$.  
\item We denote by 
$H^{\T'}$, with $\T'\subset \T$, the restriction of $H$ to $\ell^2(\T')$.  E.g., $H^{\T_{v_1,...v_n}}$
is  the operator obtained by eliminating all the matrix elements of  $H$ involving  any of the removed sites. 
\item 
The Green function,  $
G^{\T'}(x,y;\zeta) $,  for a subgraph $ \mathcal{T}' $ as above,   is the kernel of the resolvent operator $(H^{\T'} - \zeta)^{-1} $, with $ \zeta \in \C^+ $. 
This function vanishes if $ x $ and $ y $ belong to different connected components of $ \T'$, and otherwise it 
stands for the Green function corresponding to the component  which contains the two. 
 
In particular:  $ G^{\T_{u}}(x,y;\zeta)$ and $G^{\T_{u,v}}(x,y;\zeta)$ 
are the Green functions for the subtree which is obtained by removing $ u $ or, respectively $u$ and $v$, and all the vertices which are past  the removed site(s) from the perspective of $x$ and $ y $. 
\item Given an oriented simple path in $ \T $ which passes through $ u \neq 0 $, we abbreviate (assuming the path itself is clear within the context):
\begin{align}\label{def:Gamma}
\Gamma(u;\zeta) \equiv \Gamma_-(u;\zeta) \ := \  G^{\T_{u_-}} (u,u;\zeta)\, , &  \\ 
\Gamma_+(u;\zeta) \ = \  G^{\T_{u_+}} (u,u;\zeta)\, , & \notag
\end{align}
where $ u_ -$ and $u_+$ are the  neighboring sites of $u$ on that path.   (The paths we shall encounter below typically start at the root, of a rooted tree, and are oriented away from it.) For the root $ 0 $, we will also use the convention
\begin{equation}
	\Gamma(0;\zeta) := G(0,0;\zeta) \, . 
\end{equation}
\item Any rooted tree $ \T $ is partially ordered by the relation $x \prec y$ (resp.\ $ x \preceq y $) which means that $x$ lies on the unique path from the root to~$y$ (possibly coinciding with~$y$). 
 \end{enumerate} 
In order to ease the notation, we will drop the superscript on the Green function of the rooted regular tree, i.e., $ G(x,y;\zeta) = G^\T(x,y;\zeta) $. Moreover, we also drop the dependence  of various quantities on $ \lambda $ at our convenience.

\subsection{Recursion and factorization} \label{sec:randf}

 \begin{proposition}  \label{prop:2relations}
 	Let $ \T $ be the vertex set of a tree graph  (not necessarily a regular  and rooted one). Then, at the complex energy parameter $ \zeta \in \mathbb{C}^+ $, the Green function of  the operator \eqref{eq:O} satisfies:
	\begin{enumerate}
		\item For any $x\in \T$:
		\be \label{eq:recur_gen}
 G(x,x; \zeta) \ = \ \Big( \lambda V(x) - \zeta - \sum_{y\in \mathcal{N}_x}  G^{\T_{x}}(y,y; \zeta) \Big)^{-1} \, ,
\ee 
where $ \mathcal{N}_x := \left\{ y \in \T\, |\, \dist(x,y)=1\right\} $ denotes the set of neighbors of $ x $.
\item For any pair of partially ordered sites,  $0 \prec  x \prec y $, 
\be  \label{gen_fact}
 G(x,y; \zeta) \ = \  G(x,x;\zeta) \prod_{x\prec u \preceq y} \Gamma_-(u;\zeta)  \ 
  = \  G(y,y;\zeta) \prod_{x\preceq u \prec y} \Gamma_+(u;\zeta) \, .
 \ee 
 where the $\pm$ subscripts on $\Gamma$ are defined relative to the  root.  
 	\end{enumerate}
 \end{proposition}
These relations are among the generally used tools  for spectral analysis on trees.  They can be derived by the resolvent identity, or alternatively through a 
random walk representation of the Green function, cf.~\cite{AAT,K,ASW,FHS}.
We will use the following implication of the above.  
\begin{enumerate}
	\item  The relation~\eqref{eq:recur_gen}  yields the 
	 \emph{recursion relation}:
\be  \label{eq:recur}
\Gamma(0;\zeta) = \Big( \lambda \, V(0) - \zeta - \sum_{y \in \mathcal{N}^+_0} \Gamma(y;\zeta) \Big)^{-1}  \,  , 
\ee 
where $\mathcal{N}^+_0$ is the set of forward neighbors of the root $0$ in $\T$.  

In particular: the Green function $ G_0(0,0;\zeta) $ of the adjacency operator $T$ is given by the unique value  of $ \Gamma $ in $ \C^+ $ which satisfies the quadratic equation 
\be \label{eq:quadrG}
 K \Gamma^2 + \zeta \,  \Gamma + 1 = 0 \,. 
 \ee 
From this, one can directly determine that $T$ has the spectrum given by \eqref{sigmaT}, and the spectral measure $\mu_{0,\delta_0}(dE)$ is $ac$ with  density $ \sqrt{(4K-E^2)_+} /(2\pi K)$.  
\item As a special case of~\eqref{gen_fact}, the Green function $ G(0,x;\zeta)$  factorizes into a product of the above variables, taken along the path  from  the root to $x$: 
\be \label{factorization}
G(0,x;\zeta) :=  \prod_{0 \preceq u \preceq x } \Gamma(u;\zeta)  \, .  
\ee  
Moreover, denoting by $ x_- $ the site preceding $x$ from the direction of the root,~~\eqref{gen_fact} also implies:
\begin{equation}\label{eq:factor2}
	G(0,x;\zeta) = G^{\T_x}(0,x_-;\zeta) \, G(x,x;\zeta) \, . 
\end{equation}
More generally, for any triplet of sites  $\{x,u,y\} \subset \, \T $
such that the removal of $u$ disconnects the other two:
\be  \label{eq:fact5}
 G(x,y; \zeta) \ = \  G^{\T_{u}}(x,u_-; \zeta) \ \   G(u,u; \zeta)   \   \  G^{ \T_{u} } (u_+,y; \zeta)
 \ee 
where $ u_ -$ and $u_+$ are the  neighboring sites of $u$, on the $x$ and $y$  sides, correspondingly.  
\end{enumerate}

\subsection{Definition and properties of the free energy}
\label{Sect:pressure}

 To conclude qualitative information on the rate at which $|G_\lambda(0,x;E+i0)|$ decays in $x$, we shall now establish 
 the existence, monotonicity (in $s$), and finite volume bounds for the 
 Green function's free energy~\eqref{eq:phi}. 
It is more convenient to carry the analysis first for complex values of the energy parameter.  Thus, we extend the domain of  the function to include also $\mathbb{C}^+  = \{z \in \C | \, \Im z >0\}$, where the function is defined simply as
%
\begin{equation}\label{def:phixlimit}
  \varphi_\lambda(s;\zeta)  \ := \  \lim_{|x|\to \infty}  \frac{1}{|x|} \log \mathbb{E}\left[\left|G_\lambda(0,x;\zeta)\right|^s\right]     \, ,    
\end{equation}
for all $ \zeta  \in \mathbb{C}^+$.
For  the following statement, we recall that $ \varsigma \in( 0,1) $ is a moment for which it is assumed that $ \E[|V(0)|^\varsigma] < \infty $.

\begin{theorem} \label{thm:phi}
\begin{enumerate}
\item {\rm At any value of the energy parameter in the upper half-plane, $ \zeta \in \mathbb{C}^+ $:}~~For all $ s \in [-\varsigma,\infty) $ the limit  in~\eqref{def:phixlimit} exists and the function $  [-\varsigma,\infty)\ni  s \mapsto \varphi_\lambda(s;\zeta) $ has the following properties:
\begin{enumerate}
	\item $\varphi_\lambda(s;\zeta) $ is convex and non-increasing in $s\in [-\varsigma,\infty)$.
\item For $ s \in [0,2]$:  
\be\label{eq:exclusion}   
-s \, L_\lambda(\zeta) \ \leq \ \varphi_\lambda(s;\zeta) \  \le \ - s \, \log \sqrt{K} \, ,
\ee 
where $ L_\lambda(\zeta) := - \E\left[ \log \left|G_\lambda(0,0;\zeta)\right| \right] $ is the Lyapunov exponent. 
\item 
For any $ s \in [-\varsigma, \infty) $
and $x\in \T$:
 \begin{equation}\label{eq:finitevolume}
 	C_\pm(s;\zeta) ^{-2}  \, e^{|x| \, \varphi_\lambda(s;\zeta) }  \ \leq \  \mathbb{E}\left[\left|G_\lambda(0,x;\zeta)\right|^s\right] \ \leq \ C_\pm(s;\zeta) ^2\,  e^{|x| \, \varphi_\lambda(s;\zeta) } 
 \end{equation}
with $ C_\pm(s;\zeta)  \in (0,\infty) $, which  at any fixed $ s \in [-\varsigma,1) $  
are bounded uniformly in  $  \zeta  \in K + i (0,1] $ for any compact $ K \subset \mathbb{R} $.   
 \item The derivative at $ s =0 $ is given by the (negative) Lyapunov exponent,  i.e.\ for all $ \zeta \in \mathbb{C}^+ $:
 \begin{equation}\label{eq:derphi}
  \frac{\partial \varphi_\lambda}{\partial s} (0;\zeta) = - L_\lambda(\zeta) \, . 
  \end{equation}
\end{enumerate}
\item  {\rm At Lebesgue-almost all real energies, $ E \in \mathbb{R} $:}~~for all $ s \in [-\varsigma,1) $  the limit in \eqref{eq:phi} exists and is finite.  The function  $  [-\varsigma,1)   \ni s \mapsto \varphi_\lambda(s; E) $ 
coincides with the limiting value of   
$ \varphi_\lambda$, i.e., for all $ s \in [-\varsigma,1) $ and  all $ E \in \mathbb{R} $:
\begin{align} \label{eq:2phi}
\varphi_\lambda(s; E) \ & = \  \lim_{\eta \downarrow 0 }\,  \varphi_\lambda(s; E+ i\eta) \notag \\
   &   =  \lim_{\substack{|x|\to \infty \\ \eta \downarrow 0}} \, \frac{1}{|x|} \, \log \mathbb{E}\left[\left|G_\lambda(0,x; E+i \eta)\right|^s\right]  \, .
\end{align}
In particular, within the reduced range:   $s\in [-\varsigma,1)$, the function $\varphi_\lambda (s; E) $ shares the properties listed in {\it (a)}-{\it (c)}, and the  Lyapunov exponent relation  \eqref{eq:derphi} also holds for almost all real values of  $\zeta \, (= E)$.  
\end{enumerate}
\end{theorem}

  
The relation \eqref{eq:2phi} in particular asserts that for $s\in [-\varsigma,1)$ the limits $ \eta \downarrow 0 $ and $ |x| \to \infty $ commute. 
This does not generally extend to $s\ge 1$,  in which case the limit $\eta \downarrow 0$ may diverge if taken first (for $E$ in the regime of pure-point spectrum), while the quantity  on the left is finite and non-increasing in $s$ for all $s\ge  -\varsigma$. 
However, let us add that under certain conditions the constraint $s<1$ could be lifted.  As it should be clear from the proof in Section~\ref{sec:thmproof}, the  relevant condition for  the  finite volume bounds~\eqref{eq:finitevolume} as well as  \eqref{eq:2phi} is  that at the given $s$ and $E=\Re \zeta$  
   the super- and sub-multiplicativity bounds  of Lemma~\ref{lem:submult} and Lemma~\ref{lem:GG2}  hold with constants which are uniform in $\Im \zeta$.   This condition could  be satisfied even  at $s\ge 1$ if, for instance,   the $s$-moments of the Green function factors which yield these constants stay finite as $\eta \searrow 0$ due to a smoothing effect of the absolutely continuous spectrum.

 \subsubsection{Auxiliary results}

Our proof of Theorem~\ref{thm:phi} is based on super- and sub-multiplicativity in $|x|$ of the Green function's moments, properties 
which are related to the Green function's factorization.  
  
Following is   the  essential statement.
%
%

\begin{lemma} \label{lem:submult}
If either  $s \in [-\varsigma,\infty)$ and $ \zeta  \in \C^+ $, or $ s \in [-\varsigma,1) $ and $ \zeta = E+i0 $, then for any two vertices  $0 \prec u \prec x $ (and $ u_\pm $ and $ x_- $  defined in \eqref{eq:fact5}): 
\be \label{eq:surface}
C_-(s;\zeta)^{-1} \ \le \   \frac{\Ev{|G^{\T_{x}}(0,x_-; \zeta)|^s}} {  \Ev{|G^{\T_{u}}(0,u_-; \zeta)|^s} \,  \Ev{ |G^{\T_{u,x}}(u_+,x_-; \zeta)|^s}   } 
   \ \le \  C_+(s;\zeta)
\ee 
with some $ 0 < C_+(s;\zeta), C_-(s;\zeta) < \infty $ which, at fixed  $ s \in [-\varsigma, 1) $ are uniformly bounded in $ \zeta \in K+i (0,1] $ for any compact $ K \subset \mathbb{R} $.   Furthermore for fixed 
$s$ and $\zeta$, within the above range,  
\be \label{eq:C=>1}
 \lim_{s\to 0} C_-(s;\zeta) =   \lim_{s\to 0} C_+(s;\zeta)  = 1  \, .
 \ee  
\end{lemma}

\begin{proof}
Using the factorization representation \eqref{eq:fact5}, and the statistical independence of the two factors which are in the denominator of \eqref{eq:surface} we may write: 
\be  \label{eq:GG} 
   \frac{\Ev{|G^{\T_{x}}(0,x_-; \zeta)|^s}} {  \Ev{|G^{\T_{u}}(0,u_-; \zeta)|^s} \,  \Ev{ |G^{\T_{u,x}}(u_+,x_-; \zeta)|^s}   }   \  
  =\   \Avus{ |G^{\T_{x}}(u,u; \zeta)|^s}  
  \ee
  where $\Avus{\cdot}$ represents the weighted  probability average: 
  \be
    \Avus{Q} \ = \     \frac{ \Ev{|G^{\T_{u}}(0,u_-; \zeta)|^s \, |G^{\T_{u,x}}(u_+,x_-; \zeta)|^s\ \times \  Q}} 
    { \Ev{|G^{\T_{u}}(0,u_-; \zeta)|^s} \,  \Ev{ |G^{\T_{u,x}}(u_+,x_-; \zeta)|^s}   } 
  \ee
To estimate this quantity we note that by \eqref{eq:recur_gen}:
\be \label{eq:Grep}
G^{\T_{x}}(u,u; \zeta) \ = \ \Big( \lambda V(u)  - \zeta -  \sum_{v\in {\mathcal N}_u} G^{\T_{u,x}}(v,v;\zeta) \Big)^{-1}
\ee
\medskip

\noindent
{\it 1. The upper bound:} In case $ s \geq 1 $, the operator-theoretic bound
$|G^{\T_{x}}(u,u; \zeta)| \ \le \ (\Im \zeta)^{-1}$ yields the upper bound in \eqref{eq:surface} with $ C_+ := (\Im \zeta)^{-1} $.

In case $ s \in [0,1) $, the expression \eqref{eq:Grep} and \eqref{eq:fracmonb} readily imply that:
\begin{align} \label{eq:upperbav}
 \Avus{ |G^{\T_{x}}(u,u; \zeta)|^s}   \  \le    \frac{2^s \|\varrho \|_\infty^s}{(1-s) \, \lambda^s} \quad \left( =: \, C_+ \, \right) \, . 
\end{align} 
In case $ s \in [-\varsigma, 0) $, the expression~ \eqref{eq:Grep} together with the inequality $ (|a|+|b|)^\sigma \leq |a|^\sigma + |b|^\sigma $ for $ \sigma \in[0,1] $ also implies:
\begin{multline}
 \Avus{ |G^{\T_{x}}(u,u; \zeta)|^s}  \ \leq \  \lambda^{-s} \E\left[|V(u)|^{-s}\right] + |\zeta|^{-s} + \sum_{v\in {\mathcal N}_u}\Avus{|G^{\T_{u,x}}(v,v;\zeta)|^{-s}  }\, . 
\end{multline} 
To bound the terms $ v \not\in \{ u_- , u_+ \} $, we use~\eqref{eq:upperbav} to conclude that
\begin{equation}
\Avus{|G^{\T_{u,x}}(v,v;\zeta)|^{-s}  } \leq   \frac{\lambda^s }{(1+s) \,2^s  \|\varrho \|_\infty^{s}} \, . 
\end{equation}
In the remaining cases $ v \in \{ u_- , u_+ \} $, we use the factorization property~\eqref{eq:factor2}, Jensen's inequality and~\eqref{eq:upperbav} to conclude: 
\begin{align}\label{eq:C_+}
& \Avus{|G^{\T_{u}}(u_-,u_-;\zeta)|^{-s}  }  \  = \ \left[ Av_{u_-}^{(s)}\left( |G^{\T_{u}}(u_-,u_-;\zeta)|^{s} \right) \right]^{-1}  \notag \\
  &  \leq \ Av_{u_-}^{(s)}\left( |G^{\T_{u}}(u_-,u_-;\zeta)|^{-s} \right) \  \leq \ \frac{\lambda^s }{(1+s) \, 2^s \|\varrho \|_\infty^{s}}  \quad \left( =: \, C_+ \, \right) \, , 
 \end{align}
and similarly for $ u_+ $. (Note that in case $ u_- = 0 $, the definition of $ Av_{u_-}^{(s)} $ extends naturally.)   \\ 

\noindent
{\it 2. The lower bound:}  First assume that $ s > 0 $. The expression \eqref{eq:Grep} implies for any $ t > 0 $ and any $ \varepsilon \in (0,\min\{\varsigma,s\}]$:
\begin{align} 
& \Avus{  | G^{\T_{x}}(u,u; \zeta)|^s } \
 \geq \  \Avus{  \frac{\indfct\left[ \mbox{For all $ v \in \mathcal{N}_u$:} \quad |G^{\T_{u,x}}(v,v;\zeta)| \leq t    \right] }{\left[\lambda |V(u)| + |\zeta| + (K+1) \, t \right]^{s}}  }\notag \\
& \geq \  \frac{\prod_{ v \in \mathcal{N}_u}  \Avus{  \indfct\left[  |G^{\T_{u,x}}(v,v;\zeta)| \leq t    \right]}}{ \left[\, \lambda^\varepsilon \, \Ev{|V(0)|^\varepsilon} + |\zeta|^\varepsilon+ (K+1)^\varepsilon\,  t^\varepsilon\right]^{s/\varepsilon}}    \, .
\end{align}  
The last inequality derives from that fact that the random variables appearing in the numerator and $ V(u) $ are independent (even with respect to $ \Avus{\cdot} $), and Jensen's inequality, which yields $ \mathbb{E}\left[|Q|^{-s} \right] \ \ge \  
\E\left[  |Q|^{-\varepsilon} \right]^{s/\varepsilon } \geq \E\left[  |Q|^{\varepsilon} \right]^{-s/\varepsilon }$. 
We now choose $ t \equiv t(s) $ large enough, so that $ \Avus{  \indfct\left[  |G^{\T_{u,x}}(v,v;\zeta)| \leq t \right]} \geq 1-s$. 
In case $ v \not\in \{ u_- , u_+ \} $ this is quantified in the estimate~\eqref{eq:pickt}, and in case $ v \in \{ u_- , u_+ \} $ in~\eqref{eq:pickt2}. 

If $ s \in [-\varsigma, 0] $, we use the Jensen inequality together with~\eqref{eq:upperbav}  to conclude that
\begin{equation}
	\Avus{  | G^{\T_{x}}(u,u; \zeta)|^s }  \geq \frac{1}{\Avus{  | G^{\T_{x}}(u,u; \zeta)|^{-s} } } \geq  \ \frac{(1+s) \, \lambda^s}{ 2^s  \|\varrho \|_\infty^s} \quad \left( =: \, C_-^{-1} \, \right)\, ,
\end{equation}
which completes the proof of~\eqref{eq:surface}, and by  inspection also of~\eqref{eq:C=>1}. 
\end{proof}


The above lemma addresses the Green function restricted to subgraphs. Arguments used in the proof also imply that the full Green function may in fact be compared with its restricted versions. Moreover, the effect of peeling off one vertex is bounded:
\begin{lemma} \label{lem:GG2} 
Under the assumptions of Lemma~\ref{lem:submult}, let $ x_{--} $ stand for the neighbor of $ x_- $ towards the root:
 \begin{align}
  \label{eq:relative}
C_-(s;\zeta)^{-1} \  \leq \ &  \frac{\Ev{|G^{\T_{x}}(0,x_-; \zeta)|^s}}{\Ev{|G^{\T_{x_-}}(0,x_{--}; \zeta)|^s} } \ \leq \ C_+(s;\zeta) \, , \\
\label{eq:restunrest}
[C_+(s;\zeta) C_-(s;\zeta)]^{-1}
 \ \le \ &  \frac{\Ev{|G(0,x_-; \zeta)|^s} }{  \Ev{|G^{\T_{x}}(0,x_-; \zeta)|^s}}  \ \le \  C_+(s;\zeta) C_-(s;\zeta) \, , 
\end {align}
where $ x_{--} $ is the neighbor of $ x_- $ towards the root. 
\end{lemma}  

\begin{proof} 
For the proof of~\eqref{eq:relative} we use the factorization of the Green function: 
\begin{equation}
 G^{\T_{x}}(0,x_-; \zeta) = G^{\T_{x_-}}(0,x_{--}; \zeta) \, G^{\T_{x}}(x_-,x_-; \zeta) \, . 
 \end{equation}
 Since the last factor is of the form~\eqref{eq:Grep}, the argument used in the proof of Lemma~\ref{lem:submult} yields~\eqref{eq:relative}. 

For a proof of~\eqref{eq:restunrest}  we employ the factorization:
\begin{equation} 
G(0,x; \zeta) =   G^{\T_{x}}(0,x_-; \zeta)  \ G (x,x; \zeta) \, .
\end{equation}
Thus, by arguments as in the proof of Lemma~\ref{lem:submult}, the quantity $ \Ev{|G(0,x; \zeta)|^s} $ is bounded from above and below in terms of $ \Ev{|G^{\T_{x}}(0,x_-; \zeta)|^s} $. Since the latter lacks $ x  $, we apply~\eqref{eq:relative}  to append this vertex. 
\end{proof}

 \subsubsection{Proof of Theorem~\ref{thm:phi} }    \label{sec:thmproof}  
 
We now turn to the main results on the free energy function.  In this context, we recall that a supermultiplicative positive sequence is one satisfying: 
$  \alpha_{m+n} \ge B \,  \alpha_m \, \alpha_n >0 $.  
By Fekete's lemma \cite{Fek} for such sequences the limit 
$ \lim_{n\to \infty} \ n^{-1}  \log \alpha_n =: \Psi  $, 
exists and  $ \alpha_m \ \le \ B ^{-1} e^{m \Psi  } $ for every $m\in \N$. 
For submultiplicative sequences the reversed inequalities hold.  
 
 \begin{proof}[Proof of Theorem~\ref{thm:phi}]
In the following we pick a simple path in $ \T $ to infinity, and label its vertices by $ 0=: x_0  , x_1, x_2, \dots$.   
We first show that 
\begin{equation}
	\alpha_n(\zeta) := \mathbb{E}\left[ \left| G^{\T_{x_{n+1}}}(x_0,x_n;\zeta)\right|^s \right] 
\end{equation}
is supermultiplicative in the two cases of interest: 1.~$ s \in [-\varsigma,\infty) $ and $ \zeta \in \mathbb{C}^+ $ and 2.~$ s \in [-\varsigma,1) $ and $ \zeta = E + i0 $. In both cases, the factorization property~\eqref{eq:fact5},  Lemma~\ref{lem:submult}  and~\eqref{eq:relative}  imply for all $ n,m \in \mathbb{N} $:
\begin{equation}
	\alpha_{n+m+1}(\zeta) \ \geq \ C_-^{-1}\,  \alpha_n(\zeta)\,  \alpha_m(\zeta) \ \geq (C_+ C_-)^{-1} \,  \alpha_{n+1}(\zeta)\,  \alpha_m(\zeta)\, .
\end{equation}
By Fekete's lemma~\cite{Fek}, the limit $ \Psi(\zeta) := \lim_{n\to \infty} n^{-1} \log \alpha_n(\zeta) $ exists.   

Analogous reasoning using Lemma~\ref{lem:submult} and~\eqref{eq:relative} also show submultiplicativity, i.e., for all $ n,m \in \mathbb{N} $:
\begin{equation}
	\alpha_{n+m+1}(\zeta) \ \leq \ C_+ \, \alpha_n(\zeta) \, \alpha_m(\zeta) \ \leq \  C_+ C_- \  \alpha_{n+1}(\zeta) \, \alpha_m(\zeta) \, .
\end{equation}
By super- and sub-multiplicativity, the limit $ \Psi(\zeta) $ provides both an upper and lower bound on $ \alpha_m(\zeta) $ for  any $ m \in \mathbb{N} $:
\begin{equation}\label{eq:fekub}
(C_+ C_-)^{-1}  \, e^{m \Psi(\zeta)} \leq	\alpha_m(\zeta) \  \ \leq \ C_+ C_- \, e^{m \Psi(\zeta)} \, . 
\end{equation}
To establish the existence of the limits~\eqref{def:phixlimit}  and~\eqref{eq:phi}, we use~\eqref{eq:fekub} and~\eqref{eq:restunrest}  which reads
\begin{equation}\label{eq:restun1}
C_\pm^{-1} \,  \alpha_n(\zeta)  \leq \ \mathbb{E}\left[ \left| G(x_0,x_n;\zeta)\right|^s \right]  \leq \ C_\pm\,  \alpha_n(\zeta) \, .  
\end{equation}
with $ C_\pm := C_+ C_-$.   
Hence the limits~\eqref{def:phixlimit}  and~\eqref{eq:phi} agree with $ \Psi(\zeta) = \varphi_\lambda(s;\zeta) $ in both cases: {\it i.}~$ s \in [-\varsigma,\infty) $ and $ \zeta \in \mathbb{C}^+ $ and {\it ii.}~$ s \in [-\varsigma,1) $ and $ \zeta = E + i0 $. 

Since for any fixed $ s \in [-\varsigma,1) $ and $E\in \R$ the constants $ C_+ , C_- , C_\pm$ are  bounded uniformly in $ \Im \zeta \in (0,1] $,  the convergence~\eqref{def:phixlimit} is also uniform with respect to $ \Im \zeta \in (0,1] $, 
and the limits $ \eta \downarrow 0 $ and $ |x| \to \infty $ can be taken in any order.   This proves~\eqref{eq:2phi}.  

The finite-volume bounds~\eqref{eq:finitevolume}  now follow from~\eqref{eq:fekub}
and~\eqref{eq:restun1}.  

It remains to establish the properties listed in {\it (a)}, {\it (b)} and {\it (d)}.
Since the prelimits are convex functions of $ s $, the limit is convex. Since  for any $ \epsilon \geq 0 $
\begin{equation}
 \mathbb{E}\left[ \left| G(0,x;\zeta)\right|^{s+\epsilon} \right] \leq (\Im \zeta)^{-\epsilon} \;  \mathbb{E}\left[ \left| G(0,x;\zeta)\right|^{s} \right]  \, ,
 \end{equation}
 the limit~\eqref{def:phixlimit} is non-increasing in $ s $. This concludes the proof of {\it (a)}.

The first inequality in~\eqref{eq:exclusion} is a consequence of convexity and the factorization property~\eqref{factorization} of the Green function. In fact, if either  1.~$ s \in [-\varsigma,\infty) $ and $ \zeta \in \mathbb{C}^+ $ or 2.~$ s \in [-\varsigma,1) $ and $ \zeta = E + i0 $:
\begin{equation}\label{eq:lowerbpL}
	\log \mathbb{E}\left[  \left| G(0,x;\zeta)\right|^s \right] \ \geq \ s \, \mathbb{E}\left[\log  \left| G(0,x;\zeta)\right| \right] \ = \ - s \, |x| \,  L(\zeta) \, .
\end{equation}
The second inequality in~\eqref{eq:exclusion} relies on the following bound on the sums of squares of Green functions
\begin{equation}\label{eq:sumrule}
	\sum_{|x|= n}  |G(0,x;\zeta)|^2 \leq \sum_{x\in \T}  |G(0,x;\zeta)|^2 = \frac{\Im G(0,0;\zeta) }{\Im \zeta } \leq \frac{1}{(\Im \zeta )^2} \, . 
\end{equation}
From the finite-volume bounds~\eqref{eq:finitevolume}, we conclude that for any $ n  = \dist(x,0) \in \mathbb{N} $:
\begin{align}
	K^{n } \, e^{n \, \varphi(2;\zeta)} & \ \leq \ C_\pm^2 \, K^{n } \;  \mathbb{E}\left[  |G(0,x;\zeta)|^{2}\right] \notag \\
	& \ = \,   C_\pm^2 \;  \mathbb{E}\Big[  \sum_{|x| = n } |G(0,x;\zeta)|^2 \Big] \leq   \frac{ C_\pm^2}{ (\Im \zeta)^{2}} \, . 
\end{align}
The right side is independent of $ n $, and thus $ \varphi(2;\zeta) +  \log K \leq 0 $. Since $  \varphi(0;\zeta) = 0 $, convexity implies $   \varphi(s;\zeta) \leq - s \, \log \sqrt{K} $ for all  $ s \in [0,2] $. This concludes the proof of~{\it (b)}. 

Let us now turn to the differentiability property~{\it (d)}. If either $ s \in [-\varsigma,\infty) $ and $ \zeta \in \mathbb{C}^+ $ or~$ s \in [-\varsigma,1) $ and $ \zeta = E + i0 $, the factorization property~\eqref{factorization} of the Green function, \eqref{eq:exclusion} and the finite-volume bounds~\eqref{eq:finitevolume} imply:
\begin{align}
	0 \ \leq & \ \varphi(s;\zeta) + s \, L(\zeta) \notag \\ 
	 \leq & \ \frac{1}{|x|} \left(  \log \mathbb{E}\left[| G(0,x;\zeta)|^s \right] - \mathbb{E}\left[\log | G(0,x;\zeta)|^s \right] \right) + \frac{\log C_\pm^2}{|x|} \notag \\
	\leq & \  \frac{s^2}{ 2 |x|} \,  \mathbb{E}\left[\left(\log | G(0,x;\zeta)| \right)^2 \left(  | G(0,x;\zeta)|^s +1 \right) \right]   + \frac{\log C_\pm^2}{|x|}  \, . 
\end{align}
Here the last inequality derives from the two elementary bounds $ e^\alpha \leq 1 + \alpha + \alpha^2 ( e^\alpha + 1 ) /2 $ and $ 1+\beta \leq e^\beta $ valid for all $ \alpha, \beta \in \mathbb{R} $. 
Using the fractional moment bounds~\eqref{eq:fracmonb} and the factorization property of the Green function, it  is easy to check that there is some constant $ C < \infty $ such that for all $ s \in (0,1/4) $  and $  x \in \mathcal{T}  $
 the first factor is bounded by $ C s^2 |x| $.  
 Furthermore, since $ \log C_\pm^2(s;\zeta) = o(1) $ as $ s \to 0 $ by~\eqref{eq:C=>1}, the claim~\eqref{eq:derphi} follows by choosing $ |x | = \lfloor s^{-1}\,  ( \log C_\pm^2 )^{1/2} \rfloor$. 
 \end{proof}

 \subsection {Green function's typical decay rate, and its large deviations}
 
 The properties established  in  Theorem~\ref{thm:phi} for the free energy function $\varphi_\lambda(s;E)$ allow one to establish decay properties  of the Green function which are important for the resonance analysis which is presented below.   The typical behavior is determined by the Lyapunov exponent:
 
  \begin{figure}
\begin{minipage}{.5 \textwidth}
\includegraphics[width= .9\textwidth]{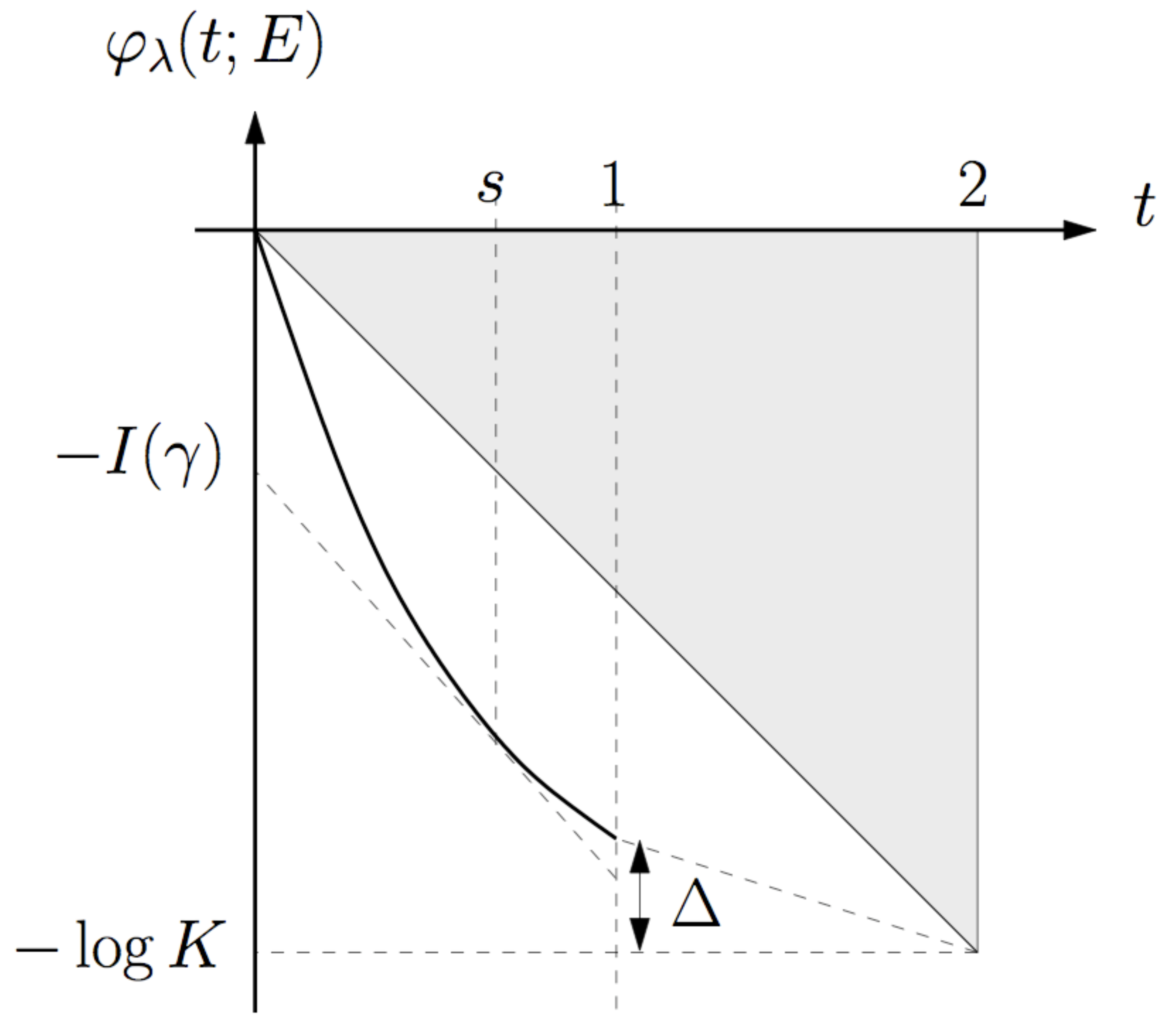} 
\end{minipage}\hfill
\begin{minipage}{.45 \textwidth}
\caption{Sketch of the free energy function in case $ \varphi_\lambda(1;E) > - \log K $. Regardless of this assumption the curve does not enter the shaded region.   The parameter $ \gamma $ is the negative slop of the tangent at $ s $ and the value of the rate function $ I(\gamma) = - \varphi_\lambda(s;E) - s \gamma$ can be read off as the negative value at the intersection of that tangent with the vertical axis. 
}\label{fig:phi}
\end{minipage}
\end{figure} 

%
%
%
\begin{theorem}~\label{thm:regGbehav}
For almost all $ E \in \mathbb{R} $  and all $ \epsilon > 0 $ there is some $ \eta_0 > 0 $ such that for all $ \eta \in (0,\eta_0) $:
\begin{equation}
	\lim_{|x| \to \infty} \, \mathbb{P}\left( | G(0,x;E+i\eta) | \in e^{ -L(E) |x| } \, \big[ e^{-\epsilon |x|} , e^{\epsilon |x|} \big] \right)  \ = \ 1 \, . 
\end{equation}
The same applies to  $ G^{\mathcal{T}_x}(0,x_-;E+i\eta) $ (when substituting $ G(0,x;E+i\eta) $).
\end{theorem}
The proof  is presented in  Appendix~\ref{App:Ldevs}, based on the general and more comprehensive large-deviation  Theorem~\ref{thm:ld}. 
The latter is established through some standard arguments for which enabling bounds are provided by Theorem~\ref{thm:phi}.   

Other values of $|x|^{-1} \log | G(0,x;E+i\eta) |   $ can also be observed, but these represent large deviations for which the rate function is given 
by the Legendre transform:
\be  \label{eq:Legendre} I(\gamma) \ := \ - \inf_{ s\in[-\varsigma,1)} \left[ \varphi_\lambda(s;E) + s \gamma  \right]  \, . 
\ee 
More explicitly,  for any $\gamma$ which is attainable as  $\gamma = - \partial \varphi_\lambda(s;E)/\partial s $ at $s\in[-\varsigma,1)$: 
\be  \label{eq:LD}
\mathbb{P}\left( | G(0,x;E+i\eta) | \in e^{ -\gamma |x| } \, \big[ e^{-\epsilon |x|} , e^{\epsilon |x|} \big] \right)  \ \approx \   e^{-I(\gamma)] |x|}  \, , 
\ee 
where $\approx$ means that the ratio of the two terms is of the order $e^{o(|x|)}$ for  large $|x|$. 
A  stronger large-deviation principle is presented in Theorem~\ref{cor:ldev1}.

%
%
%

\section{The Lyapunov exponent  delocalization criterion}  \label{Sec:pf_part1}

Our goal in this section is to prove   Theorem~\ref{thm:mainL}.   We start with some useful preparatory observations. 

\subsection{A zero-one law and the relative tightness of $ \Im  \Gamma(0;E+i\eta) $}

\begin{lemma}\label{lem:01}
For  Lebesgue-almost all $E\in \R$, the probability that $ \Im  \Gamma(0;E+i0)\  =\  0$ 
is either $0$ or $1$. 
\end{lemma} 

\begin{proof} 
Taking the imaginary part of \eqref{eq:recur} one gets: 
\begin{align} \label{eq:ImG1}
\Im \Gamma(0;E+i\eta)  \ & =    \ \left| G(0,0;E+i\eta) \right|^2 \Big[\eta + \sum_{x \in \mathcal{N}^+_0 }   \Im \Gamma(x;E+i\eta)  \Big] 
\notag \\
 \  & \ge \   \left| G(0,0;E+i\eta) \right|^2  \sum_{x \in \mathcal{N}^+_0 }  \Im \Gamma(x;E+i\eta)   \, ,
\end{align} 
with equality in case $ \eta = 0 $ for those $ E $ for which the boundary values exist, that is for Lebesgue-almost all $E \in \mathbb{R}  $.
Let now $q:= \Pr{   {\mbox{$\Im  \Gamma(0;E+i0)\  = \  0$} } } $. The factor $|G(0,0;E+i0)|$  is  almost surely non-zero, since, for example, $\mathbb{E}[|G(0,0;E+i0)^{-\varsigma}] < \infty $, using the recursion relation~\eqref{eq:recur}, Assumption~\ref{assC} and the finiteness of fractional moments. Since the $K$ different terms, $\Im \Gamma(x;E+i0)  $, $ x \in\mathcal{N}^+_0$,   
are independent variables of the same distribution as $\Im \Gamma(0;E+i0)  $, and 
$|G(0,0;E+i0)| \neq 0 $  almost surely, we may conclude that  
$ q \ =   \  q^{K} $ or $q\, [1-q^{K-1}] \ = \ 0$, and hence    either  $q=0$ or $q=1$.  
\end{proof}  

In order to quantify the way the distribution of $ \Im \Gamma(0;\zeta) $ settles on its limit as $ \Im \zeta \downarrow 0 $, we introduce the following quantity. 
\begin{definition} \label{def:xi}
 For $ \zeta \in \C^+ $ and $ \alpha \in (0,1) $
the \emph{upper  percentile} $\xi(\alpha,\zeta)$ of the distribution of $\Im\Gamma(0;\zeta)$
is  the supremum of the values of $t\geq0$ for which    
\be  \label{eq:t}
    \mathbb{P}\left(  \Im\Gamma(0;\zeta) \geq t  \right) \  \ge \ \alpha  \, .  
\ee  
\end{definition}

\begin{lemma} \label{lem:xi}
For $ \zeta \in \C^+ $ and any $\alpha \in (0,  1)$:  
\quad $ 
0 \ < \  \xi(\alpha,\zeta) \  < \  \infty $.
\end{lemma}  
\begin{proof}  For $\zeta \in \C^+$ one has $0 < \Im \Gamma(0;\zeta)  \ \leq \ (\Im \zeta)^{-1} $.   Hence the claim derives from the following observations:  
{\it i.\/} The collection of strictly positive values of $t$ at which \eqref{eq:t}  holds is not empty, since otherwise $ \Im \Gamma(0;\zeta)=0 $ with probability one. 
{\it ii.\/}  The above collection of values of $t$ does not include any value above $ (\Im \zeta)^{-1}$.  
\end{proof}
Iterating \eqref{eq:ImG1} we conclude that for any  $ n \in \mathbb{N} $ and $\zeta \in \mathbb{C}^+$:  
\begin{equation}\label{eq:Imrecrel}
	\Im \Gamma(0;\zeta)  \ \geq \ \sum_{x\in \mathcal{S}_n } \left| G(0,x;\zeta) \right|^2  \sum_{y\in \mathcal{N}^+_x}\Im \Gamma(y;\zeta)  
\end{equation}
where $ \mathcal{S}_n := \{ x\in \T |  \dist(0,x)=n \}$.  As a first consequence of this important relation, we note that the distribution of $ \Im \Gamma(0;\zeta) $ does not broaden too fast as $ \Im \zeta \downarrow 0  $. As a measure of the (relative) width of the distribution we use the ratios $ \xi(\alpha;\zeta)/\xi(\beta;\zeta) $. %
\begin{lemma}
	For any $E \in  \R $ the distribution of $ \Im \Gamma(0;E+i\eta) $ remains relatively tight in the limit $ \eta \downarrow 0 $ in the sense that 
	for any pair $ \alpha, \beta  \in (0,1) $:
	\begin{equation}\label{eq:apriorib}
		\liminf_{\eta \downarrow 0 } \; \frac{\xi(\alpha;E + i\eta)}{\xi(\beta;E + i\eta)} \ > \ 0  \, . 
	\end{equation}
\end{lemma}
\begin{proof}
We fix $ \alpha , \beta \in (0,1) $ (by monotonicity it would suffice to consider the case $ \alpha > \beta $) and pick an arbitrary $0 < \epsilon < 1 - \beta $. 
For a given $ x \in \mathcal{S}_n $, let us consider the event  
$	R_x := \{ (| G(0,x;E + i\eta) | \geq e^{- n \ell} \}  $, where  $ \ell >  L(E)$ is fixed at an arbitrary value. 
We now choose $ n \in \mathbb{N} $ large enough and $ \eta_0 > 0 $ small enough such that  for all $\eta \in (0,\eta_0) $ simultaneously 
\begin{align}\label{eq:param1}
	\mathbb{P}\left(R_x^c \right) \ \leq \ \alpha \left( 1 - \sqrt{\frac{\beta}{1 - \epsilon}} \right) \quad \mbox{and} \quad 
	K^n \alpha \, \sqrt{\frac{1 - \epsilon}{\beta}} \ \geq \ \frac{\beta}{\epsilon} \, ,
\end{align}
where the superscript indicates the complementary event. 
While the second requirement is obviously satisfied for $ n = |x| $ large enough, it follows from Theorem~\ref{thm:regGbehav} that also the first requirement can be met. 
In order to control the sum in~\eqref{eq:Imrecrel} we also introduce the event
$ I_x := \bigcup_{y\in \mathcal{N}^+_x} \{ \Im \Gamma(y;E + i\eta)   \geq  \xi(\alpha;E + i\eta) \} $. 
From~\eqref{eq:Imrecrel} and the Cauchy-Schwarz inequality it then follows that 
\begin{align}
\mathbb{P}\left( \Im\Gamma(0;\zeta) \ \geq \ e^{-2 \ell n } \,  \xi(\alpha;E + i\eta) \right) \ \geq \ \mathbb{P}\left( N \geq 1 \right) \ \geq \ \frac{\mathbb{E}\left[N\right]^2}{\mathbb{E}\left[N^2\right]}  \, , \label{eq:N2one}
\end{align}
where
$ N := \sum_{x \in \mathcal{S}_n} \indfct_{R_x\cap I_x} $ denotes the number of joint events $ R_x\cap I_x $ on the sphere $ \mathcal{S}_n $. 
The right side in~\eqref{eq:N2one} is estimated using the independence of the events $ I_x $ for all $ x \in \mathcal{S}_n $:
\begin{equation}
	\mathbb{E}\left[N^2\right] - \mathbb{E}\left[N\right]  \   = \ \mathbb{E}\left[N(N-1)\right] \ \leq  \sum_{ \substack{x, y  \in \mathcal{S}_n \\ x\neq y}} \mathbb{P}\left(I_x\right) \mathbb{P}\left(I_y\right)  \leq \  K^{2n}\, \mathbb{P}\left(I_x\right)^2 \, . 
\end{equation}
Together with the lower bound
\begin{equation}
	\mathbb{E}\left[N\right] \ = \  K^n \, \mathbb{P}\left( R_x\cap I_x \right) \ \geq \ K^n \left(  \mathbb{P}\left(I_x \right)  - \mathbb{P}\left( R_x^c \right) \right) \ \geq \ K^n\ \left(  \alpha   - \mathbb{P}\left( R_x^c \right) \right) \ \geq \ \frac{\beta}{\epsilon} \, , 
\end{equation}
the inverse of the right side in \eqref{eq:N2one} is bounded from above using~\eqref{eq:param1}: 
\begin{equation}\label{eq:N2two}
	\frac{\mathbb{E}\left[N^2\right]}{\mathbb{E}\left[N\right]^2} \ \leq \ \frac{1}{\mathbb{E}\left[N\right] } +  \left( 1 - \frac{\mathbb{P}\left( R_x^c \right) }{\alpha} \right)^{-2} \ \leq \ \frac{\epsilon}{\beta} + \frac{1-\epsilon}{\beta} \ = \ \frac{1}{\beta} \, . 
\end{equation}
From the definition of the upper percentile and \eqref{eq:N2one} together with~\eqref{eq:N2two} it hence follows
$ \xi(\beta;E + i\eta) \ \geq \ e^{-2 \ell n } \,  \xi(\alpha;E + i\eta) $. 
The proof is concluded by noting that the first factor in the right side is independent of $ \eta $ and strictly positive. 
\end{proof}


\subsection{A conditional proof of the criteria} 
\label{sec:cond}

We prove Theorems~\ref{thm:mainL}  and~\ref{thm:main_phi} by contradicting the following `no-ac' hypothesis.  
\begin{definition}\label{def:noac}
 For a specified $\lambda \geq 0 $, we say that the \emph{no-ac hypothesis} at $ E \in \R $ holds if  almost surely $\Im G(0,0;E+i0) \ = \ 0$.
\end{definition}  
The relation~\eqref{eq:Imrecrel} suggests that the no-ac hypothesis is false if with uniformly positive probability  there are sites $x \in \mathcal{S}_n$ with 
$\left| G(0,x;\zeta) \right| \gg  1$, and  a forward neighbor $ y $ with a  not particularly `atypical'  value   of $\Im \Gamma(y;E+i\eta)$.  A key step is: 

\begin{theorem}\label{lem:largeG}
For almost all $ E \in \sigma(H_\lambda) $, if  either
\begin{enumerate}
\item $ L(E) < \log K $, \;    or \hfill \emph{(Lyapunov exponent criterion)}
\item $ \varphi(1;E) > - \log K $, and Assumption~\ref{assE}, \hfill \emph{(large-deviation criterion)}
\end{enumerate} 
and  the  no-ac hypothesis holds true, 
then  there are $  \delta , p_0 > 0 $ and $ n_0 \geq 0 $ such that for all $ n \geq n_0 $:
\begin{equation}\label{eq:largeG}
	\liminf_{\eta \downarrow 0 } \, \mathbb{P}\left( \max_{x \in  \mathcal{S}_n}\,   |G(0,x;E+i\eta)| \ \indfct_{\max_{y \in \mathcal{N}_x^+} \Im\Gamma(y;E+i\eta) \geq \xi(\alpha;E+i\eta) } \geq \ e^{\delta n } \right) \ \geq \  2 p_0 \, . 
\end{equation}
\end{theorem}
A heuristic argument for the validity of Theorem~\ref{lem:largeG} is given in Subsection~\ref{sec:heuristic} below.   The  proof  is split: the  Lyapunov exponent criterion is established in Subsection~\ref{subsec:Lyproof}, whereas the proof of the large-deviation criterion, which  is a bit more involved,   is given separately in Section~\ref{sec:GFextrev}. 
First however let us show how Theorem~~\ref{lem:largeG} is used for the  proof of our main results.
  
\begin{proof} [Proof of Theorem~\ref{thm:mainL}  and Theorem~\ref{thm:main_phi} -- given Theorem~\ref{lem:largeG}] 
We will argue by contraction.  Assume the no-ac hypothesis for the given energy $ E \in\sigma(H_\lambda) $.
From Lemma~\ref{lem:largeG} and~\eqref{eq:Imrecrel} it then follows that there are $ \alpha, \delta , \eta_0, p_0 > 0 $ and $ n_0 \geq 0 $ such that
for all $ \eta \in (0,\eta_0)  $ and all $ n \geq n_0 $: 
\begin{align}
	& \mathbb{P}\left( \Im \Gamma(0;E+i\eta) \ \geq \  e^{2 \delta n } \,  \xi(\alpha;E+i\eta)  \right) \notag \\
	& \geq \mathbb{P}\left( \max_{x \in  \mathcal{S}_n}\,  |G(0,x;E+i\eta)| \ \indfct_{\max_{y \in \mathcal{N}_x^+} \Im\Gamma(y;E+i\eta) \geq \xi(\alpha;E+i\eta) } \geq \ e^{\delta n } \right) \ \geq \ p_0  \, .
\end{align}
As a consequence, we conclude
$	 \xi(p_0 ;E+i\eta) \ \geq \ e^{2 \delta n } \,  \xi(\alpha;E+i\eta) $, and since $ n $ can be taken arbitrarily large
\begin{equation}
	\lim_{\eta \downarrow 0 } \, \frac{ \xi(\alpha;E+i\eta) }{ \xi(p_0 ;E+i\eta)  } \ = \ 0 \, . 
\end{equation}
This however contradicts the relative tightness condition~\eqref{eq:apriorib}. 
\end{proof}

\subsection{Heuristics of the resonance mechanism} \label{sec:heuristic}

A possible mechanism for the rare events featured in~\eqref{eq:largeG}  is the simultaneous occurrence of the following  two events, at some common value of $\gamma > 0 $:
\begin{align}
\label{eq:gamma2}
  \left|G(x,x;E+i\eta) \right| \  & \ge    \  e^{(\gamma + \delta)\, |x|}   \\[1ex]  
 \left| G^{\T_x}(0,x_-;E+i\eta)\right|  \  & \ge    \  e^{- \gamma\, |x|}\,  \label{eq:gamma1}   \, .  
\end{align} 
These two conditions imply $ |G(0,x;E+i\eta)|  \geq e^{\delta |x| } $ through the relation~\eqref{eq:factor2}.

The first,~\eqref{eq:gamma2}, represents an extremely rare local resonance condition.  It occurs when the random potential at $x$ falls very close to a value at which $G(x,x;E+i0)$ diverges.  By~\eqref{eq:recur_gen}, such   divergence is possible if $G^{\T_{x}}(y,y; E+i0) $ is real at all $y\in \mathcal N_x$.   By~\eqref{eq:recur_gen} and the continuity of the probabilities in~$\eta$,  under the no-ac hypothesis  the probability of~\eqref{eq:gamma2}  occurring at a given site $x\in \mathcal{S}_n$ is of the order $e^{-(\gamma +\delta) n}$  for $\eta$ sufficiently small (depending on~$n$).   

The second condition,~\eqref{eq:gamma1}, represents 
\begin{enumerate}[i)]
\item a typical event, in case $ \gamma  \  =\    L(E) $ \quad  (cf.~Theorem~\ref{thm:regGbehav}),
\item a large deviation event, in case $ \gamma < L(E) $ \quad (cf.~\eqref{eq:LD}).
\end{enumerate}
In the first case, the mean number of  sites in the sphere $ \mathcal{S}_n $ on which~\eqref{eq:gamma2} and \eqref{eq:gamma1} occur 
is $ \mathbb{E}\left[N\right] \approx K^n \, e^{-(L(E) +\delta) n} \gg 1 $ provided $ 0 <  \delta < \log K - L(E) $. 
Unlike~\eqref{eq:gamma2},  the conditions  $ \Im\Gamma(y;E+i\eta) \geq \xi(\alpha;E+i\eta)  $ are not rare events, and their inclusion does not modify significantly the above estimate.  

In the second case, 
by a standard large deviation estimate as in~\eqref{eq:LD}, the probability of the event~\eqref{eq:gamma1} with $ \gamma \approx - \lim_{s\uparrow 1} \frac{\partial \varphi}{\partial s}(s;E) =: \varphi'_-(1) $ is of the order $e^{-n I(\gamma) + o(1)}$ with a rate function $I(\gamma)$ which is  related to  $\varphi(s) \equiv \varphi_\lambda(s;E)$ through  the Legendre transform. 
The relevant mechanism for the occurrence of~ \eqref{eq:gamma1} is the systematic stretching of the values of $|G^{\T_x}(0,u;E+i\eta)|$ along the path $0 \preceq u \preceq x_-$.  
By the above lines of reasoning, and ignoring excessive correlations  (a step which is justified under auxiliary conditions) we arrive at the mean value estimate 
$\E\left[N  \right] \ \approx \  K^n \, \exp\left(-n \, [I(\gamma) + \gamma + \delta +o(1)]\right) $. 
This  value is  much greater than $1$ for some $\delta >0$, provided
\be \label{condA} 
\sup_{\gamma } \left[ \log K -  [ I(\gamma) + \gamma)  \right] \ > \ 0
\ee 
That is, although the probabilities of the two above events are exponentially small, given the exponential growth of $|\mathcal S_n| = K^n$, under suitable assumptions $ \E[N] \to \infty $ for $n\to \infty$.     
To see what \eqref{condA} entails, let us note that by the inverse of the Legendre transform~\eqref{eq:Legendre}: 
 \be  \label{Ltransform1}
\varphi(s;E)  \ \equiv   \varphi(s) \ = \ - \inf_{\gamma}  \left[I(\gamma) + s \gamma)   \right] 
\ee 
Thus, \eqref{condA} is the condition  
$
\varphi(1;E) > - \log K $
 which is mentioned in Theorem~\ref{lem:largeG}, and in Theorem~\ref{thm:main_phi}. 
 
 The analysis which relates to  the first condition i\/) yields the Lyapunov exponent criterion which we shall prove first.   The proof of the more complete result, which uses the condition ii\/)  is a bit more involved, and is therefore  postponed  the next section.

\subsection{Resonances based on the Lyapunov behavior}\label{subsec:Lyproof}

The aim of this subsection is to prove  the  first criterion of Theorem~\ref{lem:largeG}.  Thus, we fix the disorder parameter $ \lambda > 0 $ and the energy $ E \in \mathbb{R} $, assuming that~$ L_\lambda(E) < \log K $.  In view of the general bound $ L_\lambda (E)  > \log \sqrt{K} $, for which the strict inequality  was shown in~\cite[Thm.~4.1]{ASW} (the weak inequality is explained by~\eqref{eq:exclusion}), the assumption is equivalent to:
 \begin{equation}
	 4 \, \delta \   :=  \log K  -  L_\lambda (E)   \,  \in \left(0,\log \sqrt{K} \right) \, .
\end{equation} 

In accordance with the above heuristics, we consider the following three events.
\begin{definition}\label{def:events1}
For each $ x \in \mathcal{S}_n $ and $ \eta >0 $ we associate the following events:
\begin{enumerate}[i.]
\item The \emph{extreme deviation event}, at blow-up parameter $ \; \tau \ := \  e^{(L(E) + 2\delta) n} $
$$ E_x  \ :=  \left\{ | G(x,x;E+i\eta) | \geq  \tau \right\}  \, . $$

\item The \emph{regular decay event} at decay rate $\;  \ell := L(E) +\delta $
$$ R_x \ :=  \left\{ | G^{\mathcal{T}_x}(0,x_-;E+i\eta) |  \ \geq \ e^{-\ell n }  \right\} \, . $$
\item The \emph{$\alpha$-marginality event}, at probability $ \alpha \in (0,1) $
		$$ I_x\ :=  \bigcup_{y \in \mathcal{N}^+_x} \left\{ \Im \Gamma(y;E+i\eta) \geq \xi(\alpha;E+i\eta)  \right\} \, . $$
\end{enumerate}
\end{definition}
We will suppress the dependence of these events on $ \alpha, \eta > 0 $.  The parameter $ \tau $ is chosen such that \emph{i.}~$ 	\tau^{-1}  \, K^{n} \  = \   e^{2 \delta n }  \, $ and \emph{ii.}~in the event $ E_x \cap R_x $:
\begin{equation}
 |G(0,x;E+i\eta)|  = | G^{\mathcal{T}_x}(0,x_-;E+i\eta) | \,  | G(x,x;E+i\eta) | \geq e^{\delta n }  \, ,
 \end{equation} 
by the factorization~\eqref{eq:factor2} of the Green function. 
 The decay rate $ \ell $ is chosen so that the event $ R_x $  occurs asymptotically as $ n \to \infty $ with probability one (cf.~Theorem~\ref{thm:regGbehav}).\\

We will monitor the number of simultaneous occurrences of the three events listed above, which is given by the random number
\begin{equation}\label{eq:defN1}
	N \ := \ \sum_{x \in \mathcal{S}_n } \indfct_{E_x \cap R_x \cap I_x } \, . 
\end{equation}
Since even the divergence, for $n \to \infty$, of the expectation value ${\mathbb{E}\left[N\right]}$ does not on its own imply that the probability of $N>1$ has a positive limit.  However, such a conclusion can be drawn from suitable information on the first two moments, e.g. using the following consequence of the Cauchy-Schwarz inequality
\begin{equation}
	\mathbb{P}\left(N \geq 1 \right) \ \geq \  \ \frac{\mathbb{E}\left[N\right]^2}{\mathbb{E}\left[N^2\right]}  \, .  \label{eq:N21}
\end{equation}  
We shall next derive bounds on the first two moments which will enable the proof that the above probability is bounded below.  

\subsection{Lower bound on the mean number of resonant sites}

Our lower bound on  $ \mathbb{E}\left[N\right] $  is based on a relation of the probability of extreme deviation events to the  mean (local) density of states $ D(E ) $ associated with \emph{fully regular} Caley tree~$ \mathcal{B} $ in which \emph{every} vertex has exactly $ K+1 $ neighbors. 
This density of states is given, for almost all $ E \in \mathbb{R} $, by~\cite{MD,AK}:
\begin{equation}\label{eq:DOSx}
	D(E) \ := \ \lim_{\eta \downarrow 0} \, \frac{1}{\pi} \; \mathbb{E}\left[ \Im G^\mathcal{B}(x,x;E+i\eta) \right] \, . 
\end{equation}
Since $ \zeta \mapsto \mathbb{E}\left[  G(x,x;\zeta) \right]  $ is a Herglotz function, the limit exists for almost all $ E \in \mathbb{R} $. Moreover, due to homogeneity it is independent of $ x \in \mathcal{B} $. 
The following property is well known, cf.~\cite{AK,CL}, but very important for us.
\begin{proposition}
The support of $ D$ coincides with the almost-sure spectrum, i.e., for Lebesgue-almost all $ E \in\sigma(H_\lambda) $ one has $ D(E) > 0 $. 
\end{proposition}

Varying the potential at $x$ is a rank-one perturbation of the operator $H_\lambda(\omega)$, and the response of the corresponding Green function's diagonal element is particularly simple: 
\begin{equation}\label{eq:defsigma}
G^\mathcal{B}(x,x;\zeta) \ = \ \left( \lambda V(x) - \sigma_x(\zeta) \right)^{-1} \, , \qquad   \sigma_x(\zeta) := \zeta + \sum_{y \in \mathcal{N}_x} G^{\mathcal{B}_x}(y,y;\zeta) \, , 
\end{equation}
(which a special case of~\eqref{eq:recur_gen}).  This allows us to relate the aforementioned probability of extreme deviation events to the  density of states $D(E)$.  It is at this point that the regularity Assumption D plays a helpful role.    
\begin{lemma}\label{lem:boundD2}
For Lebesgue-almost all $ E \in \R $, under the no-ac hypothesis the following holds for all  $ x \in  \mathcal{B} $:
\begin{enumerate}
	\item  $ \Im \sigma_x(E+i0) = 0 $ almost surely.
	\item  
	$\displaystyle	D(E) \ = \ \frac{\mathbb{E}\left[ \varrho\left(\lambda^{-1} \sigma_x(E+i0) \right) \right]}{\lambda} $.  
	\item for any $\hat \tau \geq \lambda^{-1}$  and any event $Z_x $ which is independent of $V(x) $:
	\begin{equation}\label{eq:DosBref}
		 D(E )  \  \leq \ 2 c \, \lambda \, \hat\tau  \; \mathbb{P}\left( \left\{ | G^\mathcal{B}(x,x;E+i0) | \geq \hat\tau\right\} \cap Z_x  \right) + \frac{\|\varrho\|_\infty}{\lambda} \;  \mathbb{P}\left(Z_x^c\right)  \, , 
	\end{equation}
	where $c \in (0,\infty) $ is the constant from Assumption~\ref{assD}. 
\end{enumerate}
\end{lemma}
\begin{proof}
 The proof of the first assertion is based on the observation that, under the no-ac hypothesis, $ \, \Im G^{\mathcal{B}_x}(y,y;E+i0,\omega) = 0  \,$ for $ \mathbb{P} $-almost all $ \omega$, all $ x \in \mathcal{T} $ and all $ y \in \mathcal{N}_x $. This follows from the fact that the Green functions $G^{\mathcal{B}_x}(y,y;E+i0) $ associated with the neighbors, $ y \in \mathcal{N}_x $, are identically distributed to $  \Gamma(0;E+i0) $ and hence $ \Im G^{\mathcal{B}_x}(y,y;E+i0,\omega)   = 0 $ for ${\rm Lebesgue}\times \mathbb{P} $-almost all $ (E,\omega)$. 

The proof of the representation~$2.$ is based on~\eqref{eq:defsigma}. We first condition on the sigma-algebra $ \mathscr{A}_x $ generated by the random variables $ V(y) $, $ y \neq x $, and write
\begin{equation}\label{eq:rank1cor}
 \mathbb{E}\left[ \Im G^\mathcal{B}(x,x;E+i\eta) \,   | \, \mathscr{A}_x  \right] 
	\ = \   \int
	 \varrho(v) \, \Im \left( \lambda v - \sigma_x(E+i\eta) \right)^{-1}  dv   \, .
\end{equation} 
Since $ \lim_{\eta\downarrow 0} \sigma_x(E+i\eta) = \sigma_x(E+i0) $ for almost all $ E \in \mathbb{R} $ and the distribution of $  \sigma_x(E+i0) $ is continuous, Lebesgue's differentiation theorem implies that for  ${\rm Lebesgue}\times \mathbb{P} $-almost all $ (E,\omega)$:
\begin{equation}\label{eq:Lebesguediff}
	\lim_{\eta \downarrow 0 } \, \frac{1}{\pi} \int
	 \varrho(v) \, \Im \left( \lambda v - \sigma_x(E+i\eta;\omega) \right)^{-1}  dv   \ = \ \frac{\varrho( \lambda^{-1} \sigma_x(E+i0;\omega))}{\lambda} \, .
\end{equation}
This together with the dominated convergence theorem, which is based on  the Wegner bound
\begin{equation}\label{eq:Wegner}
	 \mathbb{E}\left[ \Im G^\mathcal{B}(x,x;E+i\eta) \,   | \, \mathscr{A}_x  \right]  \ \leq \ \pi \, \frac{\|\varrho\|_\infty}{\lambda} \, ,
\end{equation}
concludes the proof of  the representation~$2.$

We may now refine $2.$ by first inserting an indicator function of any event  $ Z_x $ which is independent of $ V(x) $ and its complement $ Z_x^c $.  The equalities~\eqref{eq:rank1cor} and \eqref{eq:Lebesguediff} together with \eqref{eq:Wegner} then imply:
\begin{equation}
	D(E) \  \leq \ \lambda^{-1}\,  \mathbb{E}\left[\varrho( \lambda^{-1} \sigma_x(E+i0;\omega)) \, \indfct_{Z_x} \right] +  
	\frac{\|\varrho\|_\infty}{\lambda} \, \mathbb{P}\left(Z_x^c\right)  \, . 
\end{equation}
Using Assumption~\ref{assD}, the first term on the right side is now seen to relate to the probability of extreme deviation events. More precisely, 
for any $\hat \tau \geq \lambda^{-1}$ almost surely 
\begin{align}\label{eq:useD}
	\lambda \, \varrho( \lambda^{-1} \sigma_x(E+i0;\omega) )   \ & \leq \ 2 c  \, \lambda \, \hat\tau  \int \varrho(v) \, \indfct_{|\lambda v -  \sigma_x(E+i0;\omega) | \leq \  \hat\tau^{-1} }  dv  \notag \\
	&  = \ 2 c  \, \lambda \, \hat\tau \; \mathbb{P}\left(   | G^\mathcal{B}(x,x;E+i0) | \geq  \hat\tau \,   | \, \mathscr{A}_x  \right)
\end{align}
This concludes the proof of~\eqref{eq:DosBref}. 
\end{proof}

Based on the above estimates, we may now provide a lower bound on $ \mathbb{E}\left[N\right] $.

\begin{corollary}\label{lem:lowerbound1}
For Lebesgue-almost every $ E \in \sigma(H_\lambda) $ under the no-ac hypothesis
there are $ \alpha \in (0,1) $, $ C , \eta_0 \in (0,\infty) $ and $ n_0 \geq 1 $ such that for all $ n \geq n_0 $ and $ \eta \in (0,\eta_0)$:
\begin{equation}\label{eq:lowerbound1}
	\mathbb{E}\left[ N \right] \ = \ K^n \; \mathbb{P}\left(R_x \cap E_x \cap I_x \right) \geq \ K^n \, \frac{D(E)}{ C\, \tau } \   \geq  \  \frac{D(E)}{ C} \ > 0  \, .
\end{equation}
\end{corollary}
\begin{proof}
The continuity 
\begin{equation}
\lim_{\eta\downarrow 0}\,  \mathbb{P}\left( \left\{ | G^\mathcal{B}(x,x;E+i\eta) | \geq 2\tau\right\} \cap Z_x  \right) = \mathbb{P}\left( \left\{ | G^\mathcal{B}(x,x;E+i0) | \geq 2\tau\right\} \cap Z_x  \right) 
\end{equation}
for almost every $ E \in \mathbb{R} $, guarantees the validity of~\eqref{eq:DosBref}  with $ 2c $ replaced by $ c $ and all $ \eta $ small enough. To extend this estimate to the Green function associated with the regular rooted tree $ \mathcal{T} $, we naturally embed $ \ell^2(\mathcal{T }) $ into $ \ell^2(\mathcal{B} )$ and use 
perturbation theory, the general recursion relation~\eqref{eq:recur_gen} and the multiplicativity~\eqref{gen_fact}: 
\begin{align}
	\left|G^\mathcal{B}(x,x;\zeta)^{-1} -G^\mathcal{T}(x,x;\zeta)^{-1} \right| \ & \leq \ \left| 	 \Gamma^{\mathcal{B}_x}( x_-;\zeta) -  \Gamma^{\mathcal{T}_x}( x_-;\zeta)\right| \notag \\
	& \leq \  \left| G^{\mathcal{B}_x}(0_-,x_-;\zeta) \right| \left| G^{\mathcal{T}_x}(0,x_-;\zeta) \right|  \notag \\
	& = \ \left| G^{\mathcal{B}_x}(0_-,0_-;\zeta) \right| \left| G^{\mathcal{T}_x}(0,x_-;\zeta) \right|^2 \, .  
\end{align}
For all $ E \in \mathbb{R} $ such that $ D(E) >0 $ there exists $ t > 0 $ such that  according to~\eqref{eq:pickt} the event $ \hat B_x := \{ \left| G^{\mathcal{B}_x}(0_-,0_-;E+i\eta) \right| \leq t \} $ has for all $ \eta > 0 $ a probability of at least 
\begin{equation}
	\mathbb{P}( \hat B_x  ) \ \geq \ 1 -  \frac{\lambda D(E)}{8 \, \| \varrho \|_\infty} \ > 0 \, . 
\end{equation}
Moreover, according to Theorem~\ref{thm:regGbehav} and since $e^{-\delta n} \tau^{-1} = K^{-n}  > e^{-2 n L(E) } $, there is $ n_0 \geq 1 $ and $ \eta_0 \in (0,\infty) $ such that for all $ n \geq n_0 $ and $\eta \in (0,\eta_0 ) $ the event $ \hat R_x := \{  \left| G^{\mathcal{T}_x}(0,x_-;E+i\eta) \right| \leq \sqrt{e^{-\delta n} \tau^{-1}} \} $ has a probability of at least
\begin{equation}
	\mathbb{P}( \hat R_x  ) \ \geq \ 1 - \frac{\lambda D(E)}{8 \, \| \varrho \|_\infty} \ > 0 \, . 
\end{equation}  
Summarizing the above estimates, we conclude that there is $ n_0 \geq 1 $ and $ \eta_0 \in (0,\infty) $ such that for all $ n \geq n_0 $ and $\eta \in (0,\eta_0 ) $  and any event $ Z_x $ which is independent of $ V(x) $:
\begin{align}
D(E)  \ \leq & \ c \, \lambda \, \tau \; \mathbb{P}\left( \left\{ | G^\mathcal{B}(x,x;E+i\eta)^{-1} | \leq (2\tau)^{-1}\right\} \cap  \hat B_x  \cap  \hat R_x  \cap Z_x  \right) \notag \\
& \quad  +  \frac{\|\varrho\|_\infty}{\lambda} \;  \mathbb{P}\left( \hat B_x ^c \cup  \hat R_x^c \cup  Z_x^c\right)  \notag \\
\leq & \ c \, \lambda \, \tau \, \mathbb{P}\left( E_x  \cap Z_x  \right)  +  \frac{\|\varrho\|_\infty}{\lambda} \;  \mathbb{P}\left(Z_x^c\right)  + \frac{1}{4} \, D(E) \, . 
\end{align}
We apply this bound to $ Z_x = R_x \cap I_x $. Since $ \mathbb{P}\left( R_x^c \cup I_x^c \right)  \leq \mathbb{P}\left( R_x^c \right) + \mathbb{P}\left( I_x^c\right)  \ \leq \  \mathbb{P}\left( R_x^c \right) + 1 - \alpha $.  By Theorem~\ref{thm:regGbehav},  there is $ n_1 \geq n_0 $ and $ \eta_1 \in (0,\eta_0] $ such that for all $ n \geq n_1 $ and $ \eta \in (0,\eta_1) $
\begin{equation}
	\mathbb{P}\left(R_x \right) \ \geq \ 1 - \frac{\lambda D(E)}{8 \, \| \varrho \|_\infty} \ > 0 \, . 
\end{equation} 
Choosing $ \alpha := 1 -  \frac{\lambda D(E)}{8 \, \| \varrho \|_\infty} $ completes the proof of~\eqref{eq:lowerbound1}. 
\end{proof}

\subsection{The enabling second moment upper bound}
The mere fact that the mean number of events diverges, for $n\to \infty$ (cf.~\eqref{eq:lowerbound1})  does not yet imply that such events do occur with uniformly positive probability.   The  alternative is that the divergence reflects an increasingly rare but also increasingly correlated occurrence of these events. 
To prove that  the resonances do occur regularly,  on sufficiently large spheres $ \mathcal{S}_n $, we use the second-moment method which is based on the 
following estimate. 

\begin{lemma} \label{lem:upperbound1}
Assuming $ L(E) < \log K $,
 there is $ C \in (0,\infty) $ such that  for all $ n \geq 1$, all $ \eta >0 $ and all $ \alpha \in (0,1) $:
\begin{equation}
	\mathbb{E}\left[N(N-1) \right] \ \leq \ C \, \tau^{-2} \,  K^{2n} \,  .
\end{equation}
\end{lemma}
\begin{proof}
	Throughout the proof appearing constants $C\in (0,\infty) $ will be independent of $ n$, $\eta $ and $\alpha$. We start from the observation that 
	\begin{align}
		\mathbb{E}\left[N(N-1) \right] \ = &   \sum_{ \substack{x, y  \in \mathcal{S}_n \\ x\neq y}} \mathbb{P}\left(R_x \cap E_x \cap I_x\cap R_y\cap E_y \cap I_y \right) \ \leq  \sum_{ \substack{x, y  \in \mathcal{S}_n \\ x\neq y}} \mathbb{P}\left(E_x \cap E_y  \right)  \, . \label{eq:N(N-1)} 
	\end{align}
	The probability in the right side is estimated using the weak-$L^1$ bound for pairs of Green function in  Theorem~\ref{thm:twoL1} below. Denoting
		by $ \mathscr{A}_{xy} $ the sigma-algebra generated by the random variables $ V(u) $, $ u \not\in \{x ,y\} $, it yields
	\begin{align}
		\mathbb{P}\left(E_x \cap E_y  \right) \ = & \ \mathbb{E}\left[\mathbb{P}\left( E_x \cap E_y  \, \big| \, \mathscr{A}_{xy} \right)\right] \notag \\
		 \leq & \ \frac{C}{\tau}  \left(  \frac{1}{\tau} + \mathbb{E}\left[\min\left\{ 1 , | G^{\mathcal{T}_{x,y}}(x_-,y_-;E+i\eta)| \right\}\right] \right) \, ,
	\end{align}
	with some constant $ C \in (0,\infty) $. The first term is already of the desired form since the number of terms in the sum in~\eqref{eq:N(N-1)} is bounded by $ K^{2n} $. 
	To estimate the second  term we use $ \min\{ 1, |x| \} \leq |x|^s$  valid for any $ s\in [0,1] $. Choosing 	
	\begin{equation}\label{eq:picks1}
		s  := \  \frac{L(E) + 2 \delta}{\log K} \ \in \ (0,1) \, , 
	\end{equation}
	we estimate the factional-moment with the help of the finite-volume bounds~\eqref{eq:finitevolume} and the upper bound in~\eqref{eq:exclusion}:
	\begin{equation}
		\mathbb{E}\left[ | | G^{\mathcal{T}_{x,y}}(x_-,y_-;E+i\eta)|^s \right] \ \leq \ C \, K^{- \frac{s}{2} \dist(x,y) }
	\end{equation}
	with some constant $C \in (0,\infty) $. The corresponding sum contributing to~\eqref{eq:N(N-1)} is estimated by fixing $ x \in  \mathcal{S}_n $ and summing over the distance of the least common ancestor of $ x $ and $ y $ to the root:
	\begin{align}
		\sum_{ \substack{x, y  \in \mathcal{S}_n \\ x\neq y}}  \mathbb{E}\left[ | | G^{\mathcal{T}_{x,y}}(x_-,y_-;E+i\eta)|^s \right]  \  
		\leq & \  C \, K^{n} \, \sum_{j=0}^{n-1} K^{n-j} \, K^{- s (n-j)}  \notag   \\[-1ex]
		\leq & \ C \,  K^{(2-s) n }  \  =  \  C \,  \tau^{-1}  K^{2n}  \, ,  
	\end{align}
	where the last inequality is based on~\eqref{eq:picks1}. 
\end{proof}


We are now ready for the proof of the main result of this section.

\begin{proof}[Proof of Theorem~\ref{lem:largeG}; the Lyapunov exponent criterion] 
By Corollary~\ref{lem:lowerbound1} and Lemma~\ref{lem:upperbound1}, there are $ \alpha \in (0,1) $ (which is one of the parameters in the definition of $N$), $ C , \eta_0 \in (0,\infty) $ and $ n_0 \geq 0 $ such that for all $ n \geq n_0 $ and $ \eta \in (0,\eta_0)$:
\begin{align}
\frac{\mathbb{E}\left[N^2\right]}{\mathbb{E}\left[N\right]^2} \ = & \ \frac{1}{	\mathbb{E}\left[N\right]} \ + \ \frac{\mathbb{E}\left[N(N-1)\right]}{\mathbb{E}\left[N\right]^2} \ \leq  \ C \, .    
\end{align}
Hence, second-moment bound~\eqref{eq:N21} allows us to conclude that 
$\mathbb{P}\left(N \geq 1 \right) \ \geq \  C^{-1}$ uniformly in $n>n_0$ and $ \eta \in (0,\eta_0)$.   

However, whenever  $ N \geq 1 $ one may conclude that the quantity which appears in the left side of~\eqref{eq:largeG}  satisfies
\begin{equation} \label{eq:bnd}
\max_{x \in  \mathcal{S}_n}\,  |G(0,x;E+i\eta)| \ \indfct_{\max_{y \in \mathcal{N}_x^+} \Im\Gamma(y;E+i\eta) \geq \xi(\alpha;E+i\eta) } \geq \ e^{\delta n } \, . 
\end{equation}
Taken together, \eqref{eq:bnd} and the above probability estimate directly imply the part of Theorem~\ref{lem:largeG} which relates to the Lyapunov exponent criterion, with $2p_0 = C^{-1}$).  
\end{proof} 
As was shown in Section~\ref{lem:largeG}, the above result implies the Lyapunov exponent criterion which is stated in Theorem~\ref{thm:mainL}.

\section{Resonances enhanced by large deviations} 
 \label{sec:GFextrev}
 
As explained  in the introduction, while the Lyapunov exponent criterion is very useful it does not yet cover the full regime of extended states.    Our next aim is to establish an extended version of this criterion,  improved through the incorporation in the argument of the large deviation considerations.  The result is stated above as the second part of  Theorem~\ref{lem:largeG}.  We now turn to its proof, following the outline which is given in Section~\ref{sec:heuristic}. 
The strategy has much in common with the derivation of the Lyapunov exponent criterion, however  the proof involves some additional technicalities.  Since the applications which are discussed in the introduction rely on just the Lyapunov exponent criterion,  only the more dedicated reader may wish to follow this Section. 

\subsection{Selection of auxiliary parameters}\label{subsec:fixparameters}

 For the remainder of this subsection, we fix the disorder parameter $ \lambda > 0 $ and an energy $ E \in \mathbb{R} $ such that 
 $\varphi(t) \equiv \varphi(t;E) = \lim_{\eta \downarrow 0} \varphi(t;E+i\eta)$ exists for all $ t \in [-\varsigma,1) $ and~\eqref{phi_ac} holds, i.e.,
 \begin{equation}\label{eq:defDelta}
	\Delta\  := \log K + \varphi(1;E)  \,  \in \left(0, \tfrac{1}{2} \log K \right) \, .
\end{equation} 
Due to the convexity of $\varphi(s)$ and~\eqref{eq:exclusion}, under the assumption \eqref{eq:defDelta} the left derivative of $\varphi$ satisfies (see Figure~\ref{fig:phi}): 
\be  \label{eq:phi'}
0 \  <  \ - \varphi_-'(1) \  \le  \  \Delta  \, . 
\ee

We proceed by associating to the given $ \lambda $ and $ E  $  certain  parameters ($\gamma $,  $ \beta $, $\kappa $,  $ \epsilon $, and $ \tau $) which will also be kept fixed  for the remainder of this section. These parameters feature in the definition of the resonance events which will be associated with vertices on the sphere $ \mathcal{S}_{n}$ of radius $ n \in \mathbb{N} $. To control  the correlations among such events  we restrict to  vertices on the thinned sphere 
$	\mathcal{S}_{n}^\kappa \subset \mathcal{S}_{n}  $
associated with the parameter $ \kappa   $ which we pick in the range: 
\begin{equation}\label{eq:setdelta}
	 \kappa \ \in \ \left( 0 \, , \, \min\left\{\tfrac{\Delta}{16 \, \ell}, \tfrac{1}{4} \right\} \right)  \, , 
\end{equation} 
where $ \ell > L(E) $ is fixed (largely arbitrary).
The thinned sphere $ \mathcal{S}_{n}^\kappa $, whose radius shall be larger than $  4 \, \lceil \kappa^{-1} \rceil $, is characterized by the \emph{length scales}
$ n_\kappa := 2 \, \lfloor{\tfrac{\kappa n}{2} \rfloor} \, \in  2\, \mathbb{N} $ and $ N_\kappa := n -  n_\kappa $. 
The first one is only a fraction of the second length scale, i.e.
\begin{equation}\label{eq:estimaten}
	\tfrac{1}{2}\, \kappa \,  n \ \leq \ n_\kappa \leq  \ \kappa \, n \, , \quad n_\kappa \leq \tfrac{\kappa}{1-\kappa} \, N_\kappa \leq \tfrac{4}{3} \, \kappa \, N_\kappa \, .  
\end{equation}
Then $ \mathcal{S}_{n}^\kappa $ is uniquely determined by having $ K^{N_\kappa} $ vertices with $ 2 n_\kappa  +1$ vertices separating them, cf.\ Figure~\ref{fig:largedevsetup}. \\

We now pick a value $ s \in (0,1)$   at which the free energy function $ t \mapsto  \varphi(t) $ is differentiable, and such that
\begin{enumerate}[a)]
	\item  the derivative at $s$, satisfies
	\begin{equation}\label{eq:defgamma}
		\gamma  := -  \varphi'(s) \; \geq \ \Delta \ > \ 0 \, ,   
	\end{equation}
\item 	the following condition holds 
	\begin{equation}\label{eq:assumphi}
		I(\gamma) + \gamma  \ = \ -\left[  \varphi(s) +(1-s) \varphi'(s) \right] \  \leq \ \log K - \tfrac{7}{8} \Delta \, ,  
	\end{equation}
	\item and in addition $(1-s) < 1/16$ and $ \varphi(s) < - \tfrac{1}{2}\log K $.
\end{enumerate}
In view of \eqref{eq:defDelta} and \eqref{eq:phi'}, and the convexity of $\varphi$, the above conditions are satisfied at a dense collection of values of $s$ approaching $1$ from below (see Figure~\ref{fig:phi}).   
(Condition~c) is only  imposed to simplify some of the estimates.)

\begin{figure}
\begin{center}
\includegraphics[scale=.3]{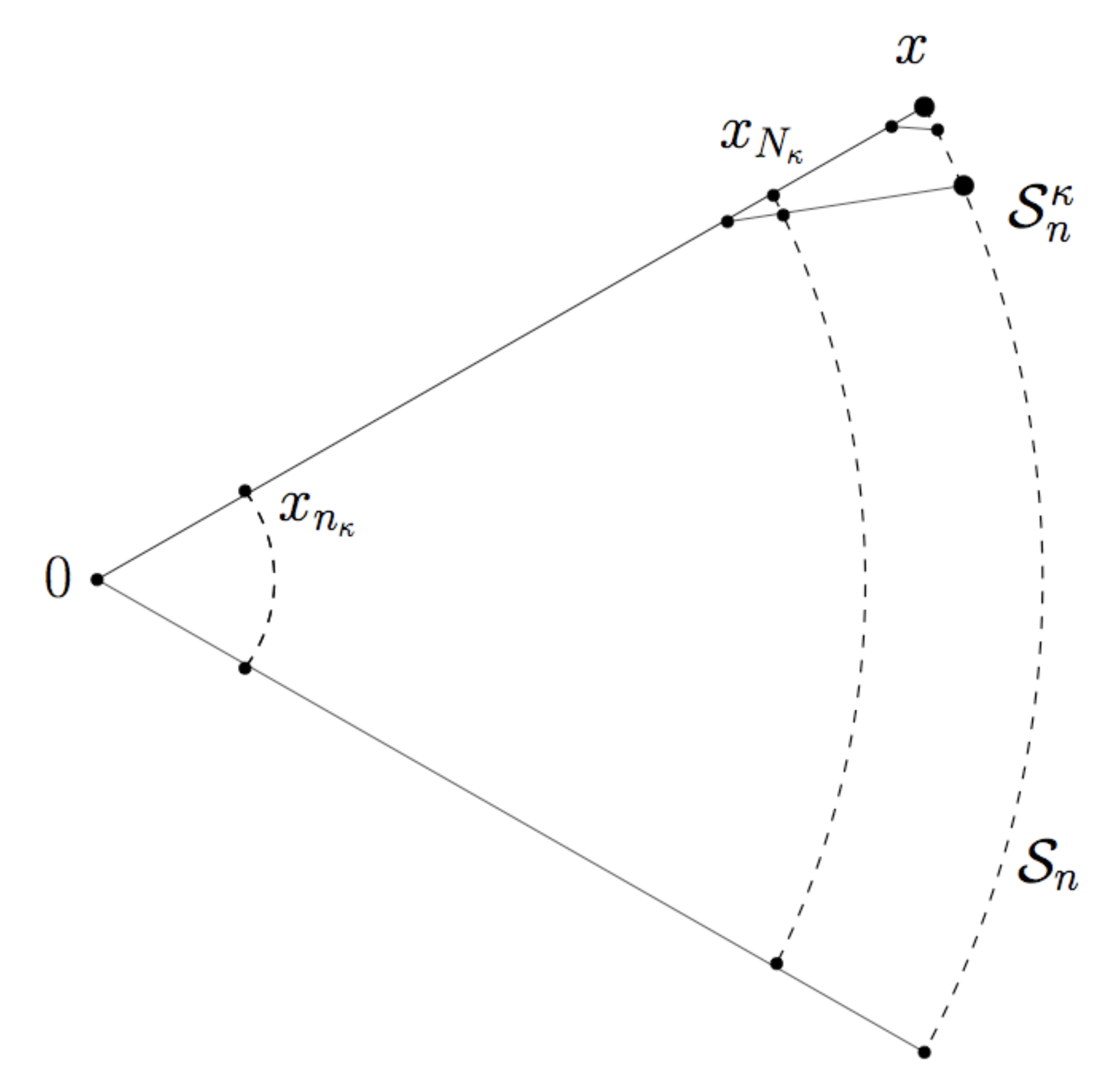} 
\caption{The geometry of the resonance-boosted large-deviation event.}\label{fig:largedevsetup}
\end{center}
\end{figure}

The parameter $ \gamma $ will be used as a target-value for the decay of the Green function in the large deviation events $L_x $ defined below. 
For any site  $ x \in \mathcal{S}_n$ we label  the vertices of the unique path from the root to $x$ as  $ x_0 = 0, x_1, \dots , x_n = x$, and we denote as  \begin{equation}
\widehat{\mathcal{T}}_x  :=  \mathcal{T}_{x_{n_\kappa-1},x} 
\end{equation}
the tree  truncated beyond the segment of length $ N_\kappa $ whose end points are 
  $ \{ x_{n_\kappa-1},x \}$   (cf.~Figure~\ref{fig:largedevsetup}). 
  Associated with this segment there are the two collections of variables $\{ \Gamma_+(j;\eta)\}_{j=1}^{N_\kappa}$ and $\{ \Gamma_-(j;\eta)\}_{j=1}^{N_\kappa}$: 
   \begin{align} 
  \label{eq:deftriangulara}
	\Gamma_+(j;\eta)\   & :=\  G^{\T_{ x_{n-j-1}, x}}( x_{n-j}, x_{n-j};E+i\eta)  \, , \notag \\
	\Gamma_-(j;\eta) \ & := \ G^{\T_{x_{n_\kappa -1},x_{n_\kappa +j }}}(x_{n_\kappa -1+j},x_{n_\kappa -1+j};E+i\eta)\,  , 
\end{align}
such that by~\eqref{gen_fact}: 
\begin{equation}\label{eq:consist}
	G^{ \widehat{\mathcal{T}}_x }(x_{n_\kappa}, x_{n-1} ; E +i\eta) \ = \ \prod_{j=1}^{N_\kappa} \Gamma_+(j;\eta) \ = \ \prod_{j=1}^{N_\kappa} \Gamma_-(j;\eta) \, . 
\end{equation}
\begin{definition}\label{def:ldev}  
We refer to the following as the \emph{large-deviation events} associated with sites $ x \in \mathcal{S}_n $  and $  \eta ,\, \epsilon> 0 $  
\begin{equation}
	L_x:= L^{(\rm bc)}_x\, \cap  \bigcap_{k=\tfrac{1}{2} n_\kappa}^{N_\kappa} \left( L^{(k,+)}_x \cap L^{(k,-)}_x  \right) \, , 
\end{equation}
where  for any $ k \in \{1, \dots , N_\kappa\}$: \quad
$\displaystyle	L^{(k,\pm)}_x  \  := \ \Big\{  \prod_{j= 1}^k |  \Gamma_\pm(j;\eta)| \, \in \,  e^{-\gamma k}  \big[e^{-\epsilon k} , \, e^{\epsilon k}\big]  \Big\} $, \\
	and $ L^{(\rm bc)}_x \  := \  \left\{   |\Gamma_+(N_\kappa;\eta) |  \ \leq \tfrac{b}{2} \right\} \cap \left\{  |\Gamma_-(N_\kappa;\eta) |  \ \leq \tfrac{b}{2} \right\} $. 
\end{definition}	
We will suppress the dependence on $\eta $ and $ \epsilon $ (whose value is fixed below). 

The boundary events $  L^{(\rm bc)}_x $ play a role in the following context: 
{\it i.\/}~the lower bound on the probability of $R_x $ given below in Lemma~\ref{lem:RandL}, and {\it ii.\/}~the estimate \eqref{eq:usedLbc}  on the size of the self-energy at $ x $ are  derived only under the condition $  L^{(\rm bc)}_x $.  
The parameter 
$ b  $ 
  is fixed at a value large enough so that
    \begin{enumerate}[a)]
    \item $ b \geq \frac{ 2 \|\varrho\|_\infty}{\lambda} \max \big\{ 16 , \big( 1- (3/4)^K \big)^{-1} \big\} $, and
  \item $ 
\mathbb{P}_s\left( L^{(\rm bc)}_x \right) \geq \tfrac{7}{8} $, cf.~\eqref{eq:defPs},
\end{enumerate}
the latter being possible thanks to~\eqref{eq:pickt2}. (The numbers are largely arbitrary.)

To fix the parameter $ \epsilon$, we invoke the following  large-deviation statement  which is derived in the Appendix~\ref{App:Ldevs}. 
\begin{theorem}\label{cor:ldev1}
	For any $ \epsilon>0  $ there is $ \eta_0 > 0 $ and $ n_0 >0 $ such that for all $ \eta \in (0,\eta_0) $ and all $n = \dist(x,0) \geq k \geq n_0$:
	\begin{align} 
	  \mathbb{P}\left( L_x(\eta;\epsilon) \right) \  &\geq \ e^{-N_\kappa \left( I(\gamma) + 2 \epsilon \right)}  \, , \label{eq:Llbound1} \\
	  	  \mathbb{P}\left( L^{(k,\pm)}_x(\eta;\epsilon)  \right) \ & \leq \  e^{-  (I(\gamma)-2 \epsilon) \, k  }\, .  \label{eq:Llbound2}
	\end{align}
	\end{theorem}

We now fix   $ \epsilon$ at a value at which: 
\begin{equation}\label{eq:eps'}
	2  \epsilon  \in \left( 0 \, , \,\min\left\{\tfrac{\Delta}{24}, \tfrac{\kappa \, \Delta}{4}\right\} \right) \, . \end{equation} 
This parameter will be used in controlling the probabilities of various large deviation events.

Before turning to the main definitions, we introduce yet another event which refers to the behavior of the Green function  
between $x_0$ and $x_{n_\kappa-1}$, for which we require  the  (largely arbitrary) minimal decay rate 
$	 \ell >  L(E)  $
combined with a condition at an end point.  

\begin{definition}\label{def:rdev}
We refer to the following as  the \emph{regular events} associated with  sites $ x \in \mathcal{S}_n $  and $ \eta > 0 $: 
\begin{equation}
R_x :=   R^{(\rm bc)}_x \cap \left\{   \, | G^{\mathcal{T}_{x}}(0,x_{n_\kappa-1};E+i\eta) |  \in \big[ e^{-   n_\kappa \ell}, 1 \big]\right\}  
\end{equation}
where 
$
	R^{(\rm bc)}_x := \left\{   \, | G^{\mathcal{T}_{x}}(x_{n_\kappa-1},x_{n_\kappa-1};E+i\eta) | \leq \tfrac{b}{2} \right\} $.
\end{definition}
This event is regular in the sense that it occurs with a probability of order one, which is independent of $ n $, cf.~Theorem~\ref{thm:regGbehav}. 
The reason for its inclusion in the paper is mainly of technical origin: in the subsequent proof of a second moment bound, Theorem~\ref{thm:N2test} below, we cannot allow the large deviation event $ L_x $ to extend down to the root, but we nevertheless need some control on the Green function on this segment.\\
Having fixed the basic parameters, we now turn to the precise definition of the events.
\begin{definition}\label{def:RBLD}
For each  
$ x \in \mathcal{S}_{n}$ and $ \eta >0$ we define
\begin{enumerate}[i.]
\item the \emph{resonance-boosted large-deviation event}, 
\begin{equation}
   D_x  :=  E_x \cap L_x  \cap R_x  
\end{equation}	
which consists of the following three events:
\begin{enumerate}[a)]
\item extreme deviation event with blow-up scale $   \tau :=   \exp\left(\left(\gamma + \frac{3}{4}\, \Delta\right)  \,  N_\kappa \right) $:
$$ E_x  \ := \ \left\{ | G(x,x;E+i\eta) | \geq \tau \right\}  \, , $$
\item large deviation event:  \quad $ L_x  $ \quad (cf.~Definition~\ref{def:ldev})
\item regular event: \quad $ R_x $ \quad (cf.~Definition~\ref{def:rdev})\end{enumerate}
\item the \emph{$\alpha$-marginality event} at probability  $ \alpha \in (0,1) $:
$$
		 I_x\ := \  \bigcup_{y \in \mathcal{N}^+_x} \left\{ \Im \Gamma(y;E+i\eta) \geq \xi(\alpha;E+i\eta)  \right\}  	\, . 	
$$
\end{enumerate}
The joint event $ D_x  \cap I_x $ will be referred to as a \emph{resonance event} at~$ x $.  
\end{definition}
\noindent
Several remarks are in order:
\begin{enumerate}
\item The resonance-boosted large-deviation events are tailored so that in the event $ D_x $ the Green function associated with the root and $ x $ exhibits an exponential blow-up. Namely, by the factorization property of the Green function,
\begin{align}\label{eq:factorinproof}
			  G(0,x;\zeta)  \ & = \ G^{\mathcal{T}_{x}}(0,x_{n-1};\zeta) \, G(x,x;\zeta)  \notag  \\
			& = \ G^{\mathcal{T}_{x}}(0,x_{n_\kappa-1};\zeta) \, G^{\widehat{\mathcal{T}}_{x}}(x_{n_\kappa},x_{n-1};\zeta) \,  G(x,x;\zeta)  \, . 
	\end{align}
	 For $ \zeta = E +i\eta $, the first term is controlled by $R_x $. 
	   The large deviation event $ L_x $ controls the second factor and the extreme fluctuation event $ E_x$ compensates for the decay of the first 	two terms. Using~\eqref{eq:estimaten}, \eqref{eq:setdelta}, and \eqref{eq:eps'}, we hence arrive at the estimate:
	\begin{align}\label{eq:Gexplodes}
		\left| G(0,x;E+i\eta) \right| \ & \geq \ e^{- n_\kappa \ell } \, e^{-(\gamma+\epsilon) N_\kappa} \, \tau \notag \\
		& \geq \ \exp\left( N_\kappa \left( \tfrac{3}{4}\Delta -\epsilon- \tfrac{4}{3} \kappa  \ell \right)\right)  \notag \\
			& \geq \  \exp\left(\tfrac{1}{2} \, \Delta \,  N_\kappa \right) \ \geq \ \exp\left(\tfrac{3}{8}\, \Delta \, n \right) \, .
	\end{align}
	`\item The choice of the blow-up scale $ \tau $ is tailored to: {\it i.\/}~compensate the decay of the Green function on the segment preceeding $ x $, cf.~\eqref{eq:Gexplodes}, and {\it ii.\/}~ensure that for $ n $ large enough and $ \eta $ small enough:
\begin{align}\label{eq:expesttau}
	\tau^{-1}  \, K^{N_\kappa} \, \mathbb{P}\left( L_x \right) \ & \geq \ \exp\left( N_\kappa \left( \log K - \left( \gamma + I(\gamma) \right) - 2\epsilon - \tfrac{3}{4} \Delta \right)\right) \notag \\
	& \geq \exp\left( N_\kappa \, \tfrac{\Delta}{16} \right)  \, ,
\end{align}
by~\eqref{eq:Llbound1}, \eqref{eq:assumphi} and~\eqref{eq:eps'}.  The fact that this term can be made large as $ n \to \infty$ will be essential in the subsequent argument.

	  	   \item We recall from Definition~\ref{def:xi} that the value $  \xi(\alpha;E+i\eta) $ ensures that  $
	   	\mathbb{P}\left( I_x \right) \geq  \alpha $.
           \end{enumerate}

 \subsection{The strategy}\label{sec:lemkeyproof}
 
Postponing the proof of the occurrence of the above resonance events, the proof of our key statement, the large-deviations criterion of Theorem~\ref{lem:largeG}, is along the same lines as in the Lyapunov regime. 

\begin{proof}[Proof of Theorem~\ref{lem:largeG} the large-deviation criterion] 
We monitor the number
\begin{equation}
 N \ := \ \sum_{x\in \mathcal{S}_n^\kappa} \indfct_{D_x \cap I_x  }  
 \end{equation} 
 of resonances on the thinned sphere and note that the event $ N \geq 1 $ implies the event the right side of~\eqref{eq:largeG} for $ \delta = \frac{3}{8} \Delta >0 $ using~\eqref{eq:Gexplodes}.

According to Theorems~\ref{thm:1Mom} and~\ref{thm:N2test}, there are $ \alpha \in (0,1) $, $ C , \eta_0 \in (0,\infty) $ and $ n_0 \geq 0 $ such that for all $ n \geq n_0 $ and $ \eta \in (0,\eta_0)$:
\begin{align}
\frac{\mathbb{E}\left[N^2\right]}{\mathbb{E}\left[N\right]^2} \ = \ \frac{1}{\mathbb{E}\left[N\right]} + \frac{\mathbb{E}\left[N(N-1)\right]}{\mathbb{E}\left[N\right]^2} \ \leq  \ C \, .
\end{align}
Together with~\eqref{eq:N21}, this concludes the proof. 	
\end{proof}

The second-moment method on which the the above proof is based requires a lower bound on the mean number of events as well as an upper bound on their second moment. 
These will be the topics of the remaining subsections.

\subsection{The mean  number of resonant sites }\label{subsec:AvN}

The main idea behind a lower bound on the average number of resonances is that the probability of the occurrence of the extreme fluctuation $E_x $ is of order $ \tau^{-1} $. Rewriting this event,
 \begin{equation}\label{eq:rewriteE}
  E_{x} = \left\{  \left| \lambda V(x) - \sigma_x(E+i\eta) \right|  \leq \tau^{-1} \right\} 
  \end{equation}
thereby exposing the dependence of  $ G(x,x;\zeta) $ on the potential at $ x $ and on
\begin{equation}  \label{eq:selfenergy}
\sigma_x(E+i\eta):= E + i \eta +  \sum_{y \in \mathcal{N}_x} G^{\mathcal{T}_x}(y,y;\zeta)    \, ,
\end{equation}
one realizes that  if the latter has a non-zero imaginary part,  the Green function stays bounded and no resonance mechanism kicks in. 
On the other hand,  in the event $ S_{x} \cap T_{x} $,  where
	\begin{align}
	S_x & := \bigcap_{y \in \mathcal{N}_x} S_{x}(y) \, , \quad \mbox{with}\quad S_{x}(y)  :=  \left\{ \left| G^{\mathcal{T}_x}(y,y;\zeta)\right|  \leq b \right\} \notag \\
	T_{x} & :=  \left\{ \Im \sigma_x(E+i\eta) \leq\ (2\tau)^{-1}\right\} \,  ,  
\end{align}
the imaginary part of the term in the right side of~\eqref{eq:rewriteE}  is bounded  by $ (2\tau)^{-1} $ and the real part is bounded by  $ (K+1)\, b $. As a consequence, we may estimate the conditional probability of $ E_x $ conditioned on the sigma algebra $ \mathscr{A}_x  $ generated by the random variables $ V(y) $, $ y \neq x $:
\begin{align}
	\mathbb{P}\left(E_{x}  \, \big| \,  \mathscr{A}_x  \right) \ 
	&  \geq  \  \indfct_{S_{x} \cap T_{x}  } \; \mathbb{P}\left(  \left| \lambda V(x) - E - \Re \sigma_x(E+i\eta) \right|  \leq  \tfrac{1}{2 \tau}  \; \big| \;   \mathscr{A}_x \right)  \notag \\
	& \geq \  \indfct_{S_{x} \cap T_{x}  }  \inf_{ |\sigma | \leq    (K+1) \, b }  \mathbb{P}\left(  \left| \lambda V(x) - E -  \sigma \right|   \leq  \tfrac{1}{2 \tau} \; \big| \;   \mathscr{A}_x \right)   \notag \\
	& \geq \   \varrho_b  \, \tau^{-1} \,  \indfct_{S_{x} \cap T_{x}  } \, . \label{eq:lbPEx}
\end{align}
where   the last estimate relied on Assumption~\ref{assD} and we introduced
 \begin{equation}\label{def:rhob}
			 \varrho_b :=  \inf_{v \in  (K+1) \, [ - b ,\, b\,  ] } (c \lambda)^{-1} \ \varrho\Big(\frac{v +E}{\lambda}\Big) \  > 0 \, . 
		\end{equation} 
		Now, $ S_x $ is a regular event, i.e., it occurs with positive probability which is independent of $ n $. Under the no-ac hypothesis the probability of the event $ T_x $ is (arbitrarily) close to one. 

\begin{lemma}\label{lemma:realselfenergy2}
Under the no-ac hypothesis, $ \, \Im \sigma_x(E+i0,\omega) = 0  \,$ for $ \mathbb{P} $-almost all $ \omega$ and all $ x \in \mathcal{T} $.
\end{lemma}
\begin{proof}
Recall that $\sigma_x$ coincides  with the sum~\eqref{eq:selfenergy} of  Green functions associated with the  neighbors of $ x $.
The Green function associated with the  forward neighbors, $ y \neq x_- $, are identically distributed to $  \Gamma(0;E+i0) $ and hence $ \Im G^{\mathcal{T}_x}(y,y;E+i0,\omega) = 0 $ for ${\rm Lebesgue}\times \mathbb{P} $-almost all $ (E,\omega)$. The Green function associated with the backward neighbor $ x_- $ differs  by a finite-rank perturbation from a variable which is identically distributed to $  \Gamma(0;E+i0) $ (i.e., the surgery which renders the rooted to into a full tree). Since finite-rank perturbations do not change the  $ac$ spectrum, we also conclude $ \Im G^{\mathcal{T}_x}(x_-,x_-;E+i0,\omega) = 0 $ for ${\rm Lebesgue}\times \mathbb{P} $-almost all $ (E,\omega)$. 
\end{proof}

The bound~\eqref{eq:lbPEx}  quantifies the essence of the resonance mechanism and leads to the following
  \begin{theorem}\label{thm:1Mom}
	Under the no-ac hypothesis, for every $ n $  large enough there exists $ \eta_0 >0 $ such that for all $ \eta \in (0,\eta_0) $, and $ \alpha \in [1/2,1) $ and all $ x \in \mathcal{S}_n $:
 \begin{equation}\label{eq:lowerboundN}
 	\mathbb{E}\left[ N  \right] \ =  K^{N_\kappa}  \, \mathbb{P}\left( D_x \cap I_x  \right)  \geq \ \tfrac{1}{16} \, \varrho_b \, \tau^{-1} \, K^{N_\kappa} \, \mathbb{P}\left(L_x\right)   \, .
 \end{equation}
The right side can be made arbitrarily large by choosing $ n $ sufficiently large. 
\end{theorem}
\begin{proof}
In order to estimate the probability of the joint occurrence of the events $ D_x $ and~$ I_x $, we first condition on the sigma algebra $ \mathscr{A}_x $ and use \eqref{eq:lbPEx} to obtain:
\begin{align}
\mathbb{P}\left( D_x \cap I_x  \right) \ & = \ \mathbb{E}\left[ \indfct_{R_x\cap L_x\cap I_x} \mathbb{P}\left( E_x \, \big| \, \mathscr{A}_x \right)\right] \notag \\
	& \geq \   \varrho_b \, \tau^{-1} \,  \mathbb{P}\left( R_x\cap L_x\cap I_x \cap S_x \cap T_x \right) \notag \\
	& \geq \    \varrho_b \, \tau^{-1} \, \left[ \mathbb{P}\left( R_x\cap L_x\cap I_x  \cap S_x \right) - \left(1- \mathbb{P}\left( T_x \right)\right) \right] \notag \\
	 & = \      \varrho_b  \, \tau^{-1} \, \left[   \mathbb{P}\left( R_x\cap L_x\cap S_x^-\right) \, \mathbb{P}\left(  I_x\cap S_x^+\right) + \mathbb{P}\left( T_x \right) -1 \right]  \, , \label{eq:DI1}
\end{align}
where we abbreviated $ S_x^- := S_x(x_-) $ and $ S_x^+ := \bigcap_{y \in \mathcal{N}_x^+} S_x(y) $. 
The first term simplifies using:
\begin{enumerate}[i)]
\item the inclusion $ R_x\cap L_x \subset S_x^- $.  This derives from second order perturbation theory. More precisely, in the event $ R_x \cap L_x $ the term corresponding to the backward neighbor $ x_- $ of $ x $ is bounded according to
\begin{align} \label{eq:usedLbc}
		  |G^{\mathcal{T}_{x}}(x_-,x_-;E+i\eta) | 
		   & \leq  \ |G^{\widehat{\mathcal{T}}_x}(x_-,x_-;E+i\eta) | \notag \\
		  & \mkern-30mu + |G^{\mathcal{T}_{x}}(x_{n_\kappa-1},x_{n_\kappa-1};E+i\eta) | \; |G^{\widehat{\mathcal{T}}_x}(x_{n_\kappa},x_-;E+i\eta) |^2  \notag \\
		&  \leq \ \tfrac{b}{2} + \tfrac{b}{2}  =  b\, \, . 
	\end{align}
\item the estimate $  \mathbb{P}\left(  I_x\cap S_x^+\right) \geq  \mathbb{P}\left(  I_x\right) + \mathbb{P}\left(  S_x^+\right) - 1 \geq \alpha + ( 1- \|\varrho\|_\infty (\lambda b )^{-1} )^K -1 \geq \frac{1}{4} $. Here the last inequality used $ \alpha \geq 1/2 $ and the particular choice of $ b $.
\end{enumerate}
To proceed with our estimate on the right side in~\eqref{eq:DI1} we use Lemma~\ref{lem:RandL} below which guarantees that for some $ \eta_0>0 $ and some $ n_0 \in \mathbb{N} $ and all $ \eta \in (0,\eta_0 ) $ and $n \geq n_0 $:
\begin{equation}\label{eq:probrandl}
 \mathbb{P}\left( R_x\cap L_x\cap S_x^-\right)  \ = \ \mathbb{P}\left( R_x\cap L_x\right) \ \geq  \ \tfrac{1}{2} \;  \mathbb{P}\left( L_x\right)  \, . 
\end{equation}
We now use Lemma~\ref{lemma:realselfenergy2} which implies that under the no-ac hypothesis and for any $ x \in \mathcal{T} $ and any $ \varepsilon > 0 $:
\begin{equation}
	\lim_{\eta\downarrow 0}\  \mathbb{P}\left(\Im \sigma_x(E+i\eta) > \varepsilon \right) \ = \ 0 \, . 
\end{equation}
Since $ \inf_{\eta \in (0,1]} \,  \mathbb{P}\left( L_x(\eta)\right) > 0 $
is strictly positive by~\eqref{eq:Llbound1}, we 
conclude that there is some $ \eta_1(n) \in (0,\eta_0] $ such that for all $ \eta \in(0, \eta_1(n) )$:
\begin{equation}
 1 -  \mathbb{P}\left( 	T_{x} \right) \leq  \tfrac{1}{16} \, \mathbb{P}\left( L_x\right) \, \, . 
\end{equation}
This concludes the proof of~\eqref{eq:lowerboundN}. 
The exponential estimate~\eqref{eq:expesttau} finally shows that the right side in~\eqref{eq:lowerboundN} is arbitrarily large  if $ n $ is chosen large.
\end{proof}

It remains to prove the following lemma.
\begin{lemma}\label{lem:RandL}
There is $ \eta_0 > 0 $ and $ n_0 >0 $ such that for all $ \eta \in (0,\eta_0) $ and all $n = \dist(x,0)\geq n_0$:
\begin{equation}
  \mathbb{P}\left( R_x\cap L_x \right) \ \geq \ \tfrac{1}{2} \,   \mathbb{P}\left(  L_x \right)\, . 
\end{equation}
\end{lemma}
\begin{proof}
The idea is to control the conditional probability conditioned on the sigma-algebra $\mathscr{A} $ generated by the random variables $ V(y) $ with $ x_{n_\kappa} \preceq y $. The assertion follows from the fact that there is $ \eta_0 > 0 $ and $ n_0 >0 $ such that for all $ \eta \in (0,\eta_0) $ and all $n = \dist(x,0)\geq n_0$:
\begin{equation}
  \mathbb{P}\left( R_x \, \big| \, \mathscr{A} \right)\, \, \indfct_{L^{(\rm bc)}_x }  \ \geq \ \tfrac{1}{2}  \, \indfct_{L^{(\rm bc)}_x }\, . 
\end{equation}
As a preparation, we expose the influence the conditioning on $ \mathscr{A} $ has on the Green function using its factorization property:
\begin{align}
	 G(\eta) \ & := \    G^{\mathcal{T}_{x}}(x_{n_\kappa-1},x_{n_\kappa-1};E+i\eta) \notag \\
	\widehat G(\eta)\ & := \ G^{\mathcal{T}_{x_{n_\kappa -1}}}(0,x_{n_\kappa-2};E+i\eta)   \ = \  G^{\mathcal{T}_{x}}(0,x_{n_\kappa-1};E+i\eta) \big/ G(\eta)\, . 
\end{align}
By the choice of the parameter $ b $, one has  $ \mathbb{P}\big( R_x^{\rm(bc)} |   \mathscr{A} \big) \geq 7/8 $ and hence
\begin{align}
 \mathbb{P}\left( R_x \, \big| \, \mathscr{A} \right) \ & \geq  \ \mathbb{P}\left( | \widehat G(\eta) \, G(\eta) | \in \big[e^{-\ell n_{\kappa}} , 1 \big] \, \big| \,  \mathscr{A} \right) - \tfrac{1}{8} \notag  \\
	 & \geq \  \mathbb{P}\left( | \widehat G(\eta)  | \in \big[B \, e^{-\ell n_{\kappa}} , b^{-1}  \big] \right) +  \mathbb{P}\left( | G(\eta) | \in \big[B^{-1} , b \big] \, \big| \,  \mathscr{A} \right) - \tfrac{1}{8} \, ,  \notag \\
	 &  \geq \  \mathbb{P}\left( | \widehat G(\eta)  | \in \big[B \, e^{-\ell n_{\kappa}} , b^{-1}  \big] \right) +  \mathbb{P}\left( | G(\eta) | \geq B^{-1}  \, \big| \,  \mathscr{A} \right)  - \tfrac{1}{4} \, ,  
\end{align}
where the last inequalities hold for any $ B \in [1,\infty) $. By Theorem~\ref{thm:regGbehav} the first term converges to one as $ n_{\kappa} \to \infty $. 
The event in the second term takes the form $$\Big| \lambda V(x_{n_\kappa-1}) - E - i \eta - \!\!\!\! \sum_{y \in \mathcal{N}_{x_{n_\kappa-1}}} G^{\widehat{\mathcal{T}}_x}(y,y;E+i\eta) \Big| \leq B \, . $$
In the event $ L_x^{\rm( bc)} $, there is $ B> 0 $ (which is independent of $ n $ and $ \eta $) such that for all $ \eta \in (0,1]$:
\begin{equation}
 \mathbb{P}\left(|G(\eta) | < B^{-1} \, \big| \, \mathscr{A} \right) \, 	 \indfct_{L^{(\rm bc)}_x } \ \leq \  \tfrac{1}{8}\,  \indfct_{L^{(\rm bc)}_x }  \, .
\end{equation}
This completes the proof. 
\end{proof}

\subsection{Establishing the events' occurrence}\label{sec:2mom}

Our aim in this subsection is to provide a uniform upper bound on 
$ \mathbb{E}\left[N^2 \right] / \mathbb{E}\left[N\right]^2$, for  $ N  = \sum_{x \in S_{n}^\kappa} \indfct_{D_x\cap I_x } $, which counts the number of resonance events on the thinned sphere.
\begin{theorem}\label{thm:N2test}
Under the no-ac hypothesis, there exists some constant $ C < \infty $ such that for all $ n  $ sufficiently large there is $ \eta_0 \equiv \eta_0(n) $ such that for all $ \eta \in (0, \eta_0 ) $, $ \alpha \in [1/2,1) $:
\begin{equation}\label{eq:N2test}
	\frac{\mathbb{E}[N(N-1)]}{\mathbb{E}[ N ]^2}\ \leq \ C\ < \infty  \, . 
\end{equation}
\end{theorem}
\begin{proof} 
Throughout the proof we will suppress the dependence on $ n$, $\eta $ and $\alpha$ at our convenience. Appearing constants $c, \,  C $ will be independent of $ n$, $\eta $ and $\alpha$.
We write
 \begin{equation}\label{eq:N(N-1)b}
	\mathbb{E}\left[N(N-1)\right] \   =  \!\sum_{\substack{x,y \in \mathcal{S}_n^\kappa \\ x \neq y }}
		\mathbb{P}\left(D_x\cap D_y \cap I_x \cap I_y \right)
		 = \ |\mathcal{S}_n^\kappa |\! \sum_{y \in \mathcal{S}_n^\kappa \backslash \{ x\}  } \mathbb{P}\left(D_x\cap D_y \cap I_x \cap I_y \right)  \, .
\end{equation}
The last equality holds for arbitrary $ x\in \mathcal{S}_n^\kappa $ which we will fix in the following.  By symmetry, the joint probability $ \mathbb{P}\left( D_x\cap D_y \cap I_x \cap I_y\right) $  depends only on the distance of the last common ancestor $ x \wedge y $ to the root. It is therefore useful to 
introduce the ratio
\begin{equation}
  \frac{\mathbb{P}\left(D_x\cap D_y \cap I_x \cap I_y\right)}{\mathbb{P}\left(D_x\cap I_x \right)\, \mathbb{P}\left(D_y\cap I_y \right)}\ := \  r(j) \; \delta_{\dist(x \wedge y,0), j} \ \, .
 \end{equation}
The sum in~\eqref{eq:N(N-1)b} may then be 
organized in terms of  the last common ancestor $ x \wedge y $ on the path $ \mathcal{P}_{0,x} = \{ x_0, \dots, x_{n} \} $ connecting the root with $ x $. 
In fact, since  $ \mathcal{S}_n^\kappa $ is thinned,  $ x\wedge y $ belongs to the shortened path
$\mathcal{P}_{0,x}^\kappa :=
\left\{ u  \in \mathcal{P}_{0,x} \, \big| \, \dist(u,0) < N_\kappa  \right\} $. 
Moreover, for a given $ x\wedge y \in \mathcal{P}_{0,x}^\kappa $, the number of vertices  $ y \in \mathcal{S}_n^\kappa  $, which for fixed $ x $ have the same common ancestor, is $|S_n^\kappa| \,  K^{ - \dist( x\wedge y,0)} $ such that 
\begin{align}\label{eq:N2sum}
	\frac{\mathbb{E}\left[N(N-1)\right]}{\mathbb{E}\left[N\right]^2} \ & = \    \sum_{j=0}^{N_\kappa -1} \frac{r(j)}{K^j} \, . 
\end{align}
In order to estimate the sum in the right side of~\eqref{eq:N2sum}, we always drop the condition $R_x $ in the definition of 
$ D_x $:
\begin{equation}\label{eq:dropR}
	r(j) \leq \frac{\mathbb{P}\left(L_x \cap L_y \cap E_x \cap E_y \cap  I_x \cap I_y \right)}{\mathbb{P}\left(D_x\cap I_x\right)\, \mathbb{P}\left(D_y\cap I_y\right)} \; \delta_{\dist(x \wedge y,0), j} \, . 
\end{equation}
For an estimate on the numerator in the right side, we first focus on the extreme fluctuation events and aim to integrate out the random variable associated with $ x $ and $ y $ using Theorem~\ref{thm:twoL1} in the Appendix. In general, what stands in the way of this procedure is the dependence of 
$L_x $ on $ V(y) $ and $ L_y$ on $ V(x) $, respectively. We therefore relax the conditions in the large deviation events and pick suitable
\begin{equation}\label{eq:supersetL}
	 \widehat{L}_{x,j} \supset L_x \, , \quad \mbox{(and hence}\quad \widehat{L}_{y,j} \supset L_y \;) 
\end{equation}
such that $ \widehat{L}_{x,j} $ and $ \widehat{L}_{y,j} $ are independent of both $ V(x) $ and $ V(y) $. Postponing the details of these choices which will
depend on $ j $, we bound the numerator on the right side in~\eqref{eq:dropR} using Theorem~\ref{thm:twoL1} in the Appendix:
\begin{align}\label{eq:relaxhatL}
	& \mathbb{P}\left( L_x \cap L_y \cap E_x \cap E_y \cap  I_x \cap I_y\right) \  \leq  \ \mathbb{E}\left[ \indfct_{ \widehat{L}_{x,j} \cap  \widehat{L}_{y,j} } \, \mathbb{P} \left( E_x \cap E_y  \, | \,\mathscr{A}_{x,y} \right) \right] \notag \\
	& \leq C\, \left(  \tau^{-2} \, \mathbb{P}\left( \widehat{L}_{x,j} \cap  \widehat{L}_{y,j} \right) \, +   \tau^{-1} \, \mathbb{E}\left[  \indfct_{ \widehat{L}_{x,j} \cap  \widehat{L}_{y,j} } \, \min\big\{\big| \widehat G_{x,y}  \big|, 1 \big\} \right]   \right) \, ,
\end{align}
where we have abbreviated by $ \mathscr{A}_{x,y} $ the sigma algebra generated by the variables $ V(\xi) $, $ \xi \not\in \{x,y\} $ and 
\begin{equation}
	\widehat G_{x,y} := G^{\mathcal{T}_{x,y}}(x_{n-1},y_{n-1};E+i\eta) \, . 
\end{equation}
This quantity measures the strength of the interaction of the events $E_x$ and $E_y$.

Under the assumptions of Theorem~\ref{thm:1Mom}, the denominator in the right side of~\eqref{eq:dropR} is bounded from below by $c\,   \tau^{-2} \,  \mathbb{P}\left( L_x\right) \mathbb{P}\left( L_y\right)  $ provided $ n $ is sufficiently large and $ \eta $ is sufficiently small. 
The terms on the right side in~\eqref{eq:relaxhatL} hence give rise to two terms, $ r(j) \leq r_1(j) + r_2(j) $,  which for fixed $j = \dist(x \wedge y,0)  $ are defined as: 
\begin{align}
	r_1(j) \;  & := \  C\,   \frac{\mathbb{P}\big(\widehat{L}_{x,j} \cap  \widehat{L}_{y,j}\big)  }{\mathbb{P}\left(L_x\right) \, \mathbb{P}\left(L_y\right) }  \label{eq:defr1}  \\
	r_2(j) \; & :=  \ \frac{C\; \tau}{ \mathbb{P}\left(L_x\right) \, \mathbb{P}\left(L_y\right)} \, \mathbb{E}\left[  \indfct_{\widehat{L}_{x,j} \cap  \widehat{L}_{y,j}} \, \min\big\{| \widehat G_{x,y} |, 1 \big\} \right]  \label{eq:defr2}
\end{align}
For the precise definition of the events $ \widehat{L}_{x,j} $ and $ \widehat{L}_{y,j} $ we distinguish three cases: 
\begin{description}
\item[Case $\mathbf{0\leq j < n_\kappa}$:] 
The events $ L_x $ and $L_y $ are already independent of  the potential at~$ x $ and~$ y $. Therefore we choose
\begin{equation}
	\widehat{L}_{x,j} = L_x \, . 
\end{equation}
As a consequence, the corresponding sum involving $ r_1(j) $  is seen to be  uniformly bounded in $ n $ and $ \eta $:
\begin{equation}
 \sum_{j=0}^{n_\kappa -1} \frac{r_1(j)}{K^j}  \ \leq \ C \,  \sum_{j=0}^\infty \frac{1}{K^j} \, . 
\end{equation}

For an estimate on $ r_2(j) $, we drop the indicator function in the right side of~\eqref{eq:defr2} and use the fact that $ \min\{ |x|, 1\} \leq |x|^\sigma $ for any $\sigma \in [0,1)$; in particular, for $ \sigma = s $:
\begin{align}\label{eq:boundonr2}
	 r_2(j) \  & \leq \  \frac{C\; \tau}{ \mathbb{P}\left(L_x\right) \, \mathbb{P}\left(L_y\right)}\; \mathbb{E}\big[|\widehat G_{x,y} |^s\big]\ \leq \  \frac{C\; \tau}{ \mathbb{P}\left(L_x\right) \, \mathbb{P}\left(L_y\right)}\; e^{2(n-j) \varphi(s)} \, . 
\end{align} 
Here the second inequality derives from the finite-volume estimates~\eqref{eq:finitevolume}. 
Since  $ \varphi(s) < - \tfrac{1}{2}\log K $ by assumption on $ s $, the  geometric sum in the following chain of inequalities is dominated by its last term:
\begin{align}
\sum_{j=0}^{n_\kappa -1} \frac{r_2(j)}{K^j} \ & \leq \   \frac{C \, \tau }{\mathbb{P}\left(L_x \right) \, \mathbb{P}\left( L_y\right)} \sum_{j=0}^{n_\kappa -1} \frac{e^{2(n-j) \varphi(s)}}{K^j} \notag \\
& \leq \frac{C \,  \tau }{\mathbb{P}\left(L_x \right) \, \mathbb{P}\left( L_y\right)}  \frac{e^{2N_\kappa \varphi(s)}}{K^{n_\kappa}} \, . 
\end{align}
Using the large deviation result,  Theorem~\ref{cor:ldev1}, and the fact that $ - \varphi(s) = I(\gamma) + \gamma \, s $, we estimate
\begin{align}\label{eq:ratiobound2}
\frac{\tau }{\mathbb{P}\left(L_x \right) \, \mathbb{P}\left( L_y\right)}  \,  e^{2N_\kappa \varphi(s)} \ & \leq \ e^{4 N_\kappa \epsilon}\, \tau \, e^{-2N_\kappa \gamma s }\ \leq \ e^{N_\kappa \left( \left(\tfrac{7}{4}- 2s)\right) \Delta +4 \epsilon\right)} \notag \\
& \ \leq \ e^{N_\kappa \,  (\tfrac{15}{8} - 2s )\Delta}\ \leq \ C \, ,
\end{align}
since $ 2s > 15/8  $.

\item[ Case $ \mathbf{n_\kappa \leq j \leq \tfrac{3}{2} n_\kappa} $:]
We choose
\begin{equation}
	\widehat{L}_{x,j} =  L^{(N_\kappa - \tfrac{1}{2}n_\kappa-1,+)}_x \, , 
\end{equation}
which is independent of $ \widehat{L}_{y,j} = L^{(N_\kappa - \tfrac{1}{2}n_\kappa-1,+)}_y $. 
An estimate on $ r_1(j) $ hence requires to bound the ratio:
\begin{align}\label{eq:Nnb}
	\frac{\mathbb{P}\big(\widehat{L}_x  \big)}{\mathbb{P}\left(L_x\right)}  \ \leq  \ C \, 
	\frac{e^{-(n-\tfrac{3}{2}n_\kappa-2)(I(\gamma) - 2 \epsilon) }}{e^{-N_\kappa (I(\gamma) + 2\epsilon)}} \leq C \, e^{4 N_\kappa\epsilon}\, e^{ \tfrac{n_\kappa}{2} I(\gamma) }
	\leq  \ C\, K^{n_\kappa/2} \, .
\end{align}
Here the first inequality follows from the large deviation result, Theorem~\ref{cor:ldev1}, and holds for  $ n $  large enough and $ \eta $ sufficiently small. In this situation, the third inequality also applies since $ I(\gamma) \leq \log K - \tfrac{15}{8} \Delta $ by \eqref{eq:assumphi} and \eqref{eq:defgamma}, and $ 4 N_\kappa \epsilon \leq \Delta \kappa N_\kappa / 4 \leq \Delta \,  n_\kappa /2 $. 
As a consequence, the sum corresponding to $ r_1(j) $ is bounded uniformly in $ n $: 
\begin{equation}
 \sum_{j=n_\kappa}^{\tfrac{3}{2}n_\kappa} \frac{r_1(j)}{K^j}  \leq C \, K^{n_\kappa} \, \sum_{j=n_\kappa}^\infty \frac{1}{K^j} \leq C \, \sum_{j=0}^\infty \frac{1}{K^j} \, . 
 \end{equation} 
 
For an estimate on the sum  corresponding to $ r_2(j) $ we use \eqref{eq:boundonr2} again which yields
\begin{equation}
	\sum_{j=n_\kappa}^{\tfrac{3}{2}n_\kappa } \frac{r_2(j)}{K^j} \leq \frac{C \,  \tau }{\mathbb{P}\left(L_x \right) \, \mathbb{P}\left(L_y\right)}  \frac{e^{(2N_\kappa - n_\kappa) \varphi(s)}}{K^{\tfrac{3}{2} n_\kappa}}  \leq  \frac{C \, \tau }{\mathbb{P}\left(L_x \right) \, \mathbb{P}\left( L_y\right)}  \frac{e^{2N_\kappa \varphi(s)}}{K^{n_\kappa/2}}  \leq \ C
\end{equation}
by \eqref{eq:ratiobound2}. 

\item[Case $\mathbf{ \tfrac{3}{2}n_\kappa < j < N_\kappa} $:]
In this main case, we pick
\begin{align}
	 \widehat{L}_{x,j} \ & = \ L^{(j -n_\kappa -1,-)}_x \cap L^{(N_\kappa +n_\kappa-j-1,+)}_x\, ,
\end{align}
Note that $ L^{(j -n_\kappa -1,-)}_x =  L^{(j -n_\kappa -1,-)}_y $ and $ L^{(N_\kappa +n_\kappa-j-1,+)}_x $ and $ L^{(N_\kappa +n_\kappa-j-1,+)}_y  $ are independent. 
We may hence estimate the numerator in the definition of $ r_1(j) $ using the large deviation result, Theorem~\ref{cor:ldev1}  to conclude that for  all $ n $ sufficiently large and $ \eta $ sufficiently small:
\begin{align}
  \mathbb{P}\big(  \widehat{L}_{x,j}  \cap \widehat{L}_{y,j}  \big) &  \leq  \ \mathbb{P}\left(L^{(j -n_\kappa -1,-)}_x\right)  \mathbb{P}\left(  L^{(N_\kappa +n_\kappa-j-1,+)}_x\right)  \mathbb{P}\left(L^{(N_\kappa +n_\kappa-j-1,+)}_y\right) \notag \\ & \leq \ C \,  \, e^{- (I(\gamma)-2\epsilon) \left(2n - j - n_\kappa\right) } \notag \\
& \leq \ C \; \mathbb{P}\left(L_x\right)  \mathbb{P}\left(L_y\right) \, e^{8 N_\kappa \epsilon} \, e^{- I(\gamma) \left(n_\kappa - j \right) } \, . 
\end{align}
Since $ I(\gamma) < \log K $, the corresponding sum is hence uniformly bounded in $ n $: 
\begin{align}
 \sum_{j=\tfrac{3}{2}n_\kappa+1}^{N_\kappa-1} \frac{r_1(j)}{K^j} \ & \leq \ C \, e^{8 N_\kappa \epsilon} \,
	 \sum_{j=\tfrac{3}{2}n_\kappa}^{N_\kappa}  \frac{e^{- I(\gamma) \left(n_\kappa - j \right) }}{K^j} \notag \\ 
	 &  \leq \ C \;   e^{8 N_\kappa \epsilon} \, \frac{e^{\tfrac{n_\kappa}{2} I(\gamma) }}{K^{\tfrac{3}{2} n_\kappa}} \, \leq \ C \, \frac{e^{8 N_\kappa \epsilon}}{K^{n_\kappa}} \ \leq \ C \,  ,
\end{align} 
 cf.~\eqref{eq:Nnb}.

For an estimate on $r_2(j) $ we drop conditions in the indicator function and use $ \min\{ |x|, 1\} \leq |x|^s $ again:
\begin{equation}\label{eq:r23case}
	 r_2(j) \  \leq \
	C \, \tau \, \frac{\mathbb{E}\big[ \indfct_{L^{(n_\kappa,j-1)}_x} \,  |\widehat G_{x,y} |^s  \big] }{\mathbb{P}\left(L_x\right) \, \mathbb{P}\left(L_y\right)}
\end{equation}
The Green function in the numerator is a product of three terms, $  \widehat G_{x,y} =    G_j \, \widehat G_x \, \widehat G_y $ with
\begin{align}
 & G_j  := G^{\mathcal{T}_{x,y}}(x_{j},y_{j})  \\
 & \widehat G_x := G^{\mathcal{T}_{x_{j},x}}(x_{j+1},x_{n-1}) \, \quad \widehat G_y := G^{\mathcal{T}_{y_{j},y}}(y_{j+1},y_{n-1}) \notag   \end{align}
 of which only the first one depends on $  V(x_j) $. Since $L^{(n_\kappa,j-1)}_x $ is independent of $ V(x_jj) $ we may hence condition on the potential elsewhere and use the uniform bound $ \mathbb{E}\left[ | G_j  |^s \, | \, \mathscr{A}_{x_j} \right] \leq C $ to estimate the numerator in~\eqref{eq:r23case}:
 \begin{align}
 \mathbb{E}\big[ \indfct_{L^{(n_\kappa,j-1)}_x } \,  |\widehat G_{x,y} |^s  \big] \leq  & \ C \; \mathbb{E}\big[ \indfct_{L^{(n_\kappa,j-1)}_x} \,  | \widehat G_x \, \widehat G_y |^s  \big]  \notag \\
  = & \  C\; \mathbb{P}\big(L^{(n_\kappa,j-1)}_x\big) \, \mathbb{E}\left[|\widehat G_x  |^s\right] \, \mathbb{E}\left[| \widehat G_y  |^s\right] \notag \\
  \leq  &\ C \; e^{- (j - n_\kappa) (I(\gamma) - 2\epsilon) } \, e^{2 (n-j) \varphi(s) } \, . 
 \end{align}
 Summing over $ j $ with a weight $ K^{-j} $ we again obtain a geometric sum which is in this case bounded by the number of terms times the maximum of its first and last term. Therefore we conclude that
 \begin{align}
 	& \sum_{j=\tfrac{3}{2}n_\kappa +1}^{N_\kappa-1} \frac{r_2(j)}{K^j} \ \leq \sum_{j=n_\kappa}^{N_\kappa-1} \frac{r_2(j)}{K^j} \ \leq \ N_\kappa \;  \frac{C \, \tau }{\mathbb{P}\left(L_x \right) \, \mathbb{P}\left(L_y\right)} \\
	& \qquad \times \; \max\left\{ \frac{e^{-(N_\kappa - n_\kappa)(I(\gamma) - 2\epsilon)} e^{2 n_\kappa \varphi(s)} }{K^{N_\kappa}} \, , \,  \frac{ e^{2 N_\kappa \varphi(s)} }{K^{n_\kappa}} \right\} \, .  \notag
 \end{align}
  In the first case, we use $ \varphi(s) < - I(\gamma) $ and Corollary~\ref{cor:ldev1} to conclude that the term is  uniformly bounded in $ n $:
  \begin{align}
  	 N_\kappa \;  \frac{C \, \tau }{\mathbb{P}\left(L_x \right) \, \mathbb{P}\left(L_y\right)} \, \frac{e^{- N_\kappa (I(\gamma) -2 \epsilon)} }{K^{N_\kappa}} \ & \leq \  N_\kappa \;  \frac{C \,e^{N_\kappa (I(\gamma) + \gamma +\tfrac{3}{4} \Delta + 6 \epsilon)}}{K^{N_\kappa}}  \notag \\
	& \leq \ \, C \, N_\kappa \, e^{-N_\kappa (\frac{1}{8} \Delta - 6 \epsilon)} \ \leq \ C \, ,
  \end{align}
since $ \epsilon <\Delta / 48$. 

 In the second case, we use  \eqref{eq:ratiobound2} to conclude that  the term is  uniformly bounded in~$ n $:
 \begin{equation}
 	N_\kappa \,   \;  \frac{C \, \tau }{\mathbb{P}\left(L_x \right) \, \mathbb{P}\left(L_y\right)} \, \frac{e^{2 N_\kappa \varphi(s)} }{K^{n_\kappa}} \   \leq \ C \, N_\kappa \, e^{N_\kappa (\frac{15}{8} - 2s)} \ \ \leq \ C \, ,
 \end{equation}
 since $ 2s > \frac{15}{8} $.
\end{description}
This concludes the proof of~\eqref{eq:N2test}. 
\end{proof}

%
\section{Semi-continuity bounds for the Lyapunov exponent }\label{sec:Lyap}

As we saw in Section~\ref{sec:applications}, the applications of the conditions which are derived here for absolutely continuous spectrum still require some additional information on the function $\varphi_\lambda(1;E)$, or at least on the Lyapunov exponent $L_\lambda(E)$.  While we do not have useful independent bounds on $\varphi_\lambda(1;E)$, in this section we present some partial continuity results for $L_\lambda(E)$ which enable the derivation of the main conclusions which were drawn in  Corollaries~\ref{thm:charcac1} and~\ref{thm:charcac2} on the spectral phase diagram.

 Let us start with some general observations:
\begin{enumerate}
\item
The Lyapunov exponent 
is the negative real part of the  Herglotz function (cf.~\cite{D,PF}) given by $W_\lambda(\zeta) :=  \mathbb{E}\left[\log  \Gamma_\lambda(0;\zeta) \right] $. As such, its boundary values
$ \lim_{\eta\downarrow 0}  L_\lambda(E+i\eta) $ exist for Lebesgue-almost all $ E \in \mathbb{R} $ and. The latter coincides with  $L_\lambda(E) $ defined in~\eqref{eq:LyapE}, as is seen using a variant of Vitali's convergence theorem whose use is based on the fact that the fractional moments of $ \Gamma_\lambda(0;E+i\eta)$ with positive and negative power are uniformly bounded in $ \eta $.
\item
 In the absence of disorder, the Lyapunov exponent is easy to compute, $ L_0(\zeta) = -\log |\Gamma_0(\zeta)| $, where 
 $  \Gamma_0(\zeta) $ is the unique solution of $ K \Gamma^2 + \zeta \Gamma +1 = 0 $ in $ \mathbb{C}^+ $, and one finds:  
  \begin{equation}\label{eq:freeLE}
		  L_0(E) \  \left\{ \begin{array}{ll}
		  	= \log \sqrt{K} \quad &  |E | \leq 2 \sqrt{K} \, , \\  
		 \in \left(\log \sqrt{K} , \log K \right) \quad &  2 \sqrt{K} < |E| < K+1\, ,  \\
		  \geq \log K \quad & |E| \geq K+1 \, . 
		 \end{array} \right.
	\end{equation}
\item
In general, $   L_\lambda(\zeta) $ is related to the free energy function $ \varphi_\lambda(s;\zeta) $  through the relation~\eqref{eq:derphi} and the inequality~\eqref{eq:exclusion} from which
one concludes the bound 
$  L_\lambda(\zeta) \geq \log \sqrt{K} $
 which is saturated if and only if $ \lambda = 0 $ and  $|E| \leq 2 \sqrt{K} $. 
\end{enumerate}

\subsection{Continuity of energy averages}\label{sec:conLE1}
Thanks to the (weak) continuity of the harmonic measure associated with $ L_\lambda $, 
energy averages turn out to be continuous in the disorder parameter $ \lambda \geq 0 $. 
  \begin{theorem}\label{lem:conti}
  	For any bounded interval $ I \subset \R $ the function $ [0,\infty) \ni \lambda \mapsto \int_I L_\lambda(E) \, dE $ is continuous, and, in particular:
	\begin{equation}\label{eq:conti}
	\lim_{\lambda \downarrow 0} \int_I L_\lambda(E) \, dE  \ = \ \int_I L_0(E) \, dE  \, . 
	\end{equation} 
  \end{theorem}
  \begin{proof}
  	Since the harmonic measure $\sigma_\lambda(I)  := \int_I L_\lambda(E) \, dE $ associated with $ L_\lambda(\zeta ) = \pi^{-1} \int \Im (E - \zeta)^{-1} \sigma_\lambda(dE) $ is absolutely continuous,
the asserted continuity thus follows from the vague continuity of  $\sigma_\lambda $, which in turn follows from the (weak) resolvent convergence $ G_\lambda(0,0;\zeta,\omega) \to G_{\lambda_0}(0,0;\zeta,\omega) $ as $ \lambda \to \lambda_0 $ for all $ \zeta \in \C^+ $ and all $ \omega $.
  \end{proof}
  
In particular, Theorem~\ref{lem:conti} ensures that the mean value of the Lyapunov exponent over any bounded, non-empty interval~$ I $, 
\begin{equation}
	M_\lambda(I) \ := \ \frac{1}{|I|} \, \int_I L_\lambda(E) \, dE \, , 
\end{equation}
is continuous in $ \lambda \geq 0$.  This immediately implies Corollary~\ref{thm:charcac1}, namely that the condition $ L_\lambda(E) < \log K $ holds on a positive fraction of every interval $ I \subset (-(K+1),K+1) $.
 \begin{proof}[Poof of Corollary~\ref{thm:charcac1}]
Since $ L_\lambda(E) \geq \log \sqrt{K} $, we may employ the Chebychev inequality to control the Lebesgue measure of that subset of $ I $ on which~\eqref{lyapcond} is violated:
\begin{align}\label{eq:condLno}
	\left| \left\{ E \in I \, | \, L_\lambda(E) \ \geq \ \log K \right\} \right| \ & \leq \ \int_I \frac{L_\lambda(E) - \log\sqrt{K}}{\log\sqrt{K}} \, dE  \ = \ | I | \; \frac{M_\lambda(I) - \log\sqrt{K} }{\log\sqrt{K}} \, . 
\end{align}
The assertion thus follows from the continuity~\eqref{eq:conti} and the fact that $ \log \sqrt{K} \leq M_0(I) < \log K $ for all closed intervals $ I \subset  (- K- 1, K+1) $ by a computation, cf.~\eqref{lem:conti}.
 \end{proof}
Note that $ M_0(I) = \log\sqrt{K} $ for all $ I \subset (-2\sqrt{K}, 2\sqrt{K}) $. Hence, in this case the measure in~\eqref{eq:condLno} tends to $ 0$ as $ \lambda \downarrow 0 $.

\subsection{The case of bounded random potentials}\label{sec:conLE2}

Let us now turn to the proof of Corollary~\ref{thm:charcac2}. Accordingly, for the remainder of this section, we will assume that $ \supp \varrho = [-1,1] $ such that 
almost surely $ \sigma(H_\lambda) = [ - |E_\lambda| , |E_\lambda| ] $ with $ E_\lambda = -2\sqrt{K} - \lambda $.

The main ideas behind the conditions in Corollary~\ref{thm:charcac2} are:
\begin{enumerate}[a)]
\item At the (lower) spectral edge the Lyapunov exponent is bounded according to:
\begin{equation}\label{eq:boundedgeL}
	L_\lambda(E_\lambda) \ \leq \ L_0(E_\lambda-\lambda) \, .
\end{equation}
(An analogous bound applies to the upper edge).
This inequality derives from the operator monotonicity of the function $ (0,\infty) \ni x \mapsto x^{-1} $ and the estimate $0 \leq H_\lambda - E_\lambda \leq T + 2\sqrt{K}+ 2\lambda $, which implies $ \Gamma_\lambda(0;E_\lambda) \geq  \Gamma_0(E_\lambda-\lambda) $.
\item Using the explicit formula for the Lyapunov exponent in case $ \lambda = 0 $ (cf.\ \eqref{eq:freeLE}),  we conclude that  
the condition $ L_0(E_\lambda-\lambda) < \log K $ holds if and only if $E_\lambda-\lambda> - (K+1) $ or equivalently if \eqref{eq:weak_disorder} holds. 
\end{enumerate}

The following theorem extends the bound~\eqref{eq:boundedgeL} to energies near $ E_\lambda $ in the spectrum.  Analogous arguments yield an upper bound near $ - E_\lambda $.
\begin{theorem}\label{thm:upper bound}
For a random potential satisfying Assumptions~\ref{assA}--\ref{assD} with $ {\rm supp}\, \varrho = [-1,1] $,   
for all $ \lambda > 0 $: 
 \begin{equation}
	 \limsup_{E \downarrow E_\lambda} \,  L_\lambda(E) \ \leq \  L_0(E_\lambda-\lambda)   \, . 
\end{equation}
\end{theorem}

Following the arguments above, this theorem in particular implies Corollary~\ref{thm:charcac2}. 

\begin{proof}[Proof of Corollary~\ref{thm:charcac2}]
	Without loss of generality, we restrict the discussion to the region near the lower edge $E_\lambda $ of $ \sigma(H_\lambda) $. 
	For fixed $ \lambda < (\sqrt{K}-1)^2/2$ we may pick $ \varepsilon(\lambda) := \log K - L_0(E_\lambda-\lambda) >0 $ which is strictly positive if and only if  \eqref{eq:weak_disorder} holds. We hence conclude from Theorem~\ref{thm:upper bound} that there is $  \delta(\lambda) > 0 $ such that $ L_\lambda(E) < \log K $ for any $ E \leq E_\lambda +  \delta(\lambda) $. 
\end{proof}
 
\subsubsection{Proof of Theorem~\ref{thm:upper bound}}

In the proof of Theorem~\ref{thm:upper bound}, we consider the finite-volume restriction of the operator to the Hilbert-space over $ B_R := \left\{ x \in \T \, | \, \dist(0,x) < R \right\} $, i.e., 
\begin{equation}
	H_\lambda^{(R)} := \indfct_{B_R} H_\lambda \indfct_{B_R}  \quad\mbox{on $ \ell^2(B_R) $.}
\end{equation}
The relation between the Green function and its finite-volume counterpart is controlled by standard perturbation theory, i.e., for almost every $ E \in \mathbb{R} $:
\begin{equation}\label{eq:pertGR}
	\left| \Gamma_\lambda(0;E+i0) - \Gamma_\lambda^{(R)}(0;E) \right| \ \leq \! \sum_{x \in \mathcal{S}_R} \big| G_\lambda^{(R)}(0,x_- ;E) \big| \, \left| G_\lambda(0,x ;E) \right| \ =: \ S_\lambda^{(R)}(E) \, . 
\end{equation}
The proof idea for Theorem~\ref{thm:upper bound} is to choose $ R $  such that:
\begin{enumerate}[a)]
	\item The following event has a good probability, 
	\begin{equation}
		Z_1 := \ \left\{  E_\lambda +\Delta \leq \inf \sigma(H_\lambda^{(R)} ) \right\}  \, .
	\end{equation} 
	In this event and for any $ E \in [E_\lambda, E_\lambda +\Delta) $ one can use the operator monotonicity of  $(0,\infty) \ni x \mapsto x^{-1}$ together with the bound $ 0 \leq H_\lambda^{(R)}  - E \leq H_0^{(R)}  + \lambda - E $ which implies
	\begin{align}\label{eq:opmon}
		\Gamma_\lambda^{(R)}(0;E) \ & \geq  \ \Gamma_0^{(R)}(0;E-\lambda) \ \geq  \    \Gamma_0^{(R)}(0;E_\lambda-\lambda) \notag \\
		 & \geq \ \Gamma_0(0;E_\lambda-\lambda) - S_0^{(R)}(E_\lambda-\lambda) \notag \\
		 & \geq \ \Gamma_0(0;E_\lambda-\lambda) \left( 1 - K^R e^{-2RL_0(E_\lambda-\lambda) } \right) \, . 
	\end{align}
	Here, the second inequality holds for all $ E_\lambda \leq E <  \inf \sigma(H_\lambda^{(R)} ) $, the third is a special case of \eqref{eq:pertGR}, and the last inequality follows from the fact that 
	\begin{equation}\label{eq:GRvsinfty}
	0 \leq \Gamma_0^{(R)}(x;E) \ \leq \ \Gamma_0(x;E)  \, ,
\end{equation}
	which, using the factorization property of the Green function, implies $ S_0^{(R)}(E)   \leq  K^R  e^{-2R L_0(E)} $ $ \Gamma_0(0;E) $ for any $ E \in \mathbb{R} $.
	\item The error terms on the right side of \eqref{eq:pertGR} and \eqref{eq:opmon} are small compared to $ \Gamma_0(0;E_\lambda-\lambda) \geq 0 $ in the sense that also the event
	\begin{equation}
		Z_2 := \ \left\{ S_\lambda^{(R)}(E) \leq \Gamma_0(0;E_\lambda-\lambda) \,  K^{-\delta R/2} \right\}
	\end{equation}
	occurs with a good probability. For reasons will become clear in the next subsection, we will choose
	\begin{equation}\label{eq:pickdelta}
		 \delta \ := \frac{\log( 1 + \frac{\lambda}{2\sqrt{K}}) }{64\, \|\varrho \|_\infty \, K^2  \log K } 
	\end{equation} 
\end{enumerate}

The probability of failure of the first event $ Z_1 $ is bounded with the help of the following lemma. Due to Lifshits tailing, this estimate is far from optimal and one expects the probability in~\eqref{eq:LT} to 
be exponentially small (see~\cite{BS} and references therein for a precise conjecture).
\begin{lemma}\label{lem:LT}
	There is some $ C >0 $ such that for all $ R > 0 $ and all $ \Delta >0 $ 
	\begin{equation}\label{eq:LT}
		\mathbb{P}\left(   \inf \sigma(H_\lambda^{(R)} ) <   E_\lambda +\Delta  \right)  \ \leq \  C \, K^R \,  \Delta^{3/2} 
	\end{equation}
\end{lemma}
\begin{proof}
By Chebychev's inequality the left side is bounded from above by
\begin{align}
	\mathbb{E}\left[ \tr \indfct_{(-\infty,E)}(H_\lambda^{(R)} ) \right] \ & \leq  \tr \indfct_{(-\infty,E+\lambda)}(H_0^{(R)} )  \leq  e^{t(E+\lambda) }  \tr e^{-t H_0^{(R)} } \notag \\
	& \leq e^{t(E+\lambda) } \, \tr \indfct_{B_R}  \ e^{-t H_0 } \indfct_{B_R}  \leq C \, K^R \, e^{t \Delta } \,  t^{-3/2} \, , 
\end{align}
where $E := E_\lambda + \Delta$ and the last inequality stems form the explicitly known form of the kernel of the (infinite-volume) semigroup. Taking $ t = \Delta^{-1} $ yields the result. 
\end{proof}

Bounds on the probability of failure of the second event $ Z_2 $ are more involved. 
Postponing the details of this probabilistic estimate, which will be the topic of the next subsection, the proof of Theorem~\ref{thm:upper bound} proceeds as follows:

\begin{proof}[Proof of Theorem~\ref{thm:upper bound}]
Abbreviating $ Z := Z_1\cap Z_2 $, we write
\begin{align}\label{eq:startPL}
	L_\lambda(E) \ & = \ - \E\left[ \indfct_{Z} \log | \Gamma_\lambda(0;E+i0) | \right] -  \E\left[ \indfct_{Z^c} \log | \Gamma_\lambda(0;E+i0) | \right]
\end{align}
In the event $ Z $ and assuming $ E \in [E_\lambda, E_\lambda +\Delta) $, one may use~\eqref{eq:pertGR} and \eqref{eq:opmon} to estimate
\begin{align}
	| \Gamma_\lambda(0;E+i0) | \ & \geq \ \Gamma_0^{(R)}(x;E)  - S_\lambda^{(R)}(E) \notag \\
	&  \geq \ \Gamma_0(0;E_\lambda-\lambda) \left( 1 - K^R e^{-2RL_0(E_\lambda-\lambda) }  -  K^{-\delta R/2} \right) \, . 
\end{align}
The right side is strictly positive for any $ \lambda > 0 $ provided $ R $ is large enough. 
In this case, the above bound and the monotonicity of the logarithm yields the following bound on the first term on the right in~\eqref{eq:startPL}:
\begin{equation}
- \E\left[ \indfct_{Z} \log | \Gamma_\lambda(0;E+i0) | \right] 
 \  \leq  \ 	 L_0(E_\lambda-\lambda)  - \log\left( 1 - K^R e^{-2RL_0(E_\lambda-\lambda)} -K^{-\delta R/2} \right) \, . 
\end{equation}
The second term in~\eqref{eq:startPL} is estimated using the Cauchy-Schwarz inequality
\begin{align}
	-   \E\left[ \indfct_{Z^c} \log | \Gamma_\lambda(0;E+i0) | \right]  \ & \leq \ \sqrt{\mathbb{P}\left(Z^c\right) } \sqrt{ \E\left[ \left|  \log | \Gamma_\lambda(0;E+i0) | \right|^2 \right]} 
\end{align}
Since $ |\log |x| | \leq 2 ( |x|^{1/2} + |x|^{-1/2} )$ the second factor is bounded with the help of fractional-moment estimates and~\eqref{eq:recur_gen} by a constant which only depends on $ \lambda $.
The probability of failure of the event $ Z $ is estimated using Lemma~\ref{lem:LT} and Lemma~\ref{lem:propest} which prove that under the condition~\eqref{eq:condE} below:
\begin{align}
\mathbb{P}\left( Z^c\right) \ & \leq \ \mathbb{P}\left( Z_2^c \, | \, Z_1 \right) + \mathbb{P}\left( Z_1^c \right) \notag \\
	& \leq \ C(\lambda) \, K^{- \frac{\delta^2}{4+4\delta} R}  + 2^{-R}  +  C \, K^R \,  \Delta^{3/2} \, . 
\end{align}
We pick $ \Delta := (E-E_\lambda) / c(\lambda) $ with $ c(\lambda ) $ from~\eqref{eq:condE} and $ R := \lceil \frac{\log \Delta^{-1} }{\log K} \rceil $. The proof is completed by noting that for any $ \lambda > 0 $: \emph{i.} $  \Delta \to 0 $ as  $ E \to E_\lambda $ and \emph{ii.} $ R \to \infty $ as $  \Delta \to 0 $. 
\end{proof}

\subsubsection{Auxiliary results}

The remaining task concerns the estimate on the error in \eqref{eq:pertGR}. We will prove
\begin{lemma}\label{lem:propest}
For every $ \lambda > 0 $ there exists a finite $ C(\lambda) $ such that if
\begin{equation}\label{eq:condE}
	E \ \leq \ E_\lambda + \Delta \,\left[ 1 - \exp\left(- \frac{\log(1 + \frac{\lambda}{2\sqrt{K}} ) }{64\, \|\varrho \|_\infty \, K^2  \log K }\right)\right] \quad \left[ =: E_\lambda + c(\lambda) \, \Delta \right] \, . 
\end{equation}
then  $\;
	\mathbb{P}\left( Z_2^c \, \big| \,  Z_1  \right) \  \leq \ C(\lambda) \, K^{- \frac{\delta^2}{4+4\delta} R}  + 2^{-R}  $.
\end{lemma}
For a proof of this auxiliary estimate,  
we need to control the first factor in the right side of~\eqref{eq:pertGR} in case $  E < \inf \sigma(H_\lambda^{(R)} ) $. This is done with the 
help of the following lemma, which might be of independent interest.
 \begin{lemma}\label{lem:aux}
\begin{enumerate}
\item
Assume $ a \leq b  < \inf \sigma(H_\lambda^{(R)}) $, then
\begin{equation}\label{eq:specgammaest}
	\Gamma_\lambda^{(R)}(x;a ) \leq \Gamma_\lambda^{(R)}(x;b)  \leq \left(1 +  \frac{b-a}{ \inf \sigma(H_\lambda^{(R)}) -b}\right) \Gamma_\lambda^{(R)}(x;a ) \, . 
\end{equation}
\item Assume $ a \leq - 2 \sqrt{K} $ and  $ x \in B_R $, then
\begin{equation}\label{eq:speed}
	\Gamma_\lambda^{(R)}(x;a -\lambda ) \leq \Gamma_0^{(R)}(x;a)   \left(1 +  \frac{\lambda}{\sqrt{K} - \frac{a}{2} }\right)^{-\frac{1}{2} (V(x)+1) } \, . 
\end{equation}
    \end{enumerate}
\end{lemma}
\begin{proof}
The inequalities~\eqref{eq:specgammaest} follow from the spectral representation $  \int (u - \zeta )^{-1}  \mu_{\lambda,\delta_x}^{(R)}(du)=	\Gamma_\lambda^{(R)}(x;\zeta)  $ and elementary inequalities for the integrand.

The second claim is based on the observation that $a -\lambda\leq  \inf \sigma(H_\lambda)  \leq  \inf \sigma(H_\lambda^{(R)})   $ for any $ R >0 $. We may hence differentiate for any $ \lambda \geq 0 $:
\begin{align}
	- \frac{d \Gamma_\lambda^{(R)}(x;a-\lambda)  }{d\lambda} \  \geq \   ( V(x) + 1 ) \; \Gamma_\lambda^{(R)}(x;a-\lambda)^2 \, . 
\end{align}
One of the last factors is estimated by $ \Gamma_\lambda^{(R)}(y,y;a-\lambda)^{-1} \leq \big\langle \delta_y , \, (H_\lambda^{(R)} + \lambda - a ) \, \delta_y \big\rangle  \leq 2 \sqrt{K} + 2 \lambda  -a  $.  Integrating the resulting inequality yields~\eqref{eq:speed}.

\end{proof}

In the following, we suppose $   E_\lambda + \Delta := \inf  \sigma(H_\lambda^{(R)})  > E >  E_\lambda $  such that
\begin{equation}
	 \xi_\lambda(E) \ := \ \frac{E - E_\lambda}{ \inf  \sigma(H_\lambda^{(R)})  -  E}  \ \in (0, \infty) \, .
\end{equation}
Then Lemma~\ref{lem:aux} and the factorization property~\eqref{factorization} of the Green function  imply for all $ x \in \mathcal{S}_R $:
\begin{align}
0\ \leq \ G_\lambda^{(R)}\big(0,x_-;E \big) \ & \leq \  \left(1+ \xi_\lambda(E)\right)^R\, G_\lambda^{(R)}(0,x_-;E_\lambda) \notag \\
					& \leq \ \frac{(1+ \xi_\lambda(E))^R}{K^{R/2}}  \, \left( 1 + \frac{\lambda}{2 \sqrt{K}} \right)^{-\frac{1}{2} \sigma(x)}  \, ,
\end{align}
where $ \sigma(x) :=  \sum_{0\preceq y \prec x} (V(y) + 1 ) \geq 0 $. To further estimate the right side, we will consider the event
\begin{equation}
	Z_0 := \left\{ \min_{x \in \mathcal{S}_R} \sigma(x)  \geq \frac{2\, \delta \log K +2\, \log(1+\xi_\lambda(E) )}{\log( 1 + \frac{\lambda}{2 \sqrt{K}} )} \, R \right\} \, ,
\end{equation}
with $ \delta > 0 $ from~\eqref{eq:pickdelta}.
This event is tailored such that $ G_\lambda^{(R)}\big(0,x_-;E \big)  \leq K^{-R (\delta + \frac{1}{2})} $ and hence
\begin{align}
	\mathbb{E}\left[ \big| S_\lambda^{(R)}(E)  \big|^{\frac{2+\delta}{2+2\delta}} \, \big| \, Z_0 \cap Z_1  \right] \ & \leq \ K^R\,  \E\left[ \left|  G_\lambda^{(R)}\big(0,x_-;E \big) \, G_\lambda(0,x ;E) \right|^{\frac{2+\delta}{2+2\delta}} \, \big| \, Z_0 \cap Z_1  \right] \notag \\
	& \leq \ C_\pm^2  \, K^{- \frac{\delta}{2} R }     \, , 
\end{align}
where the last inequality is based on~\eqref{eq:finitevolume} and the upper bound in~\eqref{eq:exclusion}. The constants $ C_+ , C_- $ depend 
(also through $ \delta $) on $ \lambda $. Chebychev's inequality hence leads to
\begin{equation}
 	\mathbb{P}\left( Z_2^c \,  \big| \, Z_0 \cap Z_1 \right) \ \leq \ C(\lambda) \, K^{- \frac{\delta^2}{4+4\delta} R}  
\end{equation}
with a finite constant $ C(\lambda) $ which only depends on $ \lambda $.  
For an estimate on the probability of the event $ Z_0 $ we use the following
\begin{lemma}\label{lem:probraybound}
For any $ 0 < \alpha \leq (8 \|\varrho\|_\infty K^2)^{-1}$:
\begin{equation}\label{eq:probraybound}
	\mathbb{P}\big(\min_{x \in \mathcal{S}_R} \sigma(x)  < \alpha R \big) 
	\leq \ K^R \, \left( 2 \sqrt{2 \|\varrho\|_\infty  \alpha} \right)^R \, . 
\end{equation}
\end{lemma}
\begin{proof}
Since there are $ K^R $ vertices with $\dist(0,x)=R$, it suffices for the proof of \eqref{eq:probraybound} to fix $ x $ and estimate 
\begin{equation}
	\mathbb{P}\big( \sigma(x)  <  \alpha R \big) \leq \left( e^{\alpha t} \, \mathbb{E}\left[ e^{-t (V(0) + 1) } \right] \right)^R \, ,
\end{equation}
for any $ t > 0 $, where we employed the help of a Chebychev inequality and the fact that the random variables $ \big( V(y) \big) $ are \emph{iid}.
Inserting indicator functions on the set $ \{ V(0) + 1 \geq 2 \alpha \} $ and its complement, we further bound
$	e^{\alpha t} \, \mathbb{E}\left[ e^{-t (V(0) + 1) } \right]  \leq e^{-t \alpha} + 2 \alpha \,  \|\varrho\|_\infty  \,  e^{t\alpha} $. Choosing $ t = - (2\alpha)^{-1} \log (4 \alpha \|\varrho\|_\infty)  \, > 0 \, $,  yields the result.
\end{proof}

We may now finally give a
\begin{proof}[Proof of Lemma~\ref{lem:propest}]
The choice of $ \delta $ in~\eqref{eq:pickdelta} and the condition~\eqref{eq:condE} together with Lemma~\ref{lem:probraybound} imply that $ \mathbb{P}\left(Z_0^c\right) \leq 2^{-R} $. 
We have thus established that
\begin{align}
	\mathbb{P}\left( Z_2^c \, \big| \,  Z_1  \right) \ & \leq \ \mathbb{P}\left( Z_2^c \, \big| \,  Z_0 \cap Z_1  \right) +  \mathbb{P}\left(Z_0^c\right) \notag \\
	& \leq \ C(\lambda) \, K^{- \frac{\delta^2}{4+4\delta} R}  + 2^{-R}  \, .  
\end{align}
\end{proof}

\appendix  
\noindent{\Large \bf  Appendix}

\section{Fractional-moment bounds}\label{app:L1corr}

The aim of this appendix is to present some basic weak-$L^1 $ bounds on Green functions of random operators, and related  fractional moment estimates.   Theorem~\ref{thm:twoL1}, which presents such bounds for pairs of Green functions,  is a  new result which  is needed here in the proof of our criteria, and which may also be of independent interest.  In the last subsection we discuss the related implications of the regularity Assumption~\ref{assD}.  \\  

The discussion in this appendix is carried  within the somewhat broader context of operators of the form: 
\be\label{eq:opappendix}
H_\lambda(\omega) = H_0 + \lambda\,  V(\omega)  \, , 
\ee
acting in  the Hilbert space $ \ell^2(\mathcal{G} )$, with 
 $\lambda \geq 0 $ the disorder-strength  parameter  and:
\begin{enumerate}[I]
\item\label{as1} $ \mathcal{G} $ the vertex set of some metric graph,
\item\label{as2} $  H_0$ a  self-adjoint operator  in $ \ell^2(\mathcal{G} )$, and 
\item\label{as3} $V(\omega) $ a random potential such that the random variables $ \{ V(x) \, | \, x \in \mathcal{G} \} $ are \emph{iid} with a probability distribution whose density  is (essentially) bounded, $ \varrho \in L^\infty(\R)$.
\end{enumerate}

\subsection{Weak-$L^1 $  bounds}  

We  recall that 
according to the Krein formula, the Green function of $ H_\lambda(\omega) $ restricted to the sites $ x , y  $ is in its dependence on $ V(x) $ and $V(y) $ of the form
\begin{equation}\label{eq:rank2Krein}
	\left(\begin{matrix} G_\lambda(x,x;\zeta) &  G_\lambda(x,y;\zeta) \\
						 G_\lambda(y,x;\zeta) &  G_\lambda(y,y;\zeta) \end{matrix}\right) = \left[\left(\begin{matrix}  \lambda \, V(x) & 0 \\ 0 & \lambda \, V(y) \end{matrix}\right)
						 	+ A_\lambda(\zeta)  \right]^{-1} \, , 
\end{equation}
where $ A_\lambda(\zeta) $ is given by the inverse of the left side for $ V(x) =V(y) = 0 $. In particular, $ G_\lambda(x,x;\zeta) =   (\lambda V(x)  - a )^{-1}  $  with some $ a \in \mathbb{C} $ which is independent of $ V(x) $.

The assumed boundedness of the density $ \varrho $ of the distribution of $ V(x) $ trivially implies bounds on probabilities of weak-$L^1$-type:
\begin{equation}\label{eq:elintwL}
\sup_{a\in\mathbb{C}} \ \int \indfct_{|v-a| < \frac{1}{t}} \varrho(v) \, dv \ \leq \ \frac{2 \|\varrho\|_\infty}{t} \, .
\end{equation} 
Since the dependence of the Green function $ G_\lambda(x,x;\zeta) $ on $ V(x) $ is of the above form, this implies that the following well-known weak-$L^1 $ bound, and hence the boundedness of fractional moments (cf.~\cite{AM}).
\begin{proposition}\label{prop:fracmom}
For a random operator $H_\lambda(\omega) = H_0 + \lambda\,  V(\omega)  $ on $ \ell^2(\mathcal{G}) $ satisfying assumptions~\ref{as1}--\ref{as3}, at any complex energy parameter $ \zeta \in \mathbb{C}^+ $ and for any $ t > 0 $ and $ s \in (0,1) $, the Green function  satisfies:
\begin{align}
	\label{eq:weakL1a}
 & \mathbb{P}\big( \left| G_\lambda(x,x;\zeta) \right| > t \, \big| \; \mathscr{A}_x \big) \ \leq \ \frac{2  \| \varrho \|_\infty}{\lambda\,  t} \, ,   \\
 &	\mathbb{E}\left[|G_\lambda(x,x;\zeta)|^s \, \big| \; \mathscr{A}_x\right] \ \leq \ \frac{2^s \| \varrho\|_\infty^s}{(1-s) \, \lambda^s} \, ,
 \label{eq:fracmonb}
\end{align}
where $ \mathscr{A}_x $  denotes the sigma-algebra generated by $ V(y) $ , $ y \neq x $. 
\end{proposition}
One trivial, but useful consequence of~\eqref{eq:weakL1a} is that for any $ p \in (0,1 ) $ and $ t \geq \frac{2 \| \varrho \|_\infty }{\lambda (1-p) } $: 
	\begin{equation}\label{eq:pickt}
	 \mathbb{P}\big( \left|G_\lambda(x,x;\zeta) \right| \leq t \, \big| \; \mathscr{A}_x \big) \ \geq \ p \, . 
	\end{equation}

Our  new result, which was vital in our second-moment analysis in Lemma~\ref{lem:upperbound1} and Theorem~\ref{thm:N2test}, concerns the joint conditional probability of events as in~\eqref{eq:weakL1a} associated with two (distinct) sites
\begin{theorem}\label{thm:twoL1}
In the situation of Proposition~\ref{prop:fracmom}, 
consider two sites $ x \neq y $ in a graph.
Then for any $ t > 0 $ and $ \zeta \in \mathbb{C}^+ $:
\begin{multline}\label{eq:twoL1}
	\mathbb{P}\Big(  \left| G_\lambda(x,x;\zeta) \right| > t \; \mbox{\rm and}\; \left| G_\lambda(y,y;\zeta) \right| > t  \; \big| \; \mathscr{A}_{xy} \Big) \\[.5ex] 
	\leq \ 
	\frac{ 2 \| \varrho \|_\infty}{\lambda^2 \, t} \; \min\left\{ 4 \|\varrho\|_\infty  \left(\sqrt{ \big| A_\lambda(x,y;\zeta) \big| \, \big|  A_\lambda(y,x;\zeta) \big|} + t^{-1} \right) \, , \, 1 \right\} \, , 
\end{multline}
where $ A_\lambda(x,y;\zeta) $ are the off-diagonal matrix elements of $ A_\lambda(\zeta) $ in~\eqref{eq:rank2Krein}, and $ \mathscr{A}_{xy} $ is the the sigma-algebra  generated by $ V(\xi) $, $ \xi \not\in \{x,y\} $.
 
\end{theorem}

In case of a tree graph, $\mathcal{G} = \mathcal{T} $, the off-diagonal matrix elements of $ A_\lambda(\zeta) $ simplify:
\begin{equation}\label{eq:ratioL1}
A_\lambda(x,y;\zeta) = \frac{G_\lambda(x,y;\zeta)}{G_\lambda(x,x;\zeta) \, G_\lambda(y,y;\zeta) - G_\lambda(x,y;\zeta) G_\lambda(y,x;\zeta)}  =  G_\lambda^{\mathcal{T}_{x,y}}(x_{-},y_{-};\zeta)  \, . 
\end{equation}
This is most easily proven by noting that the ratio does not depend on $ V(x) $ and $V(y) $ so that we may take them to infinity. In this
limit the ratio $$G_\lambda(x,y;\zeta) / [G_\lambda(x,x;\zeta) \, G_\lambda(y,y;\zeta) ] $$ tends to $G_\lambda^{\mathcal{T}_{x,y}}(x_{-},y_{-};\zeta)  $ and its numerator vanishes.

\begin{proof}[Proof of Theorem~\ref{thm:twoL1}]
Let $ A_\lambda(x,y;\zeta) $ denote the matrix elements of $ A_\lambda(\zeta) $ in the rank-two Krein formula~\eqref{eq:rank2Krein} and
abbreviate
 \begin{align} 
u \ & := \ \lambda V(x) + A_\lambda(x,x;\zeta)  \ \notag     \\ 
v \ & := \ \lambda V(y) +  A_\lambda(y,y;\zeta)   \  \notag    \, ,
\end{align}   
and $\alpha := A_\lambda(x,y;\zeta) $, $ \beta := A_\lambda(y,x;\zeta) $.
The lower bounds on $|G_\lambda(x,x;\zeta)$ and $|G_\lambda(y,y;\zeta)|$  translate to: 
 \begin{align}
\left|u - \frac{\alpha \beta}{v}\right| \ & \le \ \frac{1}{t} \label{u}   \\ 
\left|v - \frac{\alpha \beta}{u}\right| \ & \le \ \frac{1}{t}   \, .    \label{v}  
\end{align}
The claim will be proven on the basis of the following two observations:
\begin{enumerate} 
\item For any set of specified values of $\{\alpha,\beta, A(x,x;\zeta), A(y,y,;\zeta)\}$,  and of $v$, 
the set of $\Re u$ for which \eqref{u} holds is an interval of length 
at most $2 / t$, 
and a similar statement holds for  $v$ and $u$ interchanged and Eq.~\eqref{u} replaced by  \eqref{v}.   
\item  For any solution of \eqref{u} and \eqref{v}:
\be  \label{w-bound}
\min\{ |u|,|v|\} \ \le \ |\alpha| + t^{-1} \, . 
\ee 
\end{enumerate} 
The first statement is fairly obvious once one focuses on the condition 
 on the real part in \eqref{u}.  To prove the second assertion, let 
 \be 
w \ := \sqrt{|u|\cdot|v|}  \  \ge \min\{|u|,|v|\}
\ee 
Assuming \eqref{u} and \eqref{v} we have:  
\be 
|u| \, |v| \ - \ |\alpha|\, |\beta| \ \le |u\, v-\alpha \beta| \ \le \frac{\min\{ |u|,|v|\} }{t}  
\ \le \frac{\sqrt{|u| \, |v|} }{t}
\ee 
where the first relation is by the triangle inequality, and the second by \eqref{u} and \eqref{v}.  
Hence, under the assumed condition, the real quantity $w := |u| \, |v| $ satisfies:
\be
w^2 - |\alpha|\, |\beta|  \ \le \   \frac{w}{t}  \, . 
\ee 
Solving the quadratic equation we find: 
\be 
w \le \frac{1}{2t} + \sqrt{ \frac{1}{(2t)^2}   +|\alpha|\, |\beta|   }  \ \le \   
\frac{1}{2t} +  \left(\frac{1}{2t}  + \sqrt{|\alpha|\, |\beta|}  \right) \, ,
\ee 
which implies \eqref{w-bound}.

To bound the probability in \eqref{eq:twoL1}, let us consider the set 
of values of $V(x)$ and $V(y)$ for which the event occurs, at specified values of 
the $2\times 2$ matrix $A_\lambda(\zeta)$.  
Let $S\subset \R^2$ be the corresponding range of values of  $\{\Re u, \Re v\}$.  
Then by {\it 2.}, $S$ is contained within the union of two strips,  
one parallel to the $\Re v$ axis and the other parallel to the $\Re u$ axis.  To bound the 
measure of its intersection with the first one,  we note that the relevant values 
of $\Re u$ are contained in an interval of length at most $ 2\left(\frac{1}{t}  + \sqrt{|\alpha|\, |\beta|}\right)$, and 
for each value of $u$ the range of values of $\Re v$ is of Lebesgue measure 
not exceeding $2/t$ (by {\it 2.}).   Hence the measure of the intersection 
of $S$ with this strip is at most $  \frac{4}{t} \left(\frac{1}{t}  +\sqrt{|\alpha|\, |\beta|}\right)$, 
and a similar bound applies to the intersection of $S$ with the second one.  
Adding the two, one gets the bound claimed in~\eqref{eq:twoL1}.

\end{proof}

\subsection{The regularity assumption~\ref{assD}}

The class of probability densities satisfying Assumption~\ref{assD} (see Eq.~\eqref{eq:regrho}) includes those $ \varrho $ which have a single hump. More precisely, suppose there is some $ m \in \R $ such that 
$\varrho$ is monotone increasing for $ v< m $ and monotone decreasing for $ v > m $.
  If one picks $ \nu_0 > 0 $ such that $ \varrho(m) / \min\{ \varrho(m-\nu_0)\, , \, \varrho(m+\nu_0) \}=: c_0< \infty $, then~\eqref{eq:regrho} is satisfied for all $ v \in \R$ and $ c = 2 \max\{1,c_0/\nu_0\} $
  Examples of single-hump probability densities are Gaussian and the Cauchy densities. 
Similarly as above one sees that any finite linear combination of single-hump functions also lead to probability densities which satisfy~\eqref{eq:regrho}. \\

Our next goal is to  illuminate some of the consequences of~\eqref{eq:regrho}. 
Clearly, if $ \varrho $ satisfies~\eqref{eq:regrho}, then $ \varrho \in L^\infty(\R)$ and~\eqref{eq:elintwL} applies.
In fact, the assumption is tailored to provide the following extension of~\eqref{eq:elintwL}.
\begin{lemma}
	If $ \varrho \geq 0 $ satisfies~\eqref{eq:regrho} (with constant $ c > 0 $), then for any $ s \in (0,1) $, $ a \in \mathbb{C} $ and $ t  \geq 1 $:
	\begin{equation}
		\int \indfct_{|v-a| < \frac{1}{t}} \ \frac{\varrho(v) \, dv}{|v-a|^s}  \ \leq \ \frac{c}{(1-s) \, t^{1-s} }\,   \int\frac{\varrho(v) \, dv}{|v-a|^s} \, . 
	\end{equation}
\end{lemma}
\begin{proof}
We start by estimating the left side
\begin{align}
	& \int \indfct_{| v - a| < \frac{1}{t}}  \frac{\varrho(v) \, dv }{|v - a |^s} \  \leq  \sup_{|v - a| \leq \frac{1}{t}} \varrho(v) \; \int\indfct_{| v - a| < \frac{1}{t}}  \frac{ dv }{| v - a |^s} \ =  \ \frac{2}{(1-s)  \, t^{1-s}}\; \sup_{| v - a| < \frac{1}{t}} \varrho(v)  \, . 
\end{align}
Using~\eqref{eq:regrho} we then conclude that the last factor in the right side is bounded from below by
\begin{align}
	\int\frac{\varrho(v) \, dv}{|v-a|^s} \ \geq \ \int \indfct_{| v - a| \leq  1 } \varrho(v) \, dv \ \geq \ \frac{ 2 }{ c } \, \sup_{| v - a| \leq 1 } \varrho(v) \, . 
	\end{align}
The above two estimates imply the assertion. 
\end{proof}

In view of~\eqref{eq:rank2Krein} this lemma bears the following consequences for weighted averages of the following type:
\begin{equation}
	\mathbb{E}_{s}^{(x,y)}\left[ Q  \right] \ := \  \frac{\mathbb{E}\left[ |G_\lambda(x,y;\zeta)|^s \, Q \right] }{\mathbb{E}\left[ |G_\lambda(x,y;\zeta)|^s \right]} \, , 
\end{equation}
where $ x, y \in \mathcal{G} $, $ \zeta \in \mathbb{C}^+ $ and $ s \in (0,1) $. We denote by $\mathbb{P}_{s}^{(x,y)} $ the corresponding probability measure. 

\begin{proposition}
In the situation of Proposition~\ref{prop:fracmom}, assume additionally that $ \varrho $ satisfies~\eqref{eq:regrho} (with constant $  c > 0 $).
Then, at any complex energy parameter $ \zeta \in \mathbb{C}^+ $ and for any $ s \in (0,1) $ and $ t \geq \lambda^{-1} $, the Green function  satisfies:
\begin{equation}\label{eq:wL1weighted}
	\mathbb{P}_{s}^{(x,y)}\left( |G_\lambda(x,x;\zeta)| > t \, | \, \mathscr{A}_x \right) \ \leq \ \frac{ c}{(1-s) \, (\lambda t)^{1-s} } \, ,
\end{equation}
where $ \mathscr{A}_x $  denotes the sigma-algebra generated by $ V(y) $ , $ y \neq x $. 
\end{proposition}
Analogously to ~\eqref{eq:pickt}, we conclude from~\eqref{eq:wL1weighted} that for any $ p \in (0,1 ) $ and  all $ t \geq  \lambda^{-1} (c/[(1-s)(1-p)]^{1/(1-s)} $:
	\begin{equation}\label{eq:pickt2}
	\mathbb{P}_{s}^{(x,y)}\big( \left| G_\lambda(x,x;\zeta) \right| \leq t \, \big| \; \mathscr{A}_x \big) \ \geq \ p \, , 
	\end{equation}	
	uniformly in $ y \in \G $, the choice of the graph $ \G $ and $ \zeta \in \C^+ $.

\section{A large deviation principle for triangular arrays}  \label{App:Ldevs}

In our analysis of the Green function's large deviations we make use of a large deviation principle.  The statement and its proof are similar to large deviation theorems which are familiar in statistical mechanics and probability theory~\cite{DemZeit,DS,Elis}. However since a close enough reference could not be located we enclose the proof here.

\subsection{A general large deviation theorem}

The following theorem should be regarded as  a stand-alone statement.  It is intended to be read disregarding fact that the symbols which appear there  ($\Gamma$ and $\eta$ ) were assigned a specific meaning elsewhere in the paper.    The similarity does however indicate the application of this theory to the main discussion of this  work.

\begin{theorem} \label{thm:ld} 
Let $\{ \Gamma_j^{(N)}(\eta)\}_{j=1}^{N} $ with $N\in \N $,  be a family of a triangular arrays of random variables indexed by $\eta \ge 0$, satisfying the following two conditions, at some $ r_1 < r_2 $ and $C<\infty$:  
\begin{enumerate} 
\item[a.]  The functions 
\be \label{def:psiN}
\vP_N(t;\eta)\  :=\  \frac{1}{N} \log \Ev{ \prod_{j=1}^{N} |\Gamma_j^{(N)}(\eta)|^t} 
\ee 
converge pointwise in $[r_1,r_2] \subset (-1,1)$: 
\be  \label{Psi_conv}
\vP(t)\  :=\   \lim_{\substack{N\to \infty \\[.5ex] \eta \downarrow 0 }}  \vP_N(t;\eta)   \, . 
\ee 
\item[b.] 
For all $1 \le  k < N$,   and $ t_1, t_2  \in [r_1,r_2]$  
\begin{multline}   \label{eq:submult}
\Ev{\prod_{i=1}^k |\Gamma^{(N)}_i(\eta)|^{t_1} \,  \prod_{j=k+1}^N |\Gamma^{(N)}_j(\eta)|^{t_2} }    \\  \le \    C \, e^{(N-k) [\vP_N(t_1,\eta) -  \vP_N(t_2,\eta)]} \  \Ev{\prod_{i=1}^N |\Gamma^{(N)}_i(\eta)|^{t_2} }   \, .  
\end{multline} 
\end{enumerate} 

Then for every $\gamma$ which coincides with $-\varPsi' (s)$ at a point $ s \equiv s(\gamma)\in (r_1,r_2)$ where the function $\varPsi(s)$ is differentiable,  and for any $\varepsilon >0$, 
there are $\widehat N\equiv \widehat N(\varepsilon,\gamma) < \infty$  and $\hat {\eta} \equiv \hat \eta(\varepsilon,\gamma) >0$ such that for all $N\ge \widehat{N}$ and $0< \eta < \hat {\eta}$ the following estimates hold:
\begin{enumerate} 
\item Given the rate function $ I(\gamma) \ := \ - \inf_{ t\in[r_1,r_2]} \left[ \varPsi(t) + t \gamma  \right]$ one has:
\ \begin{align}    \label{eq:ld_upper}
 \Pr{ \prod_{j=1}^N |\Gamma^{(N)}_j(\eta)| \ge e^{-(\gamma + \varepsilon) N} }   \  \le  \ e^{-I(\gamma) N} \, e^{2 \varepsilon  N}  
  \end{align} 
\item With respect to the $s$-tilted probability average defined by
\be \label{eq:defPs}
 \Prs{ Q} \  = 
\ \frac{\Ev{ I_Q\times \prod_{j=1}^N |\Gamma^{(N)}_j(\eta)|^s } }{\Ev{\prod_{j=1}^N |\Gamma^{(N)}_j(\eta)|^s} }   \, , 
\ee 
for  any $\ell \in \{ 0, \dots , N\}$:
\begin{align}  \label{eq:ld_glb} 
 \mathbb{P}_{s}\left(  \prod_{j=1}^\ell |\Gamma^{(N)}_{j}(\eta)|     \ge \   e^ { -(\gamma- \varepsilon) \ell}   \right)  \  &\le \ C  \, e^{-\kappa(\varepsilon,\gamma) \ell/3}  
  \\
\mathbb{P}_{s}\left(  \prod_{j=1}^\ell |\Gamma^{(N)}_{j}(\eta)|     \le \   e^ { -(\gamma+ \varepsilon) \ell}  \right)  \  & \le  \ C  \, e^{-\kappa(\varepsilon,\gamma) \ell/3} 
\label{eq:ld_glb2} 
\end{align} 
where  $ \kappa(\varepsilon,\gamma)  \  := \ \min\left\{ \kappa_-(\varepsilon,\gamma) \, , \,  \kappa_+(\varepsilon,\gamma) \right\} >0 $ and 
\begin{align}\label{eq:defkappa}
	   \kappa_\pm(\varepsilon,\gamma) \ & := \ \sup_{\substack{{\rm sgn} \, \Delta = \pm \\ r_1 < s+\Delta < r_2 }}  \left[ \varPsi(s) +  ( \varPsi'(s) \pm\varepsilon)\,  |\Delta|  - \varPsi(s+\Delta)\right]  \, .
\end{align}   
\item  
For any event $Q $:
\begin{align}    \label{eq:ld_lower}
\Pr{ Q} \ \ge \   \ e^{-I(\gamma) N} \, e^{-2 \varepsilon N} \, 
\left[   {\mathbb P}_{s}\left(  Q\right) - C \, e^{-\kappa(\varepsilon,\gamma) N/3}  
  \right] 
 \end{align} 
\end{enumerate} 
\end{theorem}
Several remarks apply:
\begin{enumerate}
\item
The function $\varPsi$ is convex, assuming the limit \eqref{Psi_conv} exists, and therefore the above value of $I(\gamma)$ can also be presented as  
\begin{equation}
	I(\gamma)  =  - \left[\varPsi(s) + \gamma s \right] \, .
\end{equation}     
The error margins $ \kappa_\pm(\varepsilon,\gamma) $ defined in~\eqref{eq:defkappa} are strictly positive for any $ \varepsilon > 0$ due to convexity of $\varPsi$. 
\item The proof of Theorem~\ref{thm:ld} follows a standard procedure for such bounds: what is a large deviation for the value of $\frac{1}{N} \sum_{j=1}^{N} \log \Gamma^{(N)}_j$ with respect the the initial probability measure becomes a regular occurrence once the measure is suitably tilted, i.e. modified by the factor $\prod_{j=1}^N |\Gamma^{(N)}_j|^s$ at suitable $s$.    The statement is then derived by relating the original and the tilted probabilities.  In Theorem~\ref{thm:ld} we add to this standard procedure the observation that under the condition \eqref{eq:submult} the global tilt of the measure shifts the typical values of the sample mean of $\log \Gamma_j$  for all the partial sums, to values in the vicinity of $(-\gamma)$.  
\end{enumerate}

In the proof we make use of the following  fact on convergence of convex functions. 
\begin{lemma} \label{lem:uniformest}
Under the condition  \eqref{Psi_conv}, one has the uniform convergence:
\be \label{eq:uniformest}
\lim_{\substack{N\to \infty \\[.5ex] \eta \downarrow 0 }} \sup_{s\in [r_1,r_2]} \, |\varPsi _N (s;\eta) - \varPsi (s) | \, = 0  \, .    
\ee 
\end{lemma}
\begin{proof}
This follows from the fact that 
if a family of convex functions converges pointwise over an open interval, then its convergence is uniform on compact subsets, cf. \cite{Simon_convex}. 
\end{proof}

\begin{proof} [Proof of Theorem~\ref{thm:ld}]  Since the superscript of $\Gamma^{(N)}_j$ 
 is somewhat redundant  it will be occasionally omitted  (it takes a common value for all terms within each statement).    
 
   We will choose $\widehat{N} \equiv\widehat{N}(\varepsilon,\gamma) <  \infty $ and $ \hat {\eta} \equiv  \hat {\eta}(\varepsilon,\gamma) > 0$ using Lemma~\ref{lem:uniformest} such that for all $N\ge \widehat{N}(\varepsilon,\gamma)$ and $0< \eta < \hat {\eta}(\varepsilon,\gamma)$:
  \begin{equation}\label{eq:uniformest2}
  	R_N(\eta) \ := \ \sup_{s\in [r_1,r_2]} \, |\varPsi _N (s;\eta) - \varPsi (s) | \, < \,  \min\left\{\varepsilon \, , \, \tfrac{1}{3} \, \kappa(\varepsilon,\gamma)\right\} \, ,
\end{equation}

The proof of~\eqref{eq:ld_upper} relies on an elementary Chebychev estimate with $ s \in(r_1,r_2)$:
	\begin{align}    \label{eq:ld_final2}
& \Pr{ \prod_{j=1}^N |\Gamma_j(\eta) | \ge e^{-(\gamma + \varepsilon) N} }   \   \le  \ e^{N \left[ s (\gamma + \varepsilon) + \varPsi _N(s;\eta)\right]} \notag \\
& = e^{\varepsilon s N } \, e^{-N I(\gamma) } \, e^{N\left[\varPsi _N(s;\eta) - \varPsi(s)\right]} 
\ \leq \ e^{2 \varepsilon  N } \, e^{-N I(\gamma) } 
 \end{align} 
for any $N\ge \widehat{N}$ and $0< \eta < \hat {\eta}$ by \eqref{eq:uniformest2}.

For a proof of~\eqref{eq:ld_glb} we again employ the Chebychev inequality and~ \eqref{eq:submult}  to conclude for any $ \Delta $ such that $ s + \Delta \in (r_1,r_2) $:
\begin{align} \label{eq:proofld_glb1}
  \mathbb{P}_{s}\left(  \prod_{j=1}^\ell |\Gamma_{j}(\eta)|    \ge \   e^ { -(\gamma- \varepsilon) \ell}  \right)   & \le     \ 
\E_s\Big[  \prod_{j=1}^\ell |\Gamma_j(\eta)|^{\Delta } \Big] \,  e^{ \Delta  (\gamma - \varepsilon) \ell}  \notag \\
&  \le \ C\;  e^{ \left[\varPsi _N(s+\Delta;\eta) - \varPsi _N(s;\eta)\right] \, \ell}\, e^{ \Delta (\gamma-\varepsilon)\ell} 
\end{align}
Infimizing over $ \Delta  $, we hence conclude that the  left side in~\eqref{eq:proofld_glb1} is bounded by
\begin{equation}
  C \ e^{- \kappa_+(\varepsilon,\gamma)\, \ell} \, e^{2 \ell \, R_N(\eta) } \ \leq \ C \, e^{- \kappa_+(\varepsilon,\gamma) \, \ell / 3 } 
\end{equation}
 for any $N\ge \widehat{N}$ and $0< \eta < \hat {\eta}$ by~\eqref{eq:uniformest2}. 

The proof of~\eqref{eq:ld_glb2}  proceeds similarly. It starts from the observation that 
\begin{align} \label{eq:proofld_glb2}
\mathbb{P}_{s}\left(  \prod_{j=1}^\ell |\Gamma_{j}(\eta)|     \le  e^ { -(\gamma+ \varepsilon) \ell}  \right)   & \leq
\E_s\Big[  \prod_{j=\ell+1}^N |\Gamma_j(\eta)|^{-\Delta } \Big]  e^{- \Delta  (\gamma + \varepsilon) \ell} \notag \\
& \leq \, C \,  e^{ \left[\varPsi _N(s-\Delta;\eta) - \varPsi _N(s;\eta)\right] \, \ell}\, e^{- \Delta (\gamma+\varepsilon)\ell} 
\end{align}
for any $ \Delta $ such that $ s - \Delta \in (r_1,r_2) $. Infimizing over this parameter, we hence conclude that the left side in~\eqref{eq:proofld_glb2} is bounded by $ C \ e^{- \kappa_-(\varepsilon,\gamma)\, \ell} \, e^{2 \ell \, R_N(\eta) }  \leq  C \  e^{- \kappa_- (\varepsilon,\gamma) \, \ell / 3 } $ by~\eqref{eq:uniformest2}.\\

For a proof of~\eqref{eq:ld_lower} we estimate the regular probability of in terms of the one defined via the tilted measure:
\begin{align}
	\mathbb{P}(Q) \ &  \geq \ e^{N \varPsi _N(s;\eta)} \, e^{s (\gamma-\varepsilon) N } \; \mathbb{P}_s\left( Q \; \mbox{and} \;  \prod_{j=1}^N |\Gamma_j(\eta) | \le e^{-(\gamma - \varepsilon) N} \right) \notag \\ 
	& \geq   \ e^{N \varPsi _N(s;\eta)} \, e^{s (\gamma-\varepsilon) N } \, \left(
 \mathbb{P}_s\left( Q \right) -  \mathbb{P}_s\left( \prod_{j=1}^N |\Gamma_j(\eta) | \ge e^{-(\gamma - \varepsilon) N} \right) \right)\, . 
\end{align} 
The first terms are estimated from below similarly as in~\eqref{eq:ld_final2} by $ e^{-I(\gamma) } e^{-2 \varepsilon N } $. 
The second term in the bracket is bounded by $ C e^{- \kappa(\varepsilon,\gamma) \, N / 3 } $  for any $N\ge \widehat{N}$ and $0< \eta < \hat {\eta}$ according to~~\eqref{eq:ld_glb}. 
\end{proof}

\subsection{Applications to Green function's large deviations}

The aim of this subsection is to establish  the two main  large-deviation statements which are used in this paper, which were asserted in Theorems~\ref{thm:regGbehav} and~\ref{cor:ldev1}.  We start with the latter.


\begin{proof}[Proof of Theorem~\ref{cor:ldev1}] We first check the applicability of Theorem~\ref{thm:ld}.
By construction, the variables $ \{ \Gamma_\pm(j;\eta)\}_{ j = 1 }^{ N_\kappa} $, which were defined in~\eqref{eq:deftriangulara},   
are two families of triangular arrays. They satisfy the consistency condition~\eqref{eq:consist}. As a consequence,  the quantity defined in~\eqref{def:psiN} agrees for both cases:
\begin{equation}
 \varPsi_{N_\kappa}(s;\eta) \ = \ \frac{1}{N_{\kappa}} \log\,  \E\left[\left| G^{ \widehat{\mathcal{T}}_x }(x_{n_\kappa}, x_{N-1} ; E +i\eta)\right|^s \right] \, . 
 \end{equation}
Lemma~\ref{lem:GG2}  and Theorem~\ref{thm:phi} imply that for any $ t \in (-\varsigma,1) $:
\begin{equation}~\label{eq:limpsiphi}
	\varphi(t;E) \ \equiv \ \varphi(t)\ = \ \lim_{\substack{N_{\kappa} \to \infty\\ \eta\downarrow 0}} \,  \varPsi_{N_\kappa}(t;\eta) \,  . 
\end{equation}
Moreover, these bound ensure the validity of~\eqref{eq:submult} with $ r_1 = -\varsigma $ and arbitrary $ r_2 \in (0, 1) $. For a proof of this assertion, one integrates out the random variable associated with the first vertex on which $ t_2 $ occurs, cf.~\eqref{eq:surface}. \\

 The upper bound~\eqref{eq:Llbound2} is hence a consequence of~\eqref{eq:ld_upper}.
For a proof of the lower bound~\eqref{eq:Llbound1} we employ~\eqref{eq:ld_lower}. We first note that the choice of $ b $ is tailored to ensure 
$
\mathbb{P}_s\left( L^{(\rm bc)}_x \right) \geq \tfrac{7}{8} $. Furthermore, using~\eqref{eq:ld_glb}  and~\eqref{eq:ld_glb2} we conclude that there are $ \widehat{N} \equiv  \widehat{N} (\epsilon,\gamma) $ and $  \widehat{\eta}  \equiv \widehat{\eta} (\epsilon, \gamma) $ such that for all  $ N_\kappa \geq \widehat{N} $ and $ \eta \in (0,  \widehat{\eta} ) $:
\begin{align}
&	1 -  \mathbb{P}_s\Big( \bigcap_{k=\tfrac{1}{2} n_\kappa}^{N_\kappa}  L^{(k,\pm)}_x(\eta;\epsilon) \Big) \ \notag \\ 
&	\leq \ \sum_{k=\tfrac{1}{2} n_\kappa}^{N_\kappa}  \left[  \mathbb{P}_{s}\Big( \prod_{j= 1}^k |  \Gamma_\pm(j;\eta)|   \ge \   e^ { -(\gamma- \varepsilon) \ell} \Big) + 
	  \mathbb{P}_{s}\Big(   \prod_{j= 1}^k |  \Gamma_\pm(j;\eta)|   \le \   e^ { -(\gamma+\varepsilon) \ell}   \Big) \right] \notag \\
& \leq \ 2 \, C \, \sum_{k=\tfrac{1}{2} n_\kappa}^{N_\kappa}  e^{-\kappa(\varepsilon,\gamma) k /3}  \ \leq \ \frac{6 \, C }{\kappa(\varepsilon,\gamma)} \, e^{-\kappa(\varepsilon,\gamma) n_\kappa /6} \, . 
\end{align}
By choosing $ n_\kappa $ sufficiently large, this term can be made arbitrarily small since $ \kappa(\varepsilon,\gamma) > 0 $. As a consequence, we conclude that there is some 
$ n_0 $ and $ \eta_0 $ such that for all $ |x| \geq n_0 $ and $ \eta \in (0,\eta_0) $:
\begin{equation}
	\mathbb{P}_s\left( L_x(\eta;\epsilon)\right) \geq \tfrac{1}{2} \, . 
\end{equation}
Using this estimate in~\eqref{eq:ld_lower}  concludes the proof of~\eqref{eq:Llbound2}, since the second term in~\eqref{eq:ld_lower}  is seen to be arbitrarily small for $ n $ large enough and any factor may be absorbed for sufficiently large $ N_\kappa$ by decreasing the prefactor $ e^{-N_k (I(\gamma) + 2 \epsilon)} $ in~\eqref{eq:ld_lower}. 
\end{proof}  


\begin{proof}[Proof of Theorem~\ref{thm:regGbehav}]
 In a similar way as in the proof of Theorem~\ref{cor:ldev1}, the assertion follows from Theorem~\ref{thm:ld} in the special case of $ s = 0 $. This choice is admissible since, according to~\eqref{eq:derphi}, 
the free energy function $ \varphi(s;E) $, which emerges in the limit~\eqref{eq:limpsiphi}, is differentiable at $ s = 0 $ with derivative given by the negative Lyapunov exponent. 
\end{proof}

\bigskip
\footnotesize
\noindent\textit{Acknowledgments.}
It is a pleasure to thank Mira Shamis for bringing Proposition~\ref{prop:shamis} to our attention.    
We thank the  Departments of Physics and Mathematics at the  Weizmann Institute of Science for hospitality at visits during which some of the work was done.   
This research was supported in part by NSF grants PHY-1104596 and DMS-0602360 (MA), DMS-0701181 (SW), and a Sloan Fellowship~(SW).

\end{document}